\documentclass[11pt]{article}

\usepackage[T1]{fontenc}

\usepackage[margin=1in]{geometry}

\usepackage{wrapfig}

\usepackage{tcolorbox}

\usepackage{authblk}

\usepackage{hyperref}
\usepackage{xcolor}
\hypersetup{
    colorlinks,
    linkcolor={red!50!black},
    citecolor={blue!50!black},
    urlcolor={blue!80!black}
}

\usepackage{amsmath}

\usepackage{relsize}

\usepackage{pgfplots}
\usepackage{pgfplotstable}
\usetikzlibrary{matrix,arrows,shapes,positioning,calc,snakes,decorations.markings}
\usetikzlibrary{fadings,shapes.arrows,shadows}   
\usetikzlibrary{arrows.meta,patterns}

\usetikzlibrary{arrows.meta}
\tikzset{>={Latex[width=1.25mm,length=1.25mm]}}

\pgfdeclarelayer{background}
\pgfsetlayers{background,main}

\definecolor{lightgray}{gray}{0.7}
\definecolor{medgray}{gray}{0.55}
\definecolor{darkgray}{gray}{0.4}

\definecolor{pastelblue}{rgb}{0.6, 0.8, 0.91}
\definecolor{pastelred}{rgb}{1.0, 0.41, 0.38}

\definecolor{darkpastelgreen}{rgb}{0.01, 0.75, 0.24}
\definecolor{pastelyellow}{rgb}{0.01, 0.75, 0.24}

\definecolor{shade1}{rgb}{0.87, 0.91, 0.92} 
\definecolor{shade2}{rgb}{1, 1, 0.88} 
\definecolor{shade3}{rgb}{0.69, 0.94, 0.92} 
\definecolor{shade4}{rgb}{1, 0.79, 0.78} 

\definecolor{darkpastelblue}{rgb}{0.27, 0.47, 0.70}
\definecolor{darkpastelred}{rgb}{0.59, 0.18, 0.10}
\definecolor{darkpastelbrown}{rgb}{0.39, 0.31, 0.25}
\definecolor{darkpastelmagenta}{rgb}{0.59, 0.44, 0.84}

\tikzset{fontscale/.style = {font=\relsize{#1}}}

\newcommand{\mysquare}[3]{
\node [rectangle, fill=#3,minimum size=4.75mm] () at (#1,#2) {};
}

\newcommand{\mytwocoloursquare}[2]{
	\fill [pastelred]  (#1-0.5, #2-0.5) -- (#1-0.5, #2+0.5) -- (#1+0.5,#2+0.5) -- cycle;
	\fill [pastelblue] (#1-0.5, #2-0.5) -- (#1+0.5, #2-0.5) -- (#1+0.5,#2+0.5) -- cycle;
	\draw[-] (#1-0.5,#2-0.5) -- (#1+0.5,#2+0.5);
	\node [rectangle, draw=black, minimum size=4.75mm] () at (#1,#2) {};
}

\newcommand{\myuparrow}[3]{
	\mysquare{#1}{#2}{lightgray}
	\draw[->, white] (#1,#2-0.25) -- ($(#1,#2-0.25)!0.5cm!(#1,#2+1)$);
}
\newcommand{\mydnarrow}[3]{
	\draw[->, darkgray] (#1,#2+0.25) -- ($(#1,#2+0.25)!0.5cm!(#1,#2-1)$);
}

\newcommand{\myrightarrow}[3]{
	\mysquare{#1}{#2}{lightgray}
	\draw[->, white] (#1-0.25,#2) -- ($(#1-0.25,#2)!0.5cm!(#1+1,#2)$);
}
\newcommand{\myleftarrow}[3]{
	\draw[->, darkgray] (#1+0.25,#2) -- ($(#1+0.25,#2)!0.5cm!(#1-1,#2)$);
}

\newcommand{\zeroS}[3]{
	\node[white, fill=#3, minimum size=4.75mm] at (#1,#2) {$0$};
}

\newcommand{\updncol}[2]{ 
	\pgfmathsetmacro{\limbelow}{#2 - 1}
	\pgfmathsetmacro{\limabove}{#2 + 1}

	\foreach \j in {1,...,\limbelow}
		\myuparrow{#1}{\j}{black};
	\zeroS{#1}{#2}{pastelblue};
	\foreach \j in {\limabove,...,8}
		\mydnarrow{#1}{\j}{black};
}

\newcommand{\leftrightcol}[2]{ 
	\pgfmathsetmacro{\limleft}{#1 - 1}
	\pgfmathsetmacro{\limright}{#1 + 1}

	\foreach \i in {1,...,\limleft}
		\myrightarrow{\i}{#2}{black};
	\zeroS{#1}{#2}{pastelred};
	\foreach \i in {\limright,...,10}
		\myleftarrow{\i}{#2}{black};
}

\newcommand{\patharrow}[4]{
\draw[->, line width=0.2mm] (#1,#2) -- (#3,#4);
}

\usepackage[utf8]{inputenc}
\usepackage{enumitem}

\usepackage{amsthm}
\usepackage{amsmath}

\newtheorem{conjecture}{Conjecture}
\newtheorem{theorem}{Theorem}
\newtheorem{lemma}[theorem]{Lemma}
\newtheorem{definition}[theorem]{Definition}
\newtheorem{observation}[theorem]{Observation}
\newtheorem{corollary}[theorem]{Corollary}

\usepackage{cite}

\usepackage[disable]{todonotes}

\newcommand{\spencer}[1]{\todo[inline,color=blue!30, caption={}]{Spencer: #1 }}

\usepackage{microtype}

\title{\huge Unique End of Potential Line}

\author[1]{John Fearnley}
\author[2]{Spencer Gordon}
\author[3]{Ruta Mehta}
\author[1]{Rahul Savani}

\affil[1]{University of Liverpool\\
  \texttt{\{john.fearnley, rahul.savani\}@liverpool.ac.uk}}
\affil[2]{California Institute of Technology\\
\texttt{slgordon@caltech.edu}}
\affil[3]{University of Illinois at Urbana-Champaign\\
  \texttt{rutamehta@cs.illinois.edu}}

\usepackage{tabularx,amssymb,nicefrac,amsmath,multirow,stmaryrd}
\usepackage{comment,amsfonts} 
\usepackage{epsfig} \usepackage{latexsym,nicefrac,bbm}
\usepackage{xspace}
\usepackage{color,fancybox,graphicx,url}
\usepackage{tabularx} 
\usepackage{booktabs}
\usepackage{color,colortbl}
\usepackage{mathtools}

\newcommand{\tf}{{\tilde{f}}}
\newcommand{\mz}{{\times\zeta}}
\newcommand{\size}{{\mbox{size}}}

\renewcommand{\ss}{s} 
\renewcommand{\tt}{t} 
\newcommand{\yy}{{\mathbf y}}

\newcommand{\ww}{{\mathbf w}}
\newcommand{\bb}{{\mathbf b}}

\usepackage{algorithmicx}
\usepackage{algorithm}
\usepackage{algpseudocode}

\newcommand{\ra}{\rightarrow}
\long\def\symbolfootnote[#1]#2{\begingroup%
\def\thefootnote{\fnsymbol{footnote}}\footnote[#1]{#2}\endgroup}

\DeclareMathOperator{\Slice}{\mathsf{Slice}}

\DeclareMathOperator{\firsthalf}{FirstHalf}
\DeclareMathOperator{\secondhalf}{SecondHalf}

\def\cc#1{\mathsf{#1}}
\def\CLS{\ensuremath{\cc{CLS}}\xspace}

\def\FNP{\ensuremath{\cc{FNP}}\xspace}
\def\FP{\ensuremath{\cc{FP}}\xspace}
\def\NP{\ensuremath{\cc{NP}}\xspace}
\def\coNP{\ensuremath{\cc{coNP}}\xspace}

\def\TFNP{\ensuremath{\cc{TFNP}}\xspace}
\def\PPAD{\ensuremath{\cc{PPAD}}\xspace}
\def\PLS{\ensuremath{\cc{PLS}}\xspace}
\def\PPP{\ensuremath{\cc{PPP}}\xspace}
\def\PPA{\ensuremath{\cc{PPA}}\xspace}
\def\PSPACE{\ensuremath{\cc{PSPACE}}\xspace}
\def\PPADPLS{\ensuremath{\cc{PPAD} \cap \cc{PLS}}\xspace}
\def\EOPLc{\ensuremath{\cc{EOPL}}\xspace}
\def\UEOPLc{\UniqueEOPLc}
\def\PUEOPLc{\ensuremath{\cc{PromiseUEOPL}}\xspace}

\def\UniqueEOPLc{\ensuremath{\cc{UniqueEOPL}}\xspace}

\newcommand{\LinearFIXP}{\ensuremath{\cc{LinearFIXP}}\xspace}

\def\problem#1{{\scshape #1}}
\def\CM{\problem{Contraction}\xspace}
\def\LCM{\problem{PL-Contraction}\xspace}
\def\DCM{\problem{OPDC}\xspace}
\def\OPDC{\DCM}

\def\USO{\problem{Unique-Sink-Orientation}\xspace}
\def\GUSO{\problem{Grid-USO}\xspace}

\def\EOL{\problem{EndOfLine}\xspace}
\def\SOD{\problem{SinkOfDag}\xspace}
\def\EOPL{\problem{EndOfPotentialLine}\xspace}
\def\UEOPL{\problem{UniqueEOPL}\xspace}
\def\PUEOPL{\problem{PromiseUniqueEOPL}\xspace}
\def\UFEOPL{\problem{UniqueForwardEOPL}\xspace}
\def\UFEOPLp1{\problem{UniqueForwardEOPL+1}\xspace}
\def\EOML{\problem{EndOfMeteredLine}\xspace}
\def\SOVL{\problem{SinkOfVerifiableLine}\xspace}

\def\PLCP{\problem{P-LCP}\xspace}
\def\ContractionMap{\problem{ContractionMap}\xspace}

\def\eps{\varepsilon}

\def\ite{\mbox{ItoE}}
\def\eti{\mbox{EtoI}}
\def\pot{\mbox{$V$}}
\def\isvalid{\mbox{IsValid}}
\def\PLo{\solnref{Q1}\xspace}
\def\PLt{\solnref{Q2}\xspace}

\def\Real{\mathbb{R}}

\def\Natural{\mathbb{N}}

\let\N\Natural
\let\Q\Rational

\let\Z\Integer

\def\poly{\operatorname{poly}}

\def\Ceil#1{\left\lceil #1 \right\rceil}

\def\Set#1{\left\{ #1 \right\}}
\def\Abs#1{\left| #1 \right|}
\def\Card#1{\left| #1 \right|}
\def\Norm#1{\left\| #1 \right\|}
\def\Paren#1{\left( #1 \right)}		
\def\Brack#1{\left[ #1 \right]}		

\makeatletter
\def\Bigbar#1{\mathrel{\left|\vphantom{#1}\right.\n@space}}
\def\Setbar#1#2{\Set{#1 \Bigbar{#1 #2} #2}}

\def\vert{\operatorname{\mathsf{vert}}}

\def\begin@lgo{\begin{minipage}{1in}\begin{tabbing}
		\quad\=\qquad\=\qquad\=\qquad\=\qquad\=\qquad\=\qquad\=\qquad\=\qquad\=\qquad\=\qquad\=\qquad\=\qquad\=\kill}
\def\end@lgo{\end{tabbing}\end{minipage}}

\makeatother

\newcommand{\CPol}{\mbox{${\mathcal P}$}}
\newcommand{\CI}{\mbox{${\mathcal I}$}}
\newcommand{\CL}{\mbox{${\mathcal L}$}}
\newcommand{\CE}{\mbox{${\mathcal E}$}}
\newcommand{\CQ}{\mbox{${\mathcal Q}$}}

\newcommand{\uu}{\mbox{\boldmath $u$}}
\newcommand{\vv}{\mbox{\boldmath $v$}}
\newcommand{\qq}{\mbox{\boldmath $q$}}
\newcommand{\xx}{\mbox{\boldmath $x$}}
\newcommand{\cov}{\mbox{\boldmath $c$}}
\newcommand{\one}{\mbox{\boldmath $1$}}
\newcommand{\ones}{\mbox{\boldmath $1$}}
\newcommand{\zeros}{\mbox{\boldmath $0$}}
\newcommand{\MM}{\mbox{$M$}}

\newcommand{\pq}{\mbox{\boldmath $q$}}

\newcommand{\udir}{\ensuremath{\Psi}}
\DeclareMathOperator{\out}{out}
\DeclareMathOperator{\cha}{char}

\newcommand{\blank}{\ensuremath{\mathtt{*}}}
\newcommand{\up}{\ensuremath{\mathsf{up}}}
\newcommand{\down}{\ensuremath{\mathsf{down}}}
\newcommand{\zero}{\ensuremath{\mathsf{zero}}}
\newsavebox{\spacebox}
\begin{lrbox}{\spacebox}
\verb*! !
\end{lrbox}
\newcommand{\vblank}{\ensuremath{\mathtt{-}}}

\DeclareMathOperator{\decode}{decode}
\DeclareMathOperator{\subline}{subline}
\DeclareMathOperator{\fixed}{fixed}
\DeclareMathOperator{\free}{free}

\DeclareMathOperator{\isvertex}{IsVertex}
\DeclareMathOperator{\potf}{Potential}
\DeclareMathOperator{\lexpot}{LexPot}
\DeclareMathOperator{\adj}{adj}

\DeclareMathOperator{\FindFP}{\mbox{\textsc{FindFP}}}
\DeclareMathOperator{\ApproxFindFP}{\mbox{\textsc{ApproxFindFP}}}

\newcommand{\defineterm}[1]{\emph{#1}}

\newcommand{\restr}[2]{#1_{\left|#2\right.}}
\newcommand{\Restr}[2]{\tilde{#1}_{\left|#2\right.}}

\newcommand{\solnlabel}[1]{\label{sol:#1}}
\newcommand{\solnref}[1]{\ref{sol:#1}}

\begin{document}

\maketitle
\thispagestyle{empty}

\begin{abstract}
This paper studies the complexity of problems in \PPADPLS that have \emph{unique}
solutions. Three well-known examples of such problems are the problem of finding
a fixpoint of a contraction map, finding the unique sink of a Unique Sink
Orientation (USO), and solving the P-matrix Linear Complementarity Problem
(P-LCP). Each of these are promise-problems, and when the promise holds, they
always possess unique solutions.

We define the complexity class $\UEOPLc$ to capture problems of
this type. 
We first define a class that we call
$\EOPLc$, which consists of all problems that can be reduced to \EOPL.
This problem merges the canonical \PPAD-complete problem \EOL, with
the canonical \PLS-complete problem \SOD, and so \EOPL captures problems that
can be solved by a line-following algorithm that also simultaneously decreases a
potential function. 

\PUEOPLc is a promise-subclass of \EOPLc in which the line in the \EOPL instance
is guaranteed to be unique via a promise. We turn this into a non-promise class
\UEOPLc, by adding an extra solution type to \EOPL that captures any pair of
points that are provably on two different lines.

We show that $\UEOPLc \subseteq \EOPLc \subseteq \CLS$, and
that all of our motivating problems are contained in \UEOPLc:
specifically USO, P-LCP, and finding a fixpoint of a Piecewise-Linear
Contraction under an $\ell_p$-norm all lie in \UEOPLc. 
Until now, USO was not even known to lie in \PPAD or \PLS.
Our results also imply
that parity games, mean-payoff games, discounted games, and simple-stochastic
games lie in \UEOPLc.

All of our containment results are proved via a reduction to a problem that we
call One-Permutation Discrete Contraction (OPDC). This problem is motivated by a
discretized version of contraction, but it is also closely related to the
USO problem. We show that OPDC lies in \UEOPLc, and we are also able to show
that OPDC is \UEOPLc-complete.

Finally, using the insights from our reduction for Piecewise-Linear Contraction,
we obtain the first polynomial-time algorithms for finding fixed points of
contraction maps in fixed dimension for any $\ell_p$ norm, where previously such
algorithms were only known for the $\ell_2$ and $\ell_\infty$ norms. Our
reduction from \PLCP to \UEOPL allows a technique of Aldous~\cite{Aldous83} to be applied, which
in turn gives the fastest-known randomized algorithm for~\PLCP.

{\let\thefootnote\relax\footnote{This paper substantially revises and extends
the work described in our previous preprint ``End of Potential Line''
\texttt{arXiv:1804.03450}~\cite{FGMS18}.}}
\end{abstract}

\newpage
\thispagestyle{empty}


\tableofcontents

\newpage

\clearpage
\pagenumbering{arabic} 

\section{Introduction}

\paragraph{\bf Total function problems in NP.} The complexity class \TFNP
contains search problems that are guaranteed to have a solution, and whose
solutions can be verified in polynomial time~\cite{megiddo1991total}.
While it is a semantically defined complexity class and thus unlikely to
contain complete problems, a number of syntactically defined subclasses of
\TFNP have proven very successful at capturing the complexity of total search 
problems.
In this paper, we focus on two in particular, \PPAD and \PLS.
The class \PPAD was introduced in
\cite{papadimitriou1994complexity} to capture the difficulty of problems
that are guaranteed total by a parity argument. It has attracted intense
attention in the past decade, culminating in a series of papers showing that the
problem of computing a Nash-equilibrium in two-player games is \PPAD-complete
\cite{chen2009settling,daskalakis2009complexity}, and more recently a conditional
lower bound that rules out a PTAS for the problem~\cite{Rubinstein16}.
No polynomial-time algorithms for \PPAD-complete problems are known, and recent work
suggests that no such algorithms are likely to exist~\cite{BPR15,garg2016revisiting}. 
\PLS is the class of problems that
can be solved by local search algorithms (in perhaps exponentially-many steps).
It has also attracted much interest since it was introduced in
\cite{johnson1988easy}, and looks similarly unlikely to have polynomial-time
algorithms. Examples of problems that are complete for \PLS include the problem
of computing a pure Nash equilibrium in a congestion
game~\cite{fabrikant2004complexity}, a locally optimal max
cut~\cite{schaffer1991simple}, or a stable outcome in a hedonic game~\cite{GairingS10}.

\paragraph{\bf Continuous Local Search.} 
If a problem lies in both \PPAD and \PLS then it is unlikely to be complete for 
either class, since this would imply an extremely surprising containment of one class in the other.
In their 2011 paper~\cite{daskalakis2011continuous}, Daskalakis and
Papadimitriou observed that there are several prominent total function problems
in \PPADPLS for which researchers have not been able to design polynomial-time
algorithms. Motivated by this they introduced
the class $\CLS$, a syntactically defined subclass of $\cc{PPAD} \cap \cc{PLS}$.
\CLS is intended to capture the class of optimization problems over a continuous
domain in which a continuous potential function is being minimized and the
optimization algorithm has access to a polynomial-time continuous improvement
function.
They showed that many classical problems of unknown
complexity are in $\CLS$, including the problem of solving a simple
stochastic game, the more general problems of solving a Linear Complementarity
Problem with a P-matrix, finding an approximate fixpoint
to a contraction map, finding an approximate stationary point of a
multivariate polynomial, and finding a mixed Nash equilibrium of a congestion
game. 


\paragraph{\bf CLS problems with unique solutions.}

In this paper we study an interesting subset of problems that lie within \CLS,
and have \emph{unique} solutions. 

\begin{description}
\item[Contraction.] In this problem we are given a function $f : \mathbb{R}^d
\rightarrow \mathbb{R}^d$ that is purported to be $c$-contracting, meaning that for all
points $x, y \in [0, 1]^n$ we have $d(f(x),
f(y)) \le c \cdot d(x, y)$, where $c$ is a constant satisfying $0 < c < 1$, and
$d$ is a distance metric.
Banach's fixpoint theorem states that if $f$ is contracting, then it has a
unique \emph{fixpoint}~\cite{Banach1922}, meaning that there is a unique point
$x \in \mathbb{R}^d$ such that $f(x) = x$.

\item[P-LCP.] The \emph{P-matrix linear complementarity problem} (P-LCP) is a
variant of the linear complementarity problem in which the input matrix is a
P-matrix~\cite{cottle2009linear}. An interesting property of this problem is
that, if the input matrix actually is a P-matrix, then the problem is guaranteed
to have a unique solution~\cite{cottle2009linear}.
Designing a polynomial-time algorithm for P-LCP has been open for
decades, at least since the 1978 paper of Murty~\cite{murty1978computational}
that provided exponential-time examples for \emph{Lemke's algorithm}~\cite{lemke1965bimatrix} for P-LCPs.

\item[USO.] A \emph{unique sink orientation} (USO) is an orientation of the
edges of an $n$-dimensional hypercube such that every face of the cube has a
unique sink. Since the entire cube is a face of itself, this means that there is
a unique vertex of the cube that is a sink, meaning that all edges are oriented
inwards. The USO problem is to find this \emph{unique} sink.
\end{description} 

All of these problems are most naturally stated as \emph{promise} problems. This
is because we have no way of verifying up front whether a function is
contracting, whether a matrix is a P-matrix, or whether an orientation is a USO.
Hence, it makes sense, for example, to study the contraction problem where it is
promised that the function $f$ is contracting, and likewise for the other two.

However, each of these problems can be turned into non-promise problems that lie
in TFNP. In the case of Contraction, if the function $f$ is not contracting,
then there exists a short certificate of this fact. Specifically, any pair of
points $x, y \in \mathbb{R}^d$ such that $d(f(x), f(y)) > c \cdot d(x, y)$ give
an explicit proof that the function $f$ is not contracting. We call these
\emph{violations}, since they witness a violation of the promise that is
inherent in the problem. 

So Contraction can be formulated as the non-promise problem of either
finding a solution, or finding a violation. This problem is in \TFNP because in
the case where there is not a unique solution, there must exist a violation of
the promise. 
The P-LCP and USO problems also have violations that
can be witnessed by short certificates, and so they
can be turned into non-promise problems contained in the same way, and these
problems also lie in \TFNP.

For Contraction and P-LCP we actually have the stronger result that both
problems are in 
\CLS~\cite{daskalakis2011continuous}. Prior to this work USO was not known
to lie in any non-trivial subclass of \TFNP, and placing USO into a non-trivial
subclass of \TFNP was identified as an interesting open problem by
Kalai~\cite[Problem 6]{Kalai18}.

We remark that not every problem in CLS has the uniqueness properties that we
identify above. 
For example, the KKT problem~\cite{daskalakis2011continuous} lies in \CLS, but
has no apparent notion of having a unique solution. 
The problems that we identify here seem to share the special property that there
is a natural promise version of the problem, and that promise problem always has
a unique solution. 

\subsection{Our contribution}

In this paper, we define a complexity class that naturally captures the
properties exhibited by problems like Contraction, P-LCP, and USO. In fact, we
define two new sub-classes of \CLS. 

\paragraph{\bf End of potential line.}

The complexity class \EOPLc contains every problem that can be reduced in
polynomial time to 
the problem \EOPL, which we define in this paper (Definition~\ref{def:EOPL}). The \EOPL problem 
unifies
in an extremely natural way the circuit-based views of \PPAD and of \PLS.
The canonical
\PPAD-complete problem is \EOL, a problem that provides us with an
exponentially large graph consisting of lines and cycles, and asks us to find
the end of one of the lines. The canonical \PLS-complete problem provides us
with an exponentially large DAG, whose acyclicity is guaranteed by the
existence of a \defineterm{potential function} that increases along each edge.
The problem \EOPL is an instance of \EOL that \emph{also} has a potential
function that increases along each edge.

So the class \EOPLc captures problems that admit a
\emph{single} combinatorial proof of their joint membership in the classes \PPAD
of fixpoint problems and \PLS of local search problems, and is a
combinatorially-defined alternative to the class~\CLS. We are able to show that
$\EOPLc \subseteq \CLS$ (Corollary~\ref{cor:eoplinCLS}), by providing a
polynomial-time reduction from \EOPL to the \EOML problem defined by 
Hub{\'a}{\v{c}}ek and Yogev~\cite{hubavcek2017hardness}, which they have shown
to lie in~\CLS.

We remark that it is an interesting open problem to determine whether $\EOPLc =
\CLS$. The inspiration behind both classes was to capture problems in $\PPAD
\cap \PLS$. The class \CLS does this by affixing a potential function to the
\PPAD-complete Brouwer fixpoint problem, while \EOPLc does this affixing a
potential function to the \PPAD-complete problem \EOL. The class \EOPLc is not
the main focus of this paper, however.

\paragraph{\bf Unique end of potential line.}

An \EOPL instance consists of an exponentially large graph that contains only
lines (the cycles that can appear in \EOL instances are ruled out by the
potential function.) The problem explicitly gives us the start of one of these
lines. A solution to the problem is a vertex that is at the end of any line,
other than the given start vertex. 
We could find a solution by 
following the line from the start vertex until we find the other end,
although that may take exponential time. There may be many other lines
in the graph though, and the starts and ends of these lines are also solutions.

We define the promise-problem \PUEOPL, in which it is promised that there is a
unique line in the graph. This line must be the one that begins at the given
starting vertex, and so the only solution to the problem is the other end of
that line. Thus, if the promise is satisfied, the problem has a
unique solution. We can define the promise-class \PUEOPLc which contains all
promise-problems that can be reduced in polynomial-time to \PUEOPL.

We are not just studying promise problems in this paper, however. We can turn
\PUEOPL into a non-promise problem by defining appropriate violations. One might
imagine that a suitable violation would be a vertex that is the start of a
second line. Indeed, if we are given the promise that there is no start to a
second line, then we do obtain the problem \PUEOPL. However, with just this
violation, we obtain a problem that is identical to \EOPL, which is not what we
are intending to capture. 

Instead we add a violation that captures any pair of
vertices $v$ and $u$ that are provably on different lines, even if $v$ and $u$
are in the middle of their respective lines. We do this by using the potential
function: if $v$ and $u$ have the same potential, then they must be on different
lines, and likewise if the potential of $u$ lies between the potential of $v$
and the potential of the successor of~$v$. We formalise this as the problem
\UEOPL (Definition~\ref{def:UEOPL}), and we define the complexity class \UEOPLc
to contain all (non-promise) problems that can be reduced in polynomial-time to
\UEOPL.

We have that $\UEOPLc \subseteq \EOPLc$ by definition, since 
\UEOPL simply adds an extra type of violation to \EOPL.
We remark that this new violation makes
the problem substantially different from \EOPL. In \EOPL only the starts and
ends of lines are solutions, while in \UEOPL there are many more solutions
whenever there is more than one line in the instance. 
As such, we view \UEOPLc as capturing a distinct subclass of problems in \EOPLc,
and we view it as the natural class for promise-problems in \PPADPLS that have
unique solutions.


\paragraph{UEOPL containment results.}

We show that USO, P-LCP, and a variant of the Contraction problem all lie in
\UEOPLc. We define the concept of a \emph{promise-preserving} reduction, which
is a polynomial-time reduction between two problems A and B, with the property
that if A is promised to have a unique solution, then the instance of B that is
produced by the reduction will also have a unique solution. All of the
reductions that we produce in this paper are promise-preserving, which means
that whenever we show that a problem is in \UEOPLc, we also get that the
corresponding promise problem lies in \PUEOPLc.

\begin{theorem}[cf.\ Theorem~\ref{thm:uso}]
USO is in \UEOPLc under promise-preserving reductions.
\end{theorem}

For the USO problem, our \UEOPLc containment result substantially advances our
knowledge about the problem. Prior to this work, the problem was only known to
lie in \TFNP, and Kalai~\cite[Problem 6]{Kalai18} had posed the challenge to
place it in some non-trivial subclass of \TFNP. Our result places USO in
\UEOPLc, \EOPLc, \CLS, \PPAD (and hence \PPA and \PPP), and \PLS, and so we
answer Kalai's challenge by placing the problem in \emph{all} of the standard
subclasses of \TFNP.

This result is in some sense surprising. Although every face of a USO has a
unique sink, the orientation itself may contain cycles, and so there is no
obvious way to define a potential function for the problem. Moreover, none of
the well-known algorithms for solving USOs~\cite{SzaboW01,HansenPZ14} have the
line-following nature needed to produce an \EOL instance. Nevertheless, our
result shows that one can produce an \EOL instance that has potential function
for the USO problem.

\begin{theorem}[cf.\ Theorem~\ref{thm:plcp2uso} and Theorem~\ref{thm:plcp2ueopldirectly}]
There are two different variants of the P-LCP problem, both of which lie in
\UEOPLc under promise-preserving reductions.
\end{theorem}

We actually provide two different promise-preserving reductions from P-LCP to
\UEOPL. The issue here is that there are many possible types of violation that
one can define for P-LCP. So far, the standard formulation of P-LCP either asks
for a solution, or a \emph{non-positive principle minor} of the input matrix.
A matrix is a P-matrix if and only if all of its principle minors are positive,
and so this is sufficient to define a total problem. 

The reduction from P-LCP to \EOL can map the start and end of each line back to
either a solution or a non-positive principle minor. Our problem is that the
extra violations in the \UEOPL instance, corresponding to a proof that there are
multiple lines, do not easily map back to non-positive principle minors. They do
however map back to other short certificates that the input matrix is not a
P-matrix. For example, a matrix is a P-matrix if and only if it does not reverse
the sign of any non-zero vector~\cite{cottle2009linear}. So we can also
formulate P-LCP as a total problem that asks for a solution, or a non-zero
vector whose sign is reversed. 

We study the following variants of P-LCP. The first variant asks
us to either find a solution, or find a non-positive principle minor, or find a
non-zero vector whose sign is reversed. The second variant asks us to either
find a solution, or a non-positive principle minor, or a third type of violation
whose definition is inspired by a violation of the USO property. In all cases,
we either solve the problem, or obtain a short certificate that the input was
not a P-matrix, although the format of these certificates can vary.

It is not clear whether these variants
are equivalent under polynomial-time reductions, as one would need to be able to
map one type of violation to the other efficiently.
We remark that if one is only interested in the promise problem, then the choice
of violations is irrelevant. Both of our reductions show that promise P-LCP lies
in \PUEOPLc.

\begin{theorem}[cf.\ Theorem~\ref{thm:lcm}]
Finding the fixpoint of a piecewise linear contraction map in the $\ell_p$ norm
is in \UEOPLc under promise-preserving reductions, for any $p \in \Natural \cup
\{\infty\}$.
\end{theorem}

For Contraction, we study contraction maps specified by \emph{piecewise linear}
functions that are contracting with respect to an $\ell_p$ norm. This differs the contraction problem studied
previously~\cite{daskalakis2011continuous}, where 
the function is given by an arbitrary
arithmetic circuit. To see why this is necessary, note that although every contraction map has
a unique fixpoint, if we allow the function to be specified by an arbitrary arithmetic
circuit, then there is no
guarantee that the fixpoint is rational. So it is not clear whether finding the
\emph{exact} fixpoint of a contraction map even lies in \FNP.

Prior work has avoided this issue by instead asking for an \emph{approximate}
fixpoint, and the problem of finding an approximate fixpoint of a contraction
map specified by an arithmetic circuit lies in
\CLS~\cite{daskalakis2011continuous}. However, if we look for approximate
fixpoints, then we destroy the uniqueness property that we are interested in,
because there are infinitely many approximate fixpoint solutions surrounding any
exact fixpoint.

So, we study the problem where the function is represented by a \LinearFIXP
arithmetic circuit~\cite{EtessamiY10}, which is a circuit in which the multiplication of two
variables is disallowed. This ensures that, when the function actually is
contracting, there is a unique rational fixpoint that we can produce. We note
that this is still an interesting class of contraction maps, since it is
powerful enough to represent simple-stochastic
games~\cite{EtessamiY10}.

We place this problem in \UEOPLc via a promise-preserving reduction. Our
reduction can produce multiple types of violation. In addition to the standard
violation of a pair of points at which~$f$ is not contracting, our reduction
sometimes produces a different type of violation (cf.\
Definition~\ref{def:LCM}), which while not giving an explicit violation of
contraction, still gives a short certificate that~$f$ is not contracting.

\begin{theorem}
The following problems are in \UEOPLc
\begin{itemize}
\itemsep1mm
\item Solving a parity game.
\item Solving a mean-payoff game.
\item Solving a discounted game.
\item Solving a simple-stochastic game.
\item Solving the ARRIVAL problem.
\end{itemize}
\end{theorem}

Finally, we observe that our results prove that several other problems lie in
\UEOPLc. The simple-stochastic game (SSG) problem is known to reduce to
Contraction~\cite{EtessamiY10} and to
P-LCP~\cite{hansen2013complexity}, and thus our two results give two separate
proofs that the SSG problem lies in \UEOPLc. It is known that discounted games
can be reduced to SSGs~\cite{zwick1996complexity}, mean-payoff games can be
reduced to discounted games~\cite{zwick1996complexity},
and parity games can be reduced to mean-payoff games~\cite{puri1996theory}. So
all of these problems lie in \UEOPLc too. Finally, 
G{\"{a}}rtner et al.~\cite{GHHKMS18} noted that 
the ARRIVAL problem~\cite{DGKMW17} lies in \EOPLc, and in fact their \EOPL
instance always contains exactly one line, and so the problem also lies in 
\UniqueEOPLc.

We remark that none of these are promise-problems. Each of them
can be formulated so that they \emph{unconditionally} have a unique solution.
Hence, these problems seem to be easier than the problems captured by \UEOPLc,
since problems that are complete for \UEOPLc only have a unique solution
conditioned on the promise that there are no violations.

\paragraph{\bf A UEOPL-complete problem.}

In addition to our containment results, we also give a \UEOPLc-completeness
result. Specifically, we show that \emph{One-Permutation Discrete Contraction}
(OPDC) is complete for \UEOPLc.

OPDC is a problem that is inspired by both Contraction and USO. Intuitively, it 
is a discrete version of Contraction. The inputs to the
problem are (a concise representation of) a discrete grid of points $P$ covering
the space $[0, 1]^d$, and a set of \emph{direction} functions $\mathcal{D} =
D_{i = 1 \dots d}$, where each function $D_i$ has the form $D_i : P \rightarrow
\{\up, \down, \zero\}$. To discretize a map $f : [0, 1]^d \rightarrow [0, 1]^d$
we simply define $D_i$ so that 
\begin{itemize}
\item $D_i(p) = \up$ whenever $f(p)_i > p_i$,
\item $D_i(p) = \zero$ whenever $f(p)_i = p_i$, and
\item $D_i(p) = \down$ whenever $f(p)_i < p_i$.
\end{itemize} 
So the direction function for dimension $i$ simply points in the
direction that $f$ moves in dimension $i$.
To solve the problem, we seek a point $p \in P$ such that $D_i(p) = \zero$ for
all $i$, which corresponds to a fixpoint of $f$.

Why is the problem called \emph{One Permutation} Discrete Contraction? This is
due to some extra constraints that we place on the problem. An interesting
property of a function $f$ that is contracting with respect to an $\ell_p$ norm
is that if we restrict the function to a \emph{slice}, meaning that we fix some
of the coordinates and let others vary, then the resulting function is still
contracting. We encode this property into the OPDC problem, by insisting that
slices should also have unique fixpoints when we ignore the dimensions not in
the slice. However, we do not do this for all slices, but only for
\emph{$i$-slices} in
which the \emph{last $d - i + 1$} coordinates have been fixed. In this sense,
our definition depends on the order of the dimensions. If we rearranged the
dimensions into a different permutation, then we would obtain a different
problem, so each problem corresponds to some particular permutation of the
dimensions. The name of the problem was chosen to reflect this fact.

We remark that, although the problem was formulated as a discretization of
Contraction, it is also closely related to the USO problem. Specifically, if we
take the set of points $P$ to be the $\{0, 1\}^n$ hypercube, then the direction
functions actually specify an orientation of the cube. Moreover, the condition
that every slice should have have unique fixpoint \emph{exactly corresponds} to
the USO property that every face should have a unique sink. However, since we
only insist on this property for $i$-slices, OPDC can be viewed as a variant of
USO in which \emph{only some} of the faces have unique sinks.

\begin{theorem}[cf.\ Theorem~\ref{thm:opdc}]
OPDC lies in \UEOPLc under promise-preserving reductions.
\end{theorem}

OPDC actually plays a central role in the paper. We reduce Piecewise-Linear
Contraction, USO, and P-LCP to it, and we then reduce OPDC to \UEOPL, as shown
in Figure~\ref{fig:reductions}. The reduction from OPDC to \UEOPL is by far the
most difficult step.

\begin{wrapfigure}{R}{0.5\textwidth}
\scalebox{0.8}{ 
\begin{tikzpicture}[node distance=1.7cm, font=\sffamily]
\tikzset{>={Latex[width=2mm,length=2mm]}}
\tikzstyle{cc}=[rectangle, draw, rounded corners, font=\Large]
\tikzstyle{problems}=[font=\large]
\tikzstyle{problem}=[font=\Large]

\node[problems] (ssg) {\textsf{Simple Stochastic Games}};
\node[problems,below of=ssg,yshift=4.5mm] (dpg) {\textsf{Discounted Payoff Games}};
\node[problems,below of=dpg,yshift=4.5mm] (mpg) {\textsf{Mean-payoff Games}};
\node[problems,below of=mpg,yshift=4.5mm] (parity) {\textsf{Parity Games}};

\node[problem,above left of=ssg, xshift=-1.6cm, yshift=1.25cm] (contraction) {\textsf{PL Contraction}};
\node[problem,above right of=ssg, xshift=1.6cm, yshift=1.25cm] (uso) {\textsf{Unique Sink Orientation}};
\node[problem,above right of=ssg, xshift=0.3cm, yshift=0.05cm] (plcp) {\textsf{P-matrix LCP}};
\node[problem,above of=ssg, yshift=2.5cm] (opdc) {\textsf{One-Permutation Discrete Contraction}};

\node[cc, above of=opdc] (ueopl) {$\mathsf{UniqueEOPL}$};
\node[cc, above of=ueopl] (eopl) {$\mathsf{EOPL}$};
\node[cc, above of=eopl] (cls) {$\mathsf{CLS}$};



\path[->, thick] (parity) edge (mpg);
\path[->, thick] (mpg) edge (dpg);
\path[->, thick] (dpg) edge (ssg);
\path[->, thick] (ssg) edge (contraction);
\path[->, thick] (ssg) edge (plcp);
\path[->, thick] (contraction) edge node[left,xshift=-2mm] {Theorem~\ref{thm:lcm}} (opdc);
\path[->, thick] (uso) edge node[right,xshift=2mm] {Theorem~\ref{thm:uso}} (opdc);
\path[->, thick] (plcp) edge node[right,xshift=1mm, yshift=-1mm] {Theorem~\ref{thm:plcp2uso}} (uso);
\path[<->, thick] (opdc) edge node[right] {Theorem~\ref{thm:opdc-completeness}} (ueopl);
\path[->, thick] (ueopl) edge node[right] {By definition} (eopl);
\path[->, thick] (eopl) edge node[right] {Corollary~\ref{cor:eoplinCLS}} (cls);


\end{tikzpicture}
}
\caption{The relationship of \UEOPLc to other classes and problems.}
\label{fig:reductions}
\end{wrapfigure}
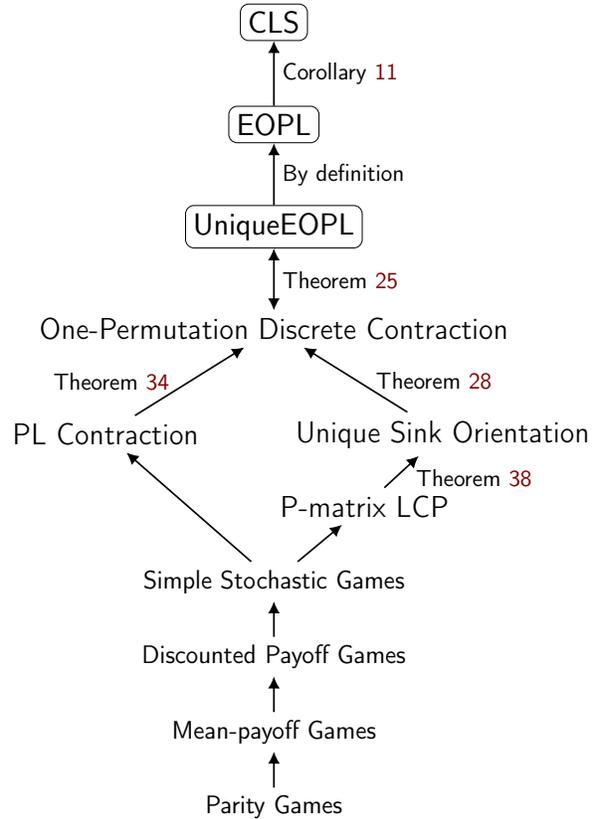

In the case where the promise is satisfied, meaning that the OPDC instance has a
unique fixpoint, the reduction
defines a single line that starts at the point $p \in P$ with $p_i = 0$ for all
$i$, and ends at the unique fixpoint. This line walks around the grid, following
the direction functions given by $\mathcal{D}$ in a specific manner that ensures
that it will find the fixpoint, while also decreasing a potential function at
each step. We need to very carefully define the vertices of
this line, to ensure that the line is unique, and for this we crucially rely on
the fact that every $i$-slice also has a unique fixpoint.

This does not get us all the way to \UEOPL though, because the line we describe
above lacks a \emph{predecessor} circuit. In \UEOPL, each vertex has a
predecessor, a successor and a potential, but the line we construct only gives
successors and potentials to each vertex. To resolve this, we apply the
\emph{pebbling game} reversibility argument introduced by 
Bitansky et al~\cite{BPR15}, and later improved by 
Hub{\'a}{\v{c}}ek and
Yogev~\cite{hubavcek2017hardness}.
Using this technique 
allows us to produce a predecessor
circuit, as long as there is exactly one line in the instance.

Our reduction also handles violations in the OPDC instance. The key challenge
here is that the pebbling game argument assumes that there is exactly one line,
and so far it has only been applied to promise-problems. We show that the
argument can be extended to work with non-promise problems that may have
multiple lines. This can cause the argument to break, specifically when multiple
lines are detected, but we are able to show that these can be mapped back to
violations in the OPDC instance.

\begin{theorem}[cf.\ Theorem~\ref{thm:opdc-completeness}]
\OPDC is \UEOPLc-complete under promise-preserving reductions, even when the set
of points $P$ is a hypercube.
\end{theorem}

We show that OPDC is \UEOPLc-hard by giving a polynomial-time
promise-preserving reduction from \UEOPL to OPDC. This means that OPDC is
\UEOPLc-complete, and the variant of OPDC in which it is promised that there are
no violations is \PUEOPLc-complete.

Our reduction produces an \OPDC instances where the set of points $P$ is the
boolean hypercube $\{0, 1\}^n$. In the case where the \UEOPL instance has no
violations, meaning that it contains a single line, the reduction embeds this
line into the hypercube. To do this, it splits the line in half. The second half
is embedded into a particular sub-cube, while the first half is embedded into
all other sub-cube. This process is recursive, so each half of the line is again
split in half, and further embedded into sub-cubes. The reduction ensures that
the only fixpoint of the instance corresponds to the end of the line. If the
\UEOPL instance does have violations, then this embedding may fail. However, in
any instance where the embedding fails, we are able to produce a violation for
the original \UEOPL instance. 

We remark that this hardness reduction makes significant progress towards
showing a hardness result for Contraction and USO. As we have mentioned, OPDC
is a discrete variant of Contraction, and when the set of points is a hypercube,
the problem is also very similar to USO. The key difference is that OPDC 
insists that only $i$-slices should have a unique fixpoint, whereas Contraction
and USO insist that \emph{all} slices should have unique fixpoints. To show
a hardness result for either of those two problems, one would need to produce an
OPDC instance with that property.

\paragraph{\bf New algorithms.}

Our final contributions are algorithmic and arise from the structural
insights provided by our containment results. 
Using the ideas from our reduction from Piecewise-Linear Contraction to \UEOPL, 
we obtain the first
polynomial-time algorithms for finding fixpoints of 
Piecewise-Linear Contraction maps in fixed dimension for any $\ell_p$
norm, where previously such algorithms were only known for $\ell_2$ and
$\ell_\infty$ norms. If this input is not a contraction map, then our algorithm
may instead produce a short certificate that the function is not contracting.

We also show that these results can be extended to the case where the
contraction map is given by a general arithmetic circuit. In this case, we
provide polynomial-time algorithm that either finds an approximate fixpoint, or
produces a short certificate that the function is not contracting.

An interesting consequence of our algorithms is that it is now unlikely that
$\ell_p$-norm Contraction in fixed-dimension is \CLS-complete. This should be
contrasted to the recent result of Daskalakis et al.~\cite{DTZ17}, who showed
the variant of the contraction problem where a metric is given as part of the
input is \CLS-complete, even in dimension 3. Our result implies that it is
unlikely that this can be directly extended to $\ell_p$ norms, at
least not without drastically increasing the number of dimensions in the
instance.

Finally, as noted in~\cite{GHHKMS18}, one of our reductions from \PLCP to \EOPL
allows a technique of Aldous~\cite{Aldous83} to be applied, which in turn gives the fastest
known randomized algorithm for~\PLCP.

\subsection{Related work}

\paragraph{\bf CLS.}

Recent work by Hub\'{a}\v{c}ek and Yogev~\cite{hubavcek2017hardness} proved 
lower bounds for \CLS. They introduced
a problem known as \EOML which they showed was in \CLS, and for which they
proved a query complexity lower bound of $\Omega(2^{n/2}/\sqrt{n})$ and
hardness under the assumption that there were one-way permutations and
indistinguishability obfuscators for problems in $\cc{P_{/poly}}$.
Another recent result showed that the search version of the Colorful
Carath\'eodory Theorem is in $\cc{PPAD} \cap \cc{PLS}$, and left open
whether the problem is also in $\CLS$~\cite{colorfulcara2017}.

Until recently, it was not known whether there was a natural \CLS-complete problem.
In their original paper, Daskalakis and Papadimitriou suggested two natural
candidates for \CLS-complete problems, \ContractionMap and \PLCP, which we 
study in this paper.
Recently, two variants of \ContractionMap have been shown to be \CLS-complete.
Whereas in the original definition of \ContractionMap it is assumed that 
an $\ell_p$ or $\ell_\infty$ norm is fixed, and the contraction property 
is measured with respect to the metric induced by this fixed norm, 
in these two new complete variants, a metric~\cite{DTZ17} and
meta-metric~\cite{FGMS17} 
are given as input to the problem\footnote{
The result for meta-metrics appeared in an unpublished earlier version of 
this paper~\cite{FGMS17} that contained only a small fraction 
of the results from the current version. 
This result for meta-metrics is dropped from this current version given
that it has been superseded, and is not important for our main message.
}. 

\paragraph{\bf P-LCP.}

Papadimitriou showed that \PLCP, the problem of solving the LCP or returning a
violation of the P-matrix property, is in
\PPAD~\cite{papadimitriou1994complexity} using Lemke's algorithm.
The relationship between Lemke's algorithm and \PPAD has been studied
by Adler and Verma~\cite{AV11}.
Later, Daskalakis and Papadimitrou showed that \PLCP is in
\CLS~\cite{daskalakis2011continuous}, using the potential reduction method
in~\cite{kojima1992interior}.  
Many algorithms for \PLCP have been studied, e.g.,~\cite{murty1978computational,morris2002randomized,kojima1991unified}.
However,
no polynomial-time algorithms are known for \PLCP, or for the promise
version where one can assume that the input matrix is a P-matrix.

The best known algorithms for \PLCP are based on
a reduction to Unique Sink Orientations (USOs) of cubes~\cite{StickneyW78}.
For an P-matrix LCP of size $n$, the USO algorithms of~\cite{SzaboW01} apply, and 
give a deterministic 
algorithm that runs in time $O(1.61^n)$ and a randomized algorithm with
expected running time $O(1.43^n)$.
The application of Aldous' algorithm~\cite{Aldous83} to the \UEOPL instance that we produce from
a P-matrix LCP takes expected time $2^{n/2}\cdot \poly(n) = O(1.4143^n)$ in the worst case.

\paragraph{\bf Unique Sink Orientations.}
In this paper we study USOs of cubes, a problem that was first studied 
by Stickney and Watson~\cite{StickneyW78} in the context of P-matrix LCPs.
A USO arising from a P-matrix may be cyclic.
Motivating by Linear Programming, \emph{acylic}, AUSOs have also
been studied, both for cubes and general polytopes~\cite{Hoke88,GartnerS06}.
Recently G\"artner and Thomas studied the computational 
complexity of recognizing USOs and AUSOs~\cite{GartnerT15}.
They found that the problem is \coNP-complete for USOs and \PSPACE-complete for AUSOs.
A series of papers provide upper and lower bounds for specific algorithms 
for solving (A)USOs, including~\cite{SzaboW01,HansenPZ14,fearnley2016complexity,MatousekS04,SchurrS05,FriedmannHZ11,Thomas17}.
An AUSOs on a $n$-dimensional cube can be solved in subexponential time,
by the RANDOM-FACET algorithm, which is essentially tight for this algorithm~\cite{Gartner02}.
An almost quadratic lower bound on the number of vertex evaluations needed
to solve a general USO is known~\cite{SchurrS04}; unlike for AUSOs, the best
running times known for general USOs, as for P-matrix LCPs, are exponential.
To be best of our knowledge, we are first to study the general problem
of solving a USO from a complexity-theoretic point of view.

\paragraph{\bf Contraction.}

The problem of computing a fixpoint of a continuous map $f:\mathcal{D}\mapsto \mathcal{D}$
with Lipschitz constant
$c$ has been extensively studied, in both continuous and discrete variants~\cite{ChenD09,ChenD08,DengQSZ11}.
For arbitrary maps with $c>1$, exponential bounds on the query
complexity of computing fixpoints are known~\cite{HirschPV89,ChenD05}.
In \cite{BoonyasiriwatSX07,HuangKS99,Sik09}, algorithms for computing fixpoints 
for specialized maps such as weakly ($c=1$) or strictly ($c < 1$) contracting maps are studied.
For both cases, algorithms are known for the case of $\ell_2$ and $\ell_\infty$ norms, 
both for absolute approximation ($||x-x^*||\le \epsilon$ where $x^*$ is
an exact fixpoint) and relative approximation ($||x-f(x)||\le \epsilon$). 
A number of algorithms are known for
the $\ell_2$ norm handling both types of approximation \cite{NemYud83,HuangKhachSik99,Sik01}. 
There is an exponential lower bound for absolute approximation with $c=1$ \cite{Sik01}.
For relative approximation and a domain of dimension $d$, an
$O(d\cdot\log{1/\epsilon})$ time algorithm is known \cite{HuangKhachSik99}. 
For absolute approximation with $c<1$, an ellipsoid-based algorithm with time complexity
$O(d \cdot [\log(1/\epsilon) + \log(1/(1-c))])$ is known \cite{HuangKhachSik99}. 
For the $\ell_\infty$ norm, \cite{ShellSik03} gave an algorithm to find an
$\epsilon$-relative approximation in time $O(\log(1/\epsilon)^d)$ which is
polynomial for constant $d$. In summary,
for the 
$\ell_2$ norm  polynomial-time algorithms are known for strictly contracting maps; for the 
$\ell_\infty$ norm algorithms that are polynomial time for constant dimension are known.
For arbitrary $\ell_p$ norms, to the best of our knowledge, no polynomial-time
algorithms for constant dimension were known before this paper.

\paragraph{\bf Infinite games.}

Simple Stochastic Games are related to Parity games, which are an extensively studied class of two-player zero-sum
infinite games that capture important problems in formal verification and logic~\cite{EmersonJ91}.
There is a sequence of polynomial-time reductions from parity games 
to mean-payoff games to discounted games to simple stochastic 
games~\cite{puri1996theory,gartner2005simple,jurdzinski2008simple,zwick1996complexity,hansen2013complexity}.
The complexity of solving these problems is unresolved and has received 
much attention over many years~(see, for example, 
\cite{zwick1996complexity,condon1992complexity,fearnley2010linear,jurdzinski1998deciding,bjorklund2004combinatorial,fearnley2016complexity}).
In a recent breakthrough~\cite{parity}, a quasi-polynomial time algorithm for parity games
have been devised, and there are now several algorithms with this running time~\cite{parity,JL17,FJS0W17}.
For mean-payoff, discounted, and simple stochastic games, the best-known 
algorithms run in randomized subexponential
time~\cite{ludwig1995subexponential}. The existence of a polynomial time
algorithm for solving any of these games would be a major breakthrough. Simple
stochastic games can also be reduced in polynomial time to Piecewise-Linear
Contraction with the $\ell_\infty$ norm~\cite{EtessamiY10}.


\subsection{Future directions}

A clear direction for future work is to show that further problems are
\UEOPLc-complete. We have several conjectures.

\begin{conjecture}
USO is hard for \UniqueEOPLc.
\end{conjecture}

We think that, among our three motivating problems, USO is the most likely to be
\UEOPLc-complete. Our hardness proof for OPDC already goes some way towards
proving this, since we showed that OPDC was hard even when the set of points is
a hypercube. The key difference between OPDC on a hypercube and USO is that OPDC
only requires that the faces corresponding to $i$-slices should have unique
sinks, while USO requires that all faces should have unique sinks. 

\begin{conjecture}
Piecewise-Linear Contraction in an $\ell_p$ norm is hard for \UniqueEOPLc.
\end{conjecture}

Our OPDC hardness result also goes some way towards showing that Piecewise-Linear 
Contraction is hard, however there are more barriers to overcome here. In
addition to the $i$-slice vs.\ all slice issue, we would also need to convert
the discrete OPDC problem to the continuous contraction problem. Converting
discrete problems to continuous fixpoint problems has been well-studied in the
context of \PPAD-hardness reductions~\cite{daskalakis2009complexity,mehta}, but
here the additional challenge is to carry out such a reduction while maintaining
the contraction property.

Aside from hardness, we also think that the relationship between Contraction and
USO should be explored further. Our formulation of the OPDC problem exposes
significant similarities between the two problems, which until this point have
not been recognised. Can we reduce USO to Contraction in polynomial time?

\begin{conjecture}
P-LCP is hard for \UniqueEOPLc.
\end{conjecture}

Of all of our conjectures, this will be the most difficult to prove. Since P-LCP reduces to USO, the
hardness of USO should be resolved before we attempt to show that P-LCP is hard.
One possible avenue towards showing the hardness of P-LCP might be to
reduce from Piecewise-Linear Contraction. Our \UEOPLc containment proof for 
Piecewise-Linear Contraction makes explicit use of the fact that the problem can
be formulated as an LCP, although in that case the resulting matrix is not a
P-matrix. Can we modify the reduction to produce a P-matrix?

\begin{conjecture}
$\UEOPLc \subset \EOPLc = \CLS$.
\end{conjecture}

The question of $\EOPLc$ vs $\CLS$ is unresolved, and we actually think it could
go either way. One could show that $\EOPLc = \CLS$ by placing either of the two
known \CLS-complete Contraction variants into \EOPLc~\cite{DTZ17,FGMS17}. If the
two classes are actually distinct, then it is interesting to ask which of the
problems in \CLS are also in \EOPLc. 

On the other hand, we believe that \UEOPLc is a strict subset of \EOPLc.
The evidence for this is that the extra violation in \UEOPL that does not appear
in \EOPL changes the problem significantly. This new violation will introduce
many new solutions whenever there are multiple lines in the instance, and so 
it is unlikely, in our view, that one could reduce \EOPL to \UEOPL. Of course, there is
no hope to unconditionally prove that $\UEOPLc \subset \EOPLc$, but we can ask
for further evidence to support the idea. For example, can oracle separations shed any
light on the issue?

Finally, we remark that \UEOPLc is the closest complexity class to \FP, among
all the standard sub-classes of \TFNP. However, we still think that further
subdivisions of \UEOPLc will be needed. Specifically, we do not believe that
simple stochastic games, or any of the problems that can be reduced to them, are
\UEOPLc-complete, since all of these problems have unique solutions
unconditionally. Further research will be needed to classify these problems.

\section{Unique End of Potential Line}
\label{sec:EOPL}

In this section we define two new complexity classes called $\EOPLc$ and
$\UniqueEOPLc$. These two classes are defined by merging the definitions of \PPAD and
\PLS, so we will begin by recapping those classes.

The complexity class \PPAD contains every problem that can be reduced to
\EOL~\cite{papadimitriou1994complexity}.

\begin{definition}[\EOL]
Given Boolean circuits $S,P : \Set{0,1}^n \to \Set{0,1}^n$ such that $P(0^n)
=0^n \neq S(0^n)$, find one of the following:
\begin{enumerate}[label=(E\arabic*)]
\item \solnlabel{E1} A point $x \in \Set{0,1}^n$ such that $P(S(x)) \neq x$.
\item \solnlabel{E2} A point $x \in \Set{0,1}^n$ such that $x \ne 0^n$ and $S(P(x)) \neq x $.
\end{enumerate}
\end{definition}

Intuitively, the problem defines an exponentially large graph where all vertices
vhave in-degree and out-degree at most one.  Each bit-string in $\Set{0,1}^n$
defines a vertex, while the functions $S$ and $P$ define \emph{successor} and
\emph{predecessor} functions for each vertex. A directed edge exists from vertex
$x$ and $y$ if and only if $S(x) = y$ and $P(y) =
x$. Any vertex $x$ for which $P(S(x)) \ne x$ has no outgoing edge, while every
vertex $y$ with $S(P(x)) \ne x$ has no incoming edge.

The condition that $P(0^n) = 0^n \neq S(0^n)$ specifies that the vertex $0^n$
has no incoming edge, and so it is the start of a line.
To solve the problem, we must find either a solution of type \solnref{E1}, which is a
vertex $x$ that is the end of a line, or a solution of type \solnref{E2}, which is a
vertex $x$ other than $0^n$ that is the start of a line. Since we know that the
graph has at least one source, there must at least exist a solution of type \solnref{E1},
and so the problem is in \TFNP.

The complexity class \PLS contains every problem that can be reduced to the
\SOD\footnote{While this is not the standard definition of PLS, it has been
observed that the standard \PLS-complete problem can easily be recast as a \SOD
instance~\cite{papadimitriou1994complexity}.}

\begin{definition}[\SOD]
Given a Boolean circuit $S : \Set{0,1}^n \to \Set{0,1}^n$ such that $S(0^n) \ne
0^n$, and a circuit 
$V: \Set{0,1}^n \to \Set{0,1,\dotsc,2^m - 1}$ find:
\begin{enumerate}[label=(S\arabic*)]
\item \solnlabel{S1} A vertex $x \in \Set{0,1}^n$ such that $S(x) \ne x$ and either $S(S(x)) = x$ or $V(S(x)) \leq V(x)$.
\end{enumerate}
\end{definition}

Once again, the problem specifies an exponentially large graph on the vertex set
$\Set{0,1}^n$, but this time the only guarantee is that each vertex has
out-degree one. The circuit $S$ gives a successor function. In this problem,
some bit-strings do not correspond to vertices in the graph. Specifically, if we
have $S(x) = x$ for some bit-string $x \in \Set{0,1}^n$, then $x$ does not
encode a vertex.

The second circuit $V$ gives a \emph{potential} to each vertex
from the set $\Set{0,1,\dotsc,2^m - 1}$. An edge exists in the graph if and only
if the potential \emph{increases} along that edge. Specifically, there is an
edge from $x$ to $y$ if and only if $S(x) = y$ and $V(x) < V(y)$. This
restriction means that the graph must be a DAG.

To solve the problem, we must find a sink of the DAG, ie., a vertex that has no
outgoing edge. Since we require that $S(0^n) \ne 0^n$, we know that the DAG has
at least one vertex, and therefore it must also have at least one sink. This
places the problem in TFNP.

\paragraph{End of potential line.} 

We define a new problem called \EOPL, which merges the two definitions of \EOL
and \SOD into a single problem.

\begin{definition}[\EOPL]
\label{def:EOPL}
Given Boolean circuits $S,P : \Set{0,1}^n \to \Set{0,1}^n$ such that $P(0^n) =0^n\neq S(0^n)$ and a Boolean circuit $V: \Set{0,1}^n \to \Set{0,1,\dotsc,2^m - 1}$ such that $V(0^n) = 0$ find one of the following:
\begin{enumerate}[label=(R\arabic*)]
\itemsep0mm
\item \solnlabel{R1} A point $x \in \Set{0,1}^n$ such that $S(P(x)) \neq x \neq 0^n$ or $P(S(x)) \neq x$.
\item \solnlabel{R2} A point $x \in \Set{0,1}^n$ such that $x \neq S(x)$, $P(S(x)) = x$, and $V(S(x)) - V(x) \leq 0$.
\end{enumerate}
\end{definition}

This problem defines an exponentially large graph where each vertex has
in-degree and out-degree at most one (as in \EOL) that is also a DAG (as in
\SOD). An edge exists from $x$ to $y$ if and only if $S(x) = y$, $P(y) = x$, and
$V(x) < V(y)$. As in \SOD, only some bit-string encode vertices, and we adopt
the same idea that if $S(x) = x$ for some bit-string $x$, then $x$ does
\emph{not} encode a vertex.

So we have a single instance that is simultaneously an instance of \EOL and an
instance of \SOD. To solve the problem, it suffices to solve \emph{either} of
these problems. Solutions of type \solnref{R1} consist of vertices $x$ that are either
the end of a line, or the start of a line (excluding the case where $x = 0^n$).
Solutions of type \solnref{R2} consist of any point $x$ where the potential does not
strictly increase on the edge between $x$ and $S(x)$.

\paragraph{\bf The complexity class \EOPLc.}

We define the complexity class \EOPLc to consist of all problems that can be
reduced in polynomial time to \EOPL. By definition the problem lies in
$\PPADPLS$, since one can simply ignore solutions of type \solnref{R2} to obtain an \EOL
instance, and ignore solutions of type \solnref{R1} to obtain a \SOD instance.

In fact we are able to show the stronger result that $\EOPLc \subseteq \CLS$. To
do this, we reduce \EOPL to the problem \EOML, which was defined
by Hub{\'a}{\v{c}}ek and Yogev, who also showed that the problem lies in \CLS~\cite{hubavcek2017hardness}. 
The main difference between the two problems is that \EOML requires that the
potential increases by \emph{exactly} one along each edge. The reduction from
\EOML to \EOPL is straightforward. The other direction is more involved, and
requires us to insert new vertices into the instance. Specifically, if there is
an edge between a vertex $x$ and a vertex $y$, but $V(y) \ne V(x) + 1$, then we
need to insert a new chain of vertices of length $V(y) - V(x) - 1$ between $x$
and $y$, so that we
can ensure that the potential always increases by exactly one along each edge. 
The full details are given in Appendix~\ref{app:eoml2eopl}, where the following
theorem is proved.

\begin{theorem}
\label{thm:eoml2eopl}
\EOML and \EOPL are polynomial-time equivalent.
\end{theorem}

As we have mentioned, Hub{\'a}{\v{c}}ek and Yogev have shown that \EOML lies
in \CLS~\cite{hubavcek2017hardness}, so we get the following corollary.

\begin{corollary}
\label{cor:eoplinCLS}
$\EOPLc \subseteq \CLS$.
\end{corollary}


\paragraph{\bf Problems with unique solutions.}

The problems that we study in this paper all share a specific set of properties
that cause them to produce an interesting subclass of \EOPL instances. Each of
the problems that we study have a \emph{promise}, and if the promise is
satisfied the problem has a \emph{unique} solution. 

For example, in the Contraction problem, we are given a function $f : [0, 1]^d
\rightarrow [0, 1]^d$ that is
promised to be \emph{contracting}, meaning that $d(f(x), f(y)) \le c \cdot f(x,
y)$ for some positive constant $c < 1$ and some distance metric $d$. We cannot
efficiently check whether $f$ is actually contracting, but if it is, then
Banach's fixpoint theorem states that $f$ has a unique
fixpoint~\cite{Banach1922}.
If $f$ is not contracting,
then there will exist \emph{violations} that can be witnessed by short
certificates. For Contraction, a violation is any pair of points $x, y$ such
that $d(f(x), f(y)) > c \cdot f(x, y)$.

We can use violations to formulate the problem as a non-promise problem that
lies in TFNP. Specifically, if we ask for either a fixpoint or a violation of
contraction, then the contraction problem is total, because if there is no
fixpoint, then the contrapositive of Banach's theorem implies that there must
exist a violation of contraction.

\paragraph{\bf Unique End of Potential Line.}

When we place this type of problem in \EOPLc, we obtain an instance with extra
properties. Specifically, if the original problem has no violations, meaning
that the promise is satisfied, then the \EOPL instance will contain a
\emph{single} line that starts at $0^n$, and ends at the unique solution of the
original problem. This means that, if we ever find two distinct lines in our
\EOPL instance, then we immediately know that original instance fails to satisfy
the promise.

We define the following problem, which is intended to capture these properties.

\begin{definition}[\UEOPL]
\label{def:UEOPL}
Given Boolean circuits $S,P : \Set{0,1}^n \to \Set{0,1}^n$ such that $P(0^n) =0^n\neq S(0^n)$ and a Boolean circuit $V: \Set{0,1}^n \to \Set{0,1,\dotsc,2^m - 1}$ such that $V(0^n) = 0$ find one of the following:
\begin{enumerate}[label=(U\arabic*), wide=0pt]
\itemsep0mm
\item \solnlabel{U1} A point $x \in \Set{0,1}^n$ such that $P(S(x)) \neq x$.
\end{enumerate}
\begin{enumerate}[label=(UV\arabic*), wide=0pt, leftmargin=\parindent]
\item \solnlabel{UV1} A point $x \in \Set{0,1}^n$ such that $x \neq S(x)$, $P(S(x)) = x$, and $V(S(x)) - V(x) \leq 0$.

\item \solnlabel{UV2} A point $x \in \Set{0,1}^n$ such that $S(P(x)) \neq x \neq 0^n$.

\item \solnlabel{UV3} Two points $x, y \in \Set{0,1}^n$, such that $x \ne y$, $x \ne S(x)$, $y
\ne S(y)$, and either $V(x) = V(y)$ or  $V(x) < V(y) < S(x)$.

\end{enumerate}
\end{definition}

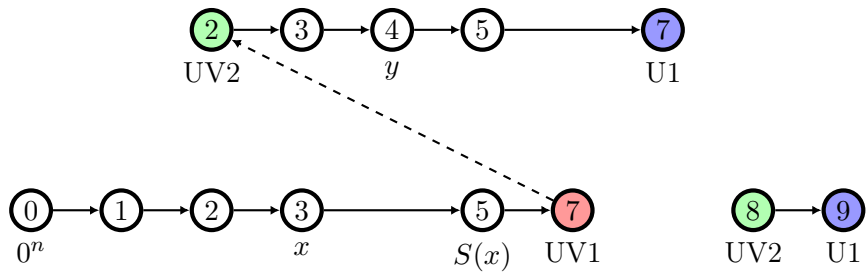
\begin{figure}[htb]
\begin{center}
\pgfdeclarelayer{background}
\pgfsetlayers{background,main}

\tikzstyle{vertex}=[draw,ultra thick,circle,fill=white,minimum size=15pt,inner sep=0pt]

\tikzstyle{red vertex} = [vertex, fill=red!40]
\tikzstyle{blue vertex} = [vertex, fill=blue!40]
\tikzstyle{green vertex} = [vertex, fill=green!30]
\tikzstyle{edge} = [draw,thick,->]
\tikzstyle{weight} = []
\tikzstyle{red edge} = [draw,line width=5pt,-,red!50]
\tikzstyle{blue edge} = [draw,line width=5pt,-,blue!50]
\tikzstyle{green edge} = [draw,line width=5pt,-,green!50]
\tikzstyle{ignored edge} = [draw,line width=5pt,-,black!20]

\begin{tikzpicture}[scale=1.2,swap]

\node[vertex]      (a0) at (0,0) [label=below:$0^n$]     {$0$};
\node[vertex]      (a1) at (1,0)                         {$1$};
\node[vertex]      (a2) at (2,0)                         {$2$};
\node[vertex]      (a3) at (3,0) [label=below:$x$]       {$3$};
\node[vertex]      (a4) at (5,0) [label=below:$S(x)$]    {$5$};
\node[red vertex]  (a5) at (6,0) [label=below:UV1]       {$7$};

\node[green vertex](b1) at (8,0) [label=below:UV2]       {$8$};
\node[blue vertex] (b2) at (9,0) [label=below:U1]        {$9$};

\node[green vertex](c1) at (2,2) [label=below:UV2]       {$2$};
\node[vertex]      (c2) at (3,2)                         {$3$};
\node[vertex]      (c3) at (4,2) [label=below:$y$]       {$4$};
\node[vertex]      (c4) at (5,2)                         {$5$};
\node[blue vertex] (c5) at (7,2) [label=below:U1]       {$7$};
    
\path[edge] (a0) -- (a1);
\path[edge] (a1) -- (a2);
\path[edge] (a2) -- (a3);
\path[edge] (a3) -- (a4);
\path[edge] (a4) -- (a5);
\path[edge] (b1) -- (b2);
\path[edge] (c1) -- (c2);
\path[edge] (c2) -- (c3);
\path[edge] (c3) -- (c4);
\path[edge] (c4) -- (c5);
\path[edge,dashed] (a5) -- (c1);


\end{tikzpicture}
\end{center}
\caption{Example of solutions, including violations, for \UEOPL. 
This figure should be viewed in color.
In this example we have 3 lines.
The \emph{main line} starts at $0^n$ and ends with a UV1 solution, when the successor of the 
red vertex, which has potential $7$, has lower potential, $2$. 
That successor is thus the start of another line, which is a UV2 solution.
There is a final line of length one to the bottom right, which also has a start, which is another 
UV2 solution.
This line does not intersect either of the other two lines in terms of the ranges of their potential values, so 
does not contribute to UV3 solutions.
The main line and top line do intersect in terms of the ranges of their potential values and they contribute
many UV3 solutions. 
We highlight one on the diagram with $x$, $S(x)$, and $y$, such that $V(x) < V(y) < V(S(x))$.
There are two U1 solutions at the end of the top and bottom right lines.
Finally, we note that if there were no violations, then there must be one line and a single U1 solution at 
the end of the main line.
}
\label{fig:ueopl}
\end{figure}

We split the solutions into two types: proper solutions and violations.
Solutions of type \solnref{U1} encode the end of any line, which are the proper solutions
to the problem.  
There are three types of violation solution.
Violations of type \solnref{UV1} are
vertices at which the potential fails to strictly increase. 
Violations of type \solnref{UV2} ask for the start of any line, other
than the vertex $0^n$. Clearly, if there are two sources in the graph, then
there are two lines. 

Violations of type \solnref{UV3} give another witness that the instance contains more than
one line. This is encoded by a pair of vertices $x$ and $y$, with either $V(x) =
V(y)$, or with the property
that the potential of $y$ lies between the potential of $x$ and $S(x)$. Since we
require the potential to strictly increase along every edge of a line, this
means that $y$ cannot lie on the same line as $x$, since all vertices before $x$
in the line have potential strictly less than $V(x)$, while all vertices after
$S(x)$ have potential strictly greater than $V(S(x))$. 


We remark that \solnref{UV2} by itself already captures the property ``there is a unique
line'', since if a second line cannot start, then it cannot exist. So why do we
insist on the extra type of violation given by \solnref{UV3}? Violations of type \solnref{UV3} allow
us to solve the problem immediately if we ever detect the existence of multiple
lines. Note that this is not the case if we only have solutions of type \solnref{UV2},
since we may find two vertices on two different lines, but both of them may be
exponentially many steps away from the start of their respective lines.




By adding \solnref{UV3} solutions, we make the problem easier than \EOPL (note that without
UV3, the problem is actually the same as \EOPL).
This means that problems that can be reduced to \UEOPL have the very special
property that, if at any point you detect the existence of multiple lines,
either through the start of a second line, or through a violation in \solnref{UV3}, then
you \emph{immediately} get a violation in the original problem without any extra
effort. All of the problems that we study in this paper share this property.

\paragraph{\bf The complexity class \UniqueEOPLc.}

We define the complexity class \UniqueEOPLc to be the class of problems that can be
reduced in polynomial time to \UEOPL. We note that $\UniqueEOPLc \subseteq \EOPLc$ is
trivial, since the problem remains total even if we disallow solutions of type
UV3.

For each of our problems, it is also interesting to consider the complexity of
the promise variant, in which it is guaranteed via a promise that no violations
exist. We define \PUEOPL to be the promise version of \UEOPL in which $0^n$ is
the only start of a line (and hence there are no solutions that are 
type \solnref{UV2} or \solnref{UV3}).
We define the complexity class \PUEOPLc to be the class of promise problems
that can be reduced in polynomial time to \PUEOPL.

\paragraph{\bf Promise-preserving reductions.}

The problem \UEOPL has the interesting property that, if it is promised that
there are no violation solutions, then there must be a unique solution. All of
the problems that we study in this paper share this property, and indeed when
when we reduce them to \UEOPL, the resulting instance will have a unique line
whenever the original problem has no violation solutions. 

We formalise this by defining the concept of a \defineterm{promise-preserving}
reduction. This is a reduction between two problems A and B, both of which have
proper solutions and violation solutions. The reduction is promise-preserving
if, when it is promised that A has no violations, then the resulting instance of
B also has no violations. Hence, if we reduce a problem to \UEOPL via a chain of
promise-preserving reductions, and we know that there are no violations in the
original problem, then there is a unique line ending at the unique proper
solution in the \UEOPL instance. 

Note that this is more restrictive than a general reduction. We could in
principle produce a reduction that took an instance of A, where it is promised that there are
no violations, and produce an instance of B that sometimes contains violations.
By using promise-preserving reductions, we are showing that our problems have
the natural properties that one would expect for a problem in \UEOPLc.
Specifically, that the promise version has a unique solution, and that this can
be found by following the unique line in the \UEOPL instance.

One added bonus is that, if we show that a problem is in \UEOPLc via a chain of
promise-preserving reductions, then we automatically get that the promise
version of that problem, where it is promised that there are no violations, lies
in \PUEOPLc. Moreover, if we show that a problem is \UEOPLc-complete via a
promise-preserving reduction, then this also implies that the promise version of
that problem is \PUEOPLc-complete.




\section{One-Permutation Discrete Contraction}

The \defineterm{One-Permutation Discrete Contraction} (\DCM) problem will play a
crucial role in our results. We will show that the problem lies in \UniqueEOPLc,
and we will then reduce both \LCM and \GUSO to \DCM, thereby showing that those
problems also lie in \UniqueEOPLc. We will also show that \UEOPL can be reduced
to \DCM, making this problem the first example of a non-trivial
\UniqueEOPLc-complete problem.

\paragraph{\bf Direction functions.} 

The \OPDC can be seen as a discrete variant of the continuous contraction
problem. Recall that a contraction map is a function $f : [0, 1]^n \rightarrow
[0, 1]^d$ that is contracting under a metric $d$, i.e., $d(f(x), f(y)) \le c
\cdot f(x, y)$ for all $x, y \in [0, 1]^d$ and some constant $c$
satisfying $0 < c < 1$. 



We will discretize the space by  overlaying a grid of points on the $[0, 1]^d$
cube. Let $[k]$ denote the set $\{0, 1, \dots, k\}$. Given a tuple of grid
widths $(k_1, k_2, \dots, k_d)$, we define the set 
\begin{equation*}
P(k_1, k_2, \dots, k_d) = [k_1] \times [k_2] \times \dots \times [k_d].
\end{equation*}
We will refer to $P(k_1, k_2, \dots, k_d)$ simply as $P$ when the grid
widths are clear from the context.
Note that each point $p \in P$ is a tuple
$(p_1, p_2, \dots, p_d)$, where $p_i$ is an integer between $0$ and $k_i$, and
this point maps onto the point $(p_1/k_1, p_2/k_2, \dots, p_d/k_d) \in [0,
1]^d$.


Instead of a single function $f$, in the discrete version of the problem we will
use a family of \defineterm{direction functions} over the grid $P$. For each dimension $i \le d$,
we have function $D_i : P \rightarrow \{\up, \down, \zero\}$. Intuitively, the
natural reduction from a contraction map $f$ to a family of direction functions
would, for each point $p \in P$ and each dimension $i \le d$ set:
\begin{itemize}
\item 
$D_i(p) = \up$ whenever $f(p)_i > p_i$,
\item 
$D_i(p) = \down$ whenever $f(p)_i < p_i$, and 
\item 
$D_i(p) = \zero$ whenever $f(p)_i = p_i$.
\end{itemize}
In other words, the function $D_i$ simply outputs whether $f(p)$ moves up, down,
or not at all in dimension~$i$. So a point $p \in P$ with $D_i(p) = \zero$ for
all $i$ would correspond to the fixpoint of $f$. Note, however, that the grid
may not actually contain the fixpoint of $f$, and so there may be no point $p$
satisfying $D_i(p) = \zero$ for all $i$.

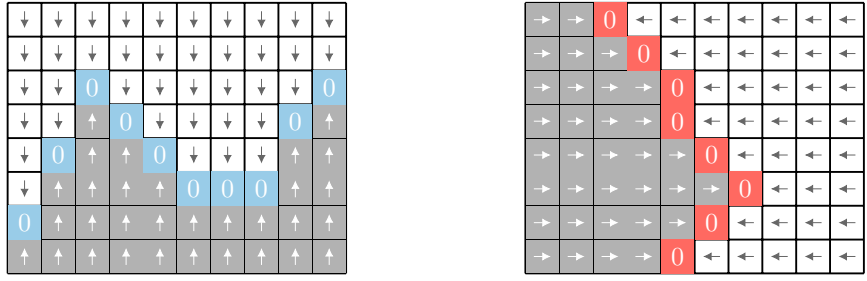
\begin{figure}
\begin{center}
\scalebox{0.9}{
\begin{tikzpicture}[scale=0.5]
\draw[thick,step=1,xshift=0.5cm,yshift=0.5cm] (0,0) grid (10,8); 
\foreach \x/\y in {1/2,2/4,3/6,4/5,5/4,6/3,7/3,8/3,9/5,10/6}
{
	\updncol{\x}{\y}
}
\end{tikzpicture}
}
\hskip 2cm
\scalebox{0.9}{
\begin{tikzpicture}[scale=0.5]
\draw[thick,step=1,xshift=0.5cm,yshift=0.5cm] (0,0) grid (10,8); 
\foreach \x/\y in {5/1,6/2,7/3,6/4,5/5,5/6,4/7,3/8}
{
	\leftrightcol{\x}{\y}
}
\end{tikzpicture}
}
\end{center}
\caption{Left: A direction function for the up/down dimension. Right: A
direction function for the left/right dimension.
This figure should be viewed in color.
}
\label{fig:direction}
\end{figure}

\paragraph{\bf A two-dimensional example.}

To illustrate this definition, consider the two-dimensional instance given in
Figure~\ref{fig:direction}, which we will use as a running example. It shows two direction functions: the figure on the left shows a direction function for the up-down
dimension, which we will call dimension $1$ and illustrate using the color blue.
The figure on the right shows a direction function for the left-right dimension,
which we will call dimension $2$ and illustrate using the color red. Each square
in the figures represents a point in the discretized space, and the value of the
direction function is shown inside the box. Note that there is exactly one point
$p$ where $D_1(p) = D_2(p) = \zero$, which is the fixpoint that we seek.

\paragraph{\bf Slices.}

We will frequently refer to subsets of $P$ in which some of the dimensions have
been fixed. A \defineterm{slice} will be represented as a tuple $(s_1, s_2,
\dots, s_d)$, where each $s_i$ is either
\begin{itemize}
\item a number in $[0, 1]$, which indicates that dimension $i$ should be fixed
to $s_i$, or
\item the special symbol $\blank$, which indicates that dimension $i$ is free to
vary.
\end{itemize}
We define $\Slice_d = \Paren{[0,1]\cup \Set{\blank}}^d$ to be the set of all
possible slices in dimension $d$. Given a slice $s \in \Slice_d$, we define $P_s
\subseteq P$ to be the set of points in that slice, ie., $P_s$ contains every
point $p \in P$ such that $p_i = s_i$ whenever $s_i \ne \blank$. We'll say that a slice $\ss' \in \Slice_d$ is a sub-slice of a slice $\ss \in \Slice_d$ if $s_j \neq \blank \implies s'_j = s_j$ for all $j\in [d]$.

An \defineterm{$i$-slice} is a slice $\ss$ for which $s_j = \blank$ for all $j \le
i$, and $s_j \ne \blank$ for all $j > i$. In other words, all dimensions up to
and including dimension $i$ are allowed to vary, while all other dimensions are
fixed.

In our two-dimensional example, there are three types of $i$-slices. There is
one $2$-slice: the slice $(\blank, \blank)$ that contains every point. For each
$x$, there is a $1$-slice $(\blank, x)$, which restricts the left/right
dimension to the value $n$. For each pair $x, y$ there is a $0$-slice $(y, x)$,
which contains only the exact point corresponding to $x$ and $y$.

\paragraph{\bf Discrete contraction maps.}

We can now define a one-permutation discrete contraction map. We say that a point $p \in P_s$ in some slice $s$
is a \defineterm{fixpoint} of $s$ if $D_i(p) = \zero$ for all dimensions $i$
where $s_i = \blank$. The following definition captures the promise version of
the problem, and we will later give a non-promise version by formulating
appropriate violations.

\begin{definition}[One-Permutation Discrete Contraction Map]
Let $P$ be a grid of points over $[0,1]^d$ and let $\mathcal{D} = (D_i)_{i = 1, \dots, d}$ be a family of direction functions over $P$.
We say that $\mathcal{D}$ and $P$ form a one-permutation discrete contraction
map if, for every $i$-slice $s$, the following conditions hold.
\begin{enumerate}
\itemsep1mm
\item There is a unique fixpoint of $s$.

\item Let $s' \in \Slice_d$ be a sub-slice of $s$ where some coordinate $i$ for which
$s_i = \blank$ has been fixed
to a value, and all other coordinates are unchanged. If $q$ is the unique fixpoint of $s$, and $p$ is the unique
fixpoint of $s'$, then
\begin{itemize}
\itemsep1mm
\item if $p_i < q_i$, then $D_i(p) = \up$, and
\item if $p_i > q_i$, then $D_i(p) = \down$.
\end{itemize}
\end{enumerate}
\end{definition}
The first condition specifies that each $i$-slice must have a unique fixed
point. Since the slice $(\blank, \blank, \dots, \blank)$ is an $i$-slice, this
implies that the full problem also has a unique fixpoint. 


The second condition is a more technical condition. It tells us that if we have
found the unique fixpoint $p$ of the $(i+1)$-slice $s'$, and if this point is
not the unique fixpoint of the $i$-slice $s$, then the direction function
$D_i(p)$ tells us which way to walk to find the unique fixpoint of $s$. This is
a crucial property that we will use in our reduction from \OPDC to \UEOPL, and
in our algorithms for contraction maps.

In our two-dimensional example, the first condition requires that every slice
$(\blank, x)$ has a unique fixpoint, and this corresponds to saying that for
every fixed slice of the left/right dimension, there is a unique blue point that
is zero. The second
condition requires that, if we are at some blue zero, then the red direction
function at that point tells us the direction of the overall fixpoint. 
It can be seen that our example satisfies both of these requirements. 

Note that both properties only consider $i$-slices. In the continuous
contraction map problem with an $L_p$ norm distance metric, \emph{every} slice
has a unique fixpoint, and so one may expect a discrete version of contraction
to share this property. The problem is that the second property is very
difficult to prove. Indeed, when we reduce \LCM to \OPDC in
Section~\ref{sec:lcm2eopl}, we must carefully choose the grid size to ensure
that both the first and second properties hold. In fact, our choice of grid size
for dimension $i$ will depend on the grid size of dimension $i+1$, which is why
our definition only considers $i$-slices.

The name \emph{One-Permutation} Discrete Contraction was chosen to emphasize
this fact. The $i$-slices correspond to restricting dimensions in order,
starting with dimension $d$. Since the order of the dimensions is arbitrary, we
could have chosen any permutation of the dimensions, but we must choose
\emph{one} of these permutations to define the problem.

\paragraph{\bf The OPDC problem.}

The OPDC problem is as follows: given a discrete contraction map $\mathcal{D} =
(D_i(p))_{i=1,\dots,d}$, find the unique point $p$ such that $D_i(p) = \zero$
for all $i$. Note that we cannot efficiently verify whether $\mathcal{D}$ is
actually a one-permutation discrete contraction map.

So, the OPDC problem is a promise problem, and we will formulate a total variant
of it that uses a set of violations to cover the cases where $\mathcal{D}$ fails
to be a discrete contraction map.

\begin{definition}[\OPDC]
\label{def:OPDC}
Given a tuple $(k_1, k_2, \dots, k_d)$ and circuits $(D_i(p))_{i=1,\dots,d}$,
where each circuit $D_i : P(k_1, k_2, \dots, k_d) \rightarrow
\{\up, \down, \zero\}$, find one of the following
\begin{enumerate}[label=(O\arabic*)]
\item \solnlabel{O1} A point $p \in P$ such that $D_i(p) = \zero$ for all $i$.
\end{enumerate}
\begin{enumerate}[label=(OV\arabic*),wide=0pt,leftmargin=\parindent]
\item \solnlabel{OV1} An $i$-slice $s$ and two points $p, q \in P_s$ with $p \ne q$ such that 
 $D_j(p) = D_j(q) = \zero$ for all $j \le i$.
\item \solnlabel{OV2} An $i$-slice $s$ and two points $p, q \in P_s$ such that
\begin{itemize}
\item $D_j(p) = D_j(q) = \zero$ for all $j < i$, 
\item $p_i = q_i + 1$, and
\item $D_i(p) = \down$ and $D_i(q) = \up$.
\end{itemize}
\item \solnlabel{OV3} An $i$-slice $s$ and a point $p \in P_s$ such that
\begin{itemize}
\item $D_j(p) = D_j(q) = \zero$ for all $j < i$, and either
\item $p_i = 0$ and $D_i(p) = \down$, or
\item $p_i = k_i$ and $D_i(p) = \up$.
\end{itemize}
\end{enumerate}
\end{definition}
\noindent Solution type \solnref{O1} encodes a fixpoint, which is the proper solution of the
discrete contraction map, while solution types \solnref{OV1} through \solnref{OV3} encode violations
of the discrete contraction map property.

Solution type \solnref{OV1} witnesses a violation of the first property of a discrete
contraction map, namely that each $i$-slice should have a unique fixpoint. A
solution of type \solnref{OV1} gives two different points $p$ and $q$ in the same $i$-slice
that are both fixpoints of that slice.

Solutions of type \solnref{OV2} witness violations of the first and second properties of a
discrete contraction map. In these solutions we have two points $p$ and $q$ that
are both fixpoints of their respective $(i-1)$-slices and are directly adjacent
in an $i$-slice $s$. If there is a fixpoint $r$ of the slice $s$, then this
witnesses a violation of the second property of a discrete contraction map,
which states that $D_i(p)$ and $D_i(q)$ should both point towards $r$, and
clearly one of them does not. On the other hand, if slice $s$ has no fixpoint,
then $p$ and $q$ also witness this fact, since the fixpoint should be in-between
$p$ and $q$, which is not possible.

Solutions of type \solnref{OV3} consist of a point $p$ that is a fixpoint of its
$(i-1)$-slice but $D_i(p)$ points outside the boundary of the grid. 
These are clear violations of the second property, since $D_i(p)$ should point
towards the fixpoint of the $i$-slice containing $p$, but that fixpoint cannot
be outside the grid.

It is perhaps not immediately obvious that \OPDC is a total problem. Ultimately
we will prove this fact in the next section by providing a promise-preserving
reducing from \OPDC to \UEOPL. This will give us a proof of totality, and will
also prove that, if the discrete contraction map has no violations, then it does
indeed have a unique solution.

\subsection{One-Permutation Discrete Contraction is in \UEOPLc}
\label{sec:opdc2ufeopl}

In this section, we will show that One-Permutation Discrete Contraction lies in
\UEOPLc under promise-preserving reductions.

\paragraph{\bf UFEOPL.} 

Our reduction will make use of an intermediate problem that we call
\defineterm{unique forward EOPL}, which is a version of \UEOPL in which we only
have a successor circuit $S$, meaning that no predecessor circuit $P$ is given. 

\begin{definition}[\UFEOPL]
Given a Boolean circuits $S : \Set{0,1}^n \to \Set{0,1}^n$ such that $S(0^n) \ne
0^n$ and a Boolean circuit $V: \Set{0,1}^n \to \Set{0,1,\dotsc,2^m - 1}$ such that $V(0^n) = 0$ find one of the following:
\begin{enumerate}[label=(UF\arabic*), wide=0pt, leftmargin=\parindent]
\itemsep0mm
\item \solnlabel{UF1} A point $x \in \Set{0,1}^n$ such that $S(x) \neq x$ and either $S(S(x)) =
S(x)$
or $V(S(x)) \le V(x)$.
\end{enumerate}
\begin{enumerate}[label=(UFV\arabic*), wide=0pt, leftmargin=\parindent]


\item \solnlabel{UFV1} Two points $x, y \in \Set{0,1}^n$, such that $x \ne y$, $x \ne S(x)$, $y
\ne S(y)$, and either $V(x) = V(y)$ or  $V(x) < V(y) < V(S(x))$.

\end{enumerate}
\end{definition}

Without the predecessor circuit, this problem bears more resemblance to \SOD
than to \EOPL. As
in \SOD, a bit-string $x$ encodes a vertex if and only if $S(x) \ne x$,
and an edge exists between vertices $x$ and $y$ if and only if $S(x) = y$ and
$V(x) < V(y)$. The proper solution type \solnref{UF1} asks us to find a vertex that is a
sink of the DAG, just as before.

The difference lies in the violation solution type \solnref{UFV1}, which is the same as
violation type \solnref{UV3} of \UEOPL. It asks for two vertices $x$ and $y$ that either
have the same potential, or for which the potential of $y$ lies strictly between
the potential of $x$ and the potential of~$S(x)$. Note that this restriction
severely constrains a \SOD instance: if there are no violation solutions, then
the DAG must consist of a single line that starts at $0^n$, and ends at the
unique solution of type \solnref{UF1}. So in this sense, the problem really does capture
instances of \UEOPL that lack a predecessor circuit.

The \UFEOPL problem will play a crucial role in our reduction. We will reduce
\OPDC to it, and we will then reduce it to \UEOPL.

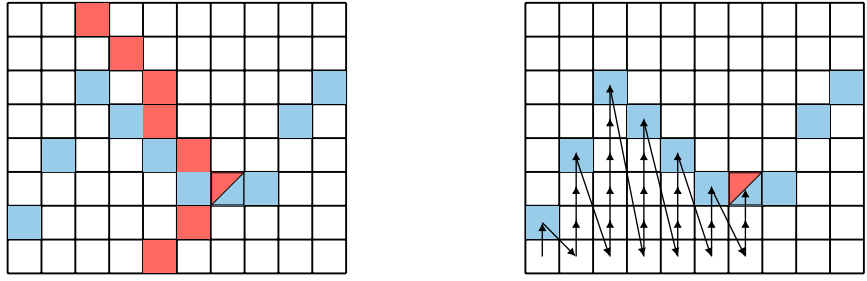
\begin{figure}
\begin{center}
\scalebox{0.9}{
\begin{tikzpicture}[scale=0.5]
\draw[thick,step=1,xshift=0.5cm,yshift=0.5cm] (0,0) grid (10,8); 
\foreach \x/\y in {1/2,2/4,3/6,4/5,5/4,6/3,7/3,8/3,9/5,10/6}
{
	\mysquare{\x}{\y}{pastelblue}
}
\foreach \x/\y in {5/1,6/2,6/4,5/5,5/6,4/7,3/8}
{
	\mysquare{\x}{\y}{pastelred}
}

\mytwocoloursquare{7}{3}

\end{tikzpicture}
}
\hskip 2cm
\scalebox{0.9}{
\begin{tikzpicture}[scale=0.5]
\draw[thick,step=1,xshift=0.5cm,yshift=0.5cm] (0,0) grid (10,8); 
\foreach \x/\y in {1/2,2/4,3/6,4/5,5/4,6/3,7/3,8/3,9/5,10/6}
{
	\mysquare{\x}{\y}{pastelblue}
	
}
\mytwocoloursquare{7}{3}

\patharrow{1}{1}{1}{2}

\patharrow{1}{2}{1.99}{1}

\patharrow{2}{1}{2}{2.1}
\patharrow{2}{2}{2}{3.1}
\patharrow{2}{3}{2}{4.1}

\patharrow{2}{3.92}{2.99}{1}

\patharrow{3}{1}{3}{2.1}
\patharrow{3}{2}{3}{3.1}
\patharrow{3}{3}{3}{4.1}
\patharrow{3}{4}{3}{5.1}
\patharrow{3}{5}{3}{6.1}

\patharrow{3}{6}{3.99}{1}

\patharrow{4}{1}{4}{2.1}
\patharrow{4}{2}{4}{3.1}
\patharrow{4}{3}{4}{4.1}
\patharrow{4}{4}{4}{5.1}

\patharrow{4}{5}{4.99}{1}

\patharrow{5}{1}{5}{2.1}
\patharrow{5}{2}{5}{3.1}
\patharrow{5}{3}{5}{4.1}

\patharrow{5}{4}{5.99}{1}

\patharrow{6}{1}{6}{2.1}
\patharrow{6}{2}{6}{3.1}

\patharrow{6}{3}{6.995}{1}

\patharrow{7}{1}{7}{2.1}
\patharrow{7}{2}{7}{3}

\end{tikzpicture}
}
\end{center}
\caption{Left: The red and blue surfaces. Right: the path that we
follow.
This figure should be viewed in color.
}
\label{fig:rbsurface}
\end{figure}
\paragraph{An illustration of the reduction.}

Before we discuss the formal definition of the construction, we first give some
intuition by describing the reduction for the two-dimensional example shown in
Figure~\ref{fig:direction}.


The reduction uses the notion of a \emph{surface}. On the left side in
Figure~\ref{fig:rbsurface}, we have overlaid the surfaces of the two direction
functions from Figure~\ref{fig:direction}. The surface of a direction function
$D_i$ is exactly the set of points $p \in P$ such that $D_i(p) = \zero$. The
fixpoint $p$ that we seek has $D_i(p) = \zero$ for all dimensions $i$, and so it
lies at the intersection of these surfaces.

To reach the overall fixpoint, we walk along a path starting from the
bottom-left corner, which is shown on the right-hand side of
Figure~\ref{fig:rbsurface}. The path begins by walking upwards until it finds
the blue surface. Once it has found the blue surface, it then there are two
possibilities:
either we have found the overall fixpoint, in which case the line ends, or
we have not found the overall fixpoint, and the red direction function
tells us that the direction of the overall fixpoint is to the right.

If we have not found the overall fixpoint, then we move one step to the right,
go back to the bottom of the diagram, and start walking upwards again. We keep
repeating this until we find the overall fixpoint. This procedure gives us the
line shown on the right-hand side of Figure~\ref{fig:rbsurface}.

\paragraph{\bf The potential.}
How do we define a potential for this line? Observe that the dimension-two 
coordinates of the points on the line are weakly monotone, meaning that the line
never moves to the left. Furthermore, for any
dimension-two slice (meaning any slice in which the left/right coordinate is
fixed), the dimension-one coordinate is monotonically increasing. So, if $p =
(p_1, p_2)$ denotes any point on the line, 
if $k$ denotes
the maximum coordinate in either dimension, then the function
\begin{equation*}
V(p_1, p_2) = k \cdot p_2 + p_1 
\end{equation*}
is a function that monotonically increases along the line, which we can use as a
potential function.


\smallskip

\paragraph{\bf Uniqueness.}

To provide a promise-preserving reduction to UFEOPL, we must argue that the line
is unique whenever the OPDC instance has no violations. Here we must carefully
define what exactly a vertex on the line actually is, to ensure that no other
line can exist. Specifically, we must be careful that only points that are to
the left of the fixpoint are actually on the line, and that no ``false'' line
exists to the right of the fixpoint.

Here we rely on the following fact: if the line visits a point with coordinate
$x$ in dimension $2$, then it must have visited the point $p$ on the blue
surface in the slice defined by $x-1$. Moreover, for that point $p$ we must have
$D_2(p) = \up$, which means that it is to the left of the overall fixpoint.

Using this fact, each vertex on our line will be a pair $(p, q)$, where $p$ is
the current point that we are visiting, and $q$ is either
\begin{itemize}
\itemsep0.5mm
\item the symbol $\vblank$, indicating that we 
are still in the first column of points, and we have never visited a point on
the blue surface, or
\item a point $q$ that is on the blue surface, that satisfies $q_2 = p_2 - 1$
and $D_2(q) = \up$.
\end{itemize}
Hence the point $q$ is always the last point that we visited on the blue
surface, which provides a witness that we have not yet walked past the overall
fixpoint.


When we finish walking up a column of points, and find the point on the blue
surface, we overwrite $q$ with the new point that we have found. This step is
the reason why only a successor circuit can be given for the line, since the
value that is overwritten cannot easily be computed by a predecessor circuit.

\paragraph{\bf Violations.}

Our two-dimensional OPDC example does not contain any violations, but our
reduction can still handle all possible violations in the OPDC instance. At a
high level, there are two possible ways in which the reduction can go wrong if
there are violations.

\begin{enumerate}
\item It is possible, that as we walk upwards in some column, we do not find
a fixpoint, and our line will get stuck. This creates an end of line solution of
type \solnref{UF1}, which
must be mapped back to an OPDC violation. In our two-dimensional example, this
case corresponds to a column of points in which there is no point on the blue surface. However, if there
is no point on the blue surface, then
we will either
\begin{itemize}
\item find two adjacent points $p$ and $q$ in that column with
$D_1(p) = \up$ and $D_2(p) = \down$, which is a solution of type \solnref{OV2}, or

\item find a point $p$ at the top of the column with $D_1(p) = \up$, or a
point $q$ at the bottom of the column with $D_1(q) = \down$. Both of these are
solutions of type \solnref{OV3}.
\end{itemize}
There is also the similar case where we walk all the way to the right without
finding an overall fixpoint, in which case we will find a point $p$ on the
right-hand boundary that satisfies $D_1(p) = \zero$ and $D_2(p) = \up$, which is
a solution of type \solnref{OV3}.

\item The other possibility is that there may be more than one point on the blue
surface in some of the columns. This will inevitably lead to multiple lines,
since if $q$ and $q'$ are both points on the blue surface in some column, and
$p$ is some point in the column to the right of $p$ and $q$, then
$(p, q)$ and $(p, q')$ will both be valid vertices on two different lines. 

These can show up as violations of type \solnref{UFV1}, which we map back to
solutions of type \solnref{OV1}. Specifically, the points $p$ and $q$, which are given as
part of the two vertices, are both fixpoints of the same slice, which is exactly
what \solnref{OV1} asks for.

\end{enumerate}
We can argue that our reduction is promise-preserving. This is because violation
solutions in the UFEOPL instance are never mapped back to proper solutions of
the OPDC instance. This means that, if we promise that the OPDC instance has no
violations, then the resulting UFEOPL instance must also contain no violations.

\paragraph{\bf The full reduction.} 
Our reduction from \DCM to \UFEOPL generalizes the approach given above to $d$
dimensions.
We say that a point $p \in P$ is on the
\defineterm{$i$-surface} if $D_j(p) = \zero$ for all $j \le i$. In our two-dimensional
example we followed a line of points on the one-surface, in order to find a
point on the two-surface. In between any two points on the one-surface, we
followed a line of points on the zero-surface (every point is trivially on the
zero-surface).

Our line will visit a sequence of points on the $(d-1)$-surface in
order to find the point on the $d$-surface, which is the fixpoint. Between
any two points on the $(d-1)$-surface the line visits a sequence of points on
the $(d-2)$-surface, between any two points on the 
$(d-2)$-surface the line visits a sequence of points on the $(d-3)$-surface, and
so on.

The line will follow the same pattern that we laid out in two dimensions. Every
time we find a point on the $i$-surface, we remember it, increment our position
in dimension $i$ by $1$, and reset our coordinates back to $0$ for all
dimensions $j < i$. Hence, a vertex will be a tuple $(p_0, p_1, \dots, p_d)$,
where each $p_i$ is either
\begin{itemize}
\item the symbol $\vblank$, indicating that we have not yet encountered a point
on the $i$-surface, or
\item the most recent point on the $i$-surface that we have visited.
\end{itemize}
This is a generalization of the witnessing scheme that we saw in two-dimensions.

The potential is likewise generalized so that the potential of a point $p$ is
proportional to $\sum_{i=1}^d k^i p_i$, where again $k$ is some constant that is
larger than the grid size. This means that progress in dimension~$i$ dominates
progress in dimension $j$ whenever $j < i$, which allows the potential to
monotonically increase along the line. 

We are also able to deal with all possible violations, using the ideas that we
have described in two-dimensional case. Full details of this construction are
given in Appendix~\ref{app:dcm2ufeopl}, where the following lemma is proved.

\begin{lemma}
\label{lem:dcm2ufeopl}
There is a polynomial-time promise-preserving reduction from
\OPDC to \UFEOPL.
\end{lemma}

\paragraph{\bf From \UFEOPL to \UFEOPLp1.}

The next step of the reduction is to slightly modify the \UFEOPL instance, so
that the potential increases by exactly one in each step. Specifically we define
the following problem

\begin{definition}[\UFEOPLp1]
Given a Boolean circuits $S : \Set{0,1}^n \to \Set{0,1}^n$ such that $S(0^n) \ne
0^n$ and a Boolean circuit $V: \Set{0,1}^n \to \Set{0,1,\dotsc,2^m - 1}$ such that $V(0^n) = 0$ find one of the following:
\begin{enumerate}[label=(UFP\arabic*), wide=0pt, leftmargin=\parindent]
\itemsep0mm
\item \solnlabel{UFP1} A point $x \in \Set{0,1}^n$ such that $S(x) \neq x$ and either $S(S(x)) =
S(x)$
or $V(S(x)) \ne V(x) + 1$.
\end{enumerate}
\begin{enumerate}[label=(UFPV\arabic*), wide=0pt, leftmargin=\parindent]
\item \solnlabel{UFPV1} Two points $x, y \in \Set{0,1}^n$, such that $x \ne y$, $x \ne S(x)$, $y
\ne S(y)$, and $V(x) = V(y)$.

\end{enumerate}
\end{definition}

There are two differences between this problem and \UFEOPL. Firstly, an edge
exists between $x$ and $y$ if and only if $S(x) = y$, and $V(y) = V(x) + 1$,
and this is reflected in the modified definition of solution type \solnref{UFP1}.
Secondly, solution type \solnref{UFPV1} has been modified to only cover the case where we
have two vertices $x$ and $y$ that have the same potential. The case where
$V(x) < V(y) < V(S(x))$ is not covered, since in this setting this would imply
$V(S(x)) > V(x) + 1$, which already gives us a solution of type \solnref{UFP1}.

It is not difficult to reduce \UFEOPL to \UFEOPLp1, using the same techniques
that we used in the reduction from \EOPL to \EOML in
Theorem~\ref{thm:eoml2eopl}. This gives us the following lemma, which is proved
in Appendix~\ref{app:uf2ufp1}.

\begin{lemma}
\label{lem:uf2ufp1}
There is a polynomial-time promise-preserving reduction from \UFEOPL to
\UFEOPLp1.
\end{lemma}

\paragraph{\bf \UFEOPLp1 to \UEOPL.}

The final step of the proof is to reduce \UFEOPL to \UEOPL. For this, we are
able to build upon existing work. The following problem was introduced by
Bitansky et al~\cite{BPR15}. 

\begin{definition}[\SOVL~\cite{BPR15}]
The input to the problem consists of a starting vertex $x_s \in \{0, 1\}^n$, a
target integer $T \le 2^n$, and two boolean circuits $S : \{0, 1\}^n
\rightarrow \{0, 1\}^n$, $W : \{0, 1\}^n \times \{0, 1\}^n \rightarrow \{0, 1\}$.
It is promised that, for every vertex $x \in \{0, 1\}^n$, and every integer $i
\le T$, we have $W(x, i) = 1$ if and only if $x = S^{i-1}(x_s)$. The goal is
to find the vertex $x_f \in \{0, 1\}^n$ such that $W(x_f, T) = 1$.
\end{definition}

\SOVL is intuitively very similar to \UFEOPL. In this problem, a single line is
encoded, where as usual the vertices are encoded as bit-strings, and the circuit
$S$ gives the successor of each vertex. The difference in this problem is that
the circuit $W$ gives a way of verifying how far along the line a given vertex
is. Specifically, $W(x, i) = 1$ if and only if $x$ is the $i$th vertex on the
line. Note that this is inherently a promise problem, since if $W(x, i) = 1$ for
some $i$, we have no way of knowing whether $x$ is \emph{actually} $i$ steps
along the line, without walking all of those steps ourselves.

It was shown by Hub\'a\v{c}ek and Yogev~\cite{hubavcek2017hardness}  that \SOVL can be reduced in polynomial
time to \EOML, and hence also to \EOPL (via Theorem~\ref{thm:eoml2eopl}).
Moreover, the resulting \EOPL instance has a unique line, so this reduction also
reduces \SOVL to \UEOPL.
It is easy to reduce the promise version of \UFEOPLp1 to \SOVL, since we can
implement the circuit $W$ so that $W(x, i) = 1$ if and only if $V(x) = i$.

However, the existing work only deals with the promise problem. Our contribution
is to deal with violations. We show that, if one creates a \SOVL instance from
a \UFEOPLp1 instance, in the way described above, and applies the reduction of 
Hub\'a\v{c}ek and Yogev to produce a \UEOPL instance, then any violation can be
mapped back to a solution in the original \UFEOPLp1 instance. 
Hence, we show
the following lemma, whose full proof appears in Appendix~\ref{app:ufeopl2eopl}.

\begin{lemma}
\label{lem:ufeopl2eopl}
There is a polynomial-time promise-preserving reduction from 
\UFEOPL to \UniqueEOPLc. 
\end{lemma}

This completes the chain of promise-preserving reductions from \OPDC to \UEOPL.
Hence, we have shown the following theorem.

\begin{theorem}
\label{thm:opdc}
\OPDC is in \UniqueEOPLc under polynomial-time promise-preserving reductions.
\end{theorem}

\subsection{One-Permutation Discrete Contraction is \UEOPLc-hard}
\label{sec:opdc-complete}

In this section we will show that One-Permutation Discrete Contraction is
\UEOPLc-complete, by giving a hardness result. Specifically, we give a reduction
from \UEOPL to \OPDC.

\paragraph{\bf Modifying the line.}
The first step of the reduction is to slightly alter the \UEOPL instance.
Specifically, we would like to ensure the following two properties.
\begin{enumerate}
\item Every edge increases the potential by exactly one. That is, $V(S(x)) =
V(x) + 1$ for every vertex $x$.
\item The line has length exactly $2^n$ for some integer $n$.
More specifically, we ensure that if $x$ is the end of any line then we have
$V(x) = 2^n - 1$. The start of the line given in the problem has potential $0$, this
ensures that the length of that line is exactly $2^n$, although other lines may
be shorter.
\end{enumerate}
We have already developed a technique for ensuring the first property in the
reduction from EOPL to EOML in Theorem~\ref{thm:eoml2eopl}, which can be reused
here. Specifically, we introduce a
chain of dummy vertices between any pair of vertices $x$ and $y$ with $S(x) = y$
and $V(y) > V(x) + 1$. The second property can be ensured by choosing $n$ 
so that $2^n$ is larger than the longest possible line in the instance. Then, at
every vertex $x$ that is the end of a line, we introduce a chain of dummy
vertices
\begin{equation*}
(x, 0) \rightarrow (x, 1) \rightarrow \dots \rightarrow (x, 2^n - V(x) - 1),
\end{equation*}
where $V(x, i) = V(x) + i$. The vertex $e = (x, 2^n - V(x) - 1)$ will be the new end of
the line, and note that $V(e) = 2^n - 1$ as required. The full details
of this are given in Appendix~\ref{app:2k}, where the following lemma is shown.

\begin{lemma}
\label{lem:2k}
Given a \UEOPL instance $L = (S, P, V)$, there is a polynomial-time
promise-preserving reduction that produces a \UEOPL instance $L' = (S', P',
V')$, where
\begin{itemize}
\item For every $x$ and $y$ with $y = S(x)$ and $x = P(y)$ we have $V(y) = V(x)
+ 1$, and
\item There exists an integer $n$ such that, $x$ is a solution of type \solnref{U1} if and
only if we have $V(x) = 2^n-1$. 
\end{itemize}
\end{lemma}
For the remainder of this section, we will assume that we have a \UEOPL
instance $L = (S, P, V)$ that satisfies that two extra conditions given by
Lemma~\ref{lem:2k}. We will use $m$ to denote the bit-length of a vertex in $L$.

\paragraph{\bf The set of points.}

We create an \OPDC instance over a boolean hypercube with $m \cdot n$
dimensions, so our set of points is $P = \{0, 1\}^{mn}$. We will interpret each
point $p \in P$ as a tuple $(v_1, v_2, \dots, v_n)$, where each $v_i$ is a
bit-string of length $m$, meaning that each $v_i$ can represent a vertex in $L$.

To understand the reduction, it helps to consider the case where there is a
unique line in $L$. We know that this line has length exactly $2^n$. The
reduction repeatedly splits this line into two equal parts. 
\begin{itemize}
\item Let $L_1$ denote the first half of $L$, which contains all vertices $v$
with potential $0 \le V(v) \le 2^{n-1}-1$. 
\item Let $L_2$ denote the second half of $L$, which contains all vertices $v$
with potential $2^{n-1} \le V(v) \le 2^n-1$.
\end{itemize}
Observe that $L_1$ and $L_2$ both contain exactly $2^{n-1}$ vertices.

The idea is to embed $L_1$ and $L_2$ into different sub-cubes 
of the $\{0, 1\}^{mn}$ point space. The line that we embed will be determined by
the \emph{last} element of the tuple. Let $v$ be the vertex satisfying $V(v) =
2^{n-1}$, meaning $v$ is the first element of $L_2$. 
\begin{itemize}
\item We embed $L_2$ into the sub-cube $(\blank, \blank, \dots,
\blank, v)$.
\item We embed a copy of $L_1$ into each sub-cube $(\blank, \blank, \dots,
\blank, u)$ with $u \ne v$.
\end{itemize}
Note that this means that we embed a single copy of $L_2$, but many copies of
$L_1$. Specifically, there are $2^m$ possibilities for the final element of the
tuple. One of these corresponds to the sub-cube containing $L_2$, while $2^m-1$
of them contain a copy of $L_1$.

The construction is recursive. So we split $L_2$ into two lines $L_{2,
1}$ and $L_{2,2}$, each containing half of the vertices of $L_2$. If $w$ is the
vertex satisfying $V(w) = 2^{n-1} + 2^{n-2}$, which is the first vertex of
$L_{2,2}$, then
we embed a copy of $L_{2, 2}$ into the sub-cube $(\blank, \blank, \dots,
\blank, w, v)$, where $v$ is the same vertex that we used above, and
we embed a copy of $L_{2, 1}$ into each sub-cube $(\blank, \blank, \dots,
\blank, u, v)$, where $u \ne w$.
Likewise $L_1$ is split into two, and embedded into the sub-cubes of $(\blank,
\blank, \dots, \blank, u)$ whenever $u \ne v$.

Given a vertex $(v_1, v_2, \dots, v_n)$, we can view the bit-string $v_n$ as
choosing either $L_1$ or $L_2$, based on whether $V(v_n) = 2^{n-1}$. Once that
decision has been made, we can then view $v_{n-1}$ as choosing one half of the
remaining line. 
Since the original line $L$ has length $2^n$, and we repeat this process $n$
times, this means that at the end of the process we will be left with a line
containing a single vertex. So in this way, a point $(v_1, v_2, \dots, v_n)$ is
a representation of some vertex in $L$, specifically the vertex that is left
after we repeatedly split the line according to the choices made by $v_n$
through $v_1$.

We can compute the vertex represented by any tuple in polynomial time. Moreover,
given a slice $(\blank, \blank, \dots, \blank, v_i, v_{i+1}, \dots, v_n)$ that
fixes elements $i$ through $n$ of the tuple, we can produce, in polynomial time, a
\UEOPL instance corresponding to the line that is embedded in that slice. This
is formalised in the following lemma, whose proof is given in
Appendix~\ref{app:points}.

\begin{lemma}
\label{lem:points}
There are polynomial-time algorithms for computing the following two functions.
\begin{itemize}
\item The function $\decode(v_1, v_2, \dots, v_n)$ which takes a point in $P$
and returns the corresponding vertex of $L$.

\item The function $\subline(v_i, v_{i+1}, \dots, v_n, L)$, which takes
bit-strings $v_i$ through $v_n$, representing the slice, $(\blank, \blank,
\dots, \blank, v_i, v_{i+1}, \dots, v_n)$, and the instance $L = (S, P, V)$, and
returns a new \UEOPL instance $L' = (S', P', V')$ which represents the instance
that is to be embedded into this slice.
\end{itemize}
\end{lemma}

Although we have described this construction in terms of a single line, the two
polynomial time algorithms given by Lemma~\ref{lem:points} are capable of
working with instances that contain multiple lines. In the case where there are
multiple lines, there may be two or more bit-strings $x$ and $y$ with $V(x) = V(y) =
2^{n-1}$. In that case, we will embed second-half instances into $(\blank,
\blank, \dots, x)$ and $(\blank, \blank, \dots, y)$. This is not a problem for
the functions $\decode$ and $\subline$, although this may lead to violations
in the resulting OPDC instance that we will need to deal with.

\paragraph{\bf The direction functions.}

The direction functions will carry out the embedding of the lines. Since our
space is $\{0, 1\}^{mn}$, we will need to define $m \cdot n$ direction
functions $D_1$ through $D_{mn}$. 

The direction functions $D_{m(n-1)+1}$ through $D_{mn}$ correspond to the bits
used to define $v_n$ in a point $(v_1, v_2, \dots, v_n)$. These direction
functions are used to implement the transition between the first and second half
of the line. For each point $p = (v_1, v_2, \dots, v_n)$ we define these
functions using the following algorithm.


\begin{enumerate}
\item 
In the case where $V(v_n) \ne 2^{n-1}$, meaning that $\decode(p)$ is a vertex in
the first half of the line, then there are two possibilities.

\begin{enumerate}
\item If $V(\decode(p)) = 2^{n-1} - 1$, meaning that $p$ is the last
vertex on the first half of the line, then
we orient the direction function of dimensions $(n-1) \cdot m$
through $m$ towards the bit-string given by $S(\decode(p))$. This captures the idea that once we reach the end of
the first half of the line, we should then move to the second half, and we do
this by moving towards the sub-cube $(\blank, \blank, \dots, \blank,
S(\decode(p))$. 
So for each $i$ in the range $m(n-1) + 1 \le i \le mn$ we define
\begin{equation*}
D_i(p) = \begin{cases}
\up & \text{if $p_i = 0$ and $S(\decode(p))_i = 1$,}\\
\down & \text{if $p_i = 1$ and $S(\decode(p))_i = 0$,}\\
\zero & \text{otherwise.}\\
\end{cases}
\end{equation*}

\item If the rule above does not apply, then we orient everything towards $0$.
Specifically, we set
\begin{equation*}
D_i(p) = \begin{cases}
\down & \text{if $p_i = 1$,}\\
\zero & \text{if $p_i = 0$.}
\end{cases}
\end{equation*}
This is an arbitrary choice: our reduction would work with any valid
direction rules in this case.
\end{enumerate}

\item If $V(v_n) = 2^{n-1}$, then we are in the second half of the line. In this
case we set $D_i(p) = \zero$ for all dimensions $i$ in the range $(n-1) \cdot m
\le i \le m$. This captures the idea that, once we have entered the second half
of the line, we should never leave it again.

\end{enumerate}
We use the same idea recursively to define direction functions for all
dimensions $i \le m(n-1)$. This gives us a family of polynomial-time computable
direction functions $\mathcal{D} = (D_i)_{i=1,\dots,mn}$. The full details can
be found in Appendix~\ref{app:direction}.

\paragraph{\bf The proof.}

We have now defined $P$ and $\mathcal{D}$, so we have an OPDC instance. We must
now argue that the reduction is correct. 
The intuitive idea is as follows. If we
are at some point $p \in P$, and $\decode(p) = v$ is a vertex that is not the
end of a line, then there is some direction function $D_i$ such that $D_i(p) \ne
\zero$. We can see this directly for the case where $V(\decode(p)) = 2^{n-1} -
1$, since the direction functions on dimensions 
$m(n-1) + 1$ through $mn$ will be oriented towards $S(V(\decode(p))$, and so $p$
will not be a solution. 

As we show in the proof, the same property holds for all
other vertices in the middle of the line. The end of the line will be a
solution, because it will be encoded by the point $p = (v_1, v_2, \dots, v_n)$,
where each $v_i$ is the first vertex on the second half of the line embedded in
the sub-cube. Our direction functions ensure that $D_j(p) = \zero$ for all $j$
for this point.

Our reduction must also deal with violations. Violations of type \solnref{OV3} are
impossible by construction. In violations of type \solnref{OV1} and \solnref{OV2} we have two points
$p$ and $q$ that are in the same $i$-slice. Here we specifically use the fact
that violations can only occur within $i$-slices. Note that an $i$-slice will
fix the last $mn - i$ bits of the tuple $(v_1, v_2, \dots, v_n)$, which means
that there will be an index $j$ such that all $v_l$ with $l > j$ are fixed. This
allows us to associate the slice with the line $L' = \subline(v_{j+1},
v_{j+2}, \dots, v_n)$, and we know that both $p$ and $q$ encode vertices of
$L'$. In both cases, we are able to recover two vertices in $L'$ that have the
same potential, and these vertices also have the same potential in $L$. So we
get a solution of type \solnref{UV3}. The details are rather involved, and we defer the
proof to Appendix~\ref{app:completeness}, where the following lemma is proved.

\begin{lemma}
\label{lem:completeness}
There is a polynomial-time promise-preserving reduction from \UEOPL to \OPDC.
\end{lemma}

Thus, we have shown the following theorem.

\begin{theorem}
\label{thm:opdc-completeness}
\OPDC is \UEOPLc-complete under promise-preserving reductions, even when the set
of points $P$ is a hypercube.
\end{theorem}

Since we have shown bidirectional promise-preserving reductions between \OPDC
and \UEOPL, we also get that the promise version of \OPDC is complete for
\PUEOPLc.

\section{UniqueEOPL containment results}

\subsection{Unique Sink Orientations}
\label{sec:uso}


\paragraph{\bf Unique sink orientations.}

Let $C = \{0, 1\}^n$ be an $n$-dimensional hypercube. An
\defineterm{orientation} of $C$ gives a direction to each edge of $C$. 
We formalise this as a function $\udir : C \rightarrow \{0, 1\}^n$, that 
assigns a bit-string to each vertex of~$C$, with the interpretation that the
$i$-th bit of the string gives an orientation of the edge in
dimension~$i$. More precisely, for each vertex $v \in C$ and each dimension $i$,
let $u$ be the vertex that is adjacent to $v$ in dimension $i$.
\begin{itemize}
\item If $\udir(v)_i = 0$ then the edge between $v$ and $u$ is oriented towards
$v$.
\item If $\udir(v)_i = 1$ then the edge between $v$ and $u$ is oriented towards
$u$.
\end{itemize}
Note that this definition does not insist that $v$ and $u$ agree on the
orientation of the edge between them, meaning that $\udir(v)_i$ and
$\udir(u)_i$ may orient the edge in opposite directions. However, this will be
a violation in our set up, and a proper orientation should be thought of as
always assigning a consistent direction to each edge.

A \defineterm{face} is a subset of $C$ in which some coordinates have been
fixed. This can be defined using the same notation that we used for slices in
OPDC. So a face $f = (f_1, f_2, \dots, f_n)$, where each $f_i$ is either $0$,
$1$, or $\blank$, and the sub-cube defined by $f$ contains every vertex $v \in
C$ such that $v_i = f_i$ whenever $f_i \ne \blank$.

A vertex $v \in C$ is a \defineterm{sink} if all of the edges of $v$ are
directed towards $v$, meaning that $\udir(v)_i = 0$ for all $i$. Given a face
$f$, a vertex $v$ is a \defineterm{sink of $f$} if it is the sink of the
sub-cube defined by $f$, meaning that $\udir(v)_i = 0$ whenever $f_i = \blank$.

A \defineterm{unique sink orientation} (USO) is an orientation in which
\emph{every} face has a unique sink. Since $f = (\blank, \blank, \dots, \blank)$
is a face, this also implies that the whole cube has a unique sink, and the USO
problem is to find the unique sink.

\paragraph{\bf Placing the problem in \TFNP.}

The USO property is quite restrictive, and there are many orientations that are
not USOs. Indeed, the overall cube may not have a sink at all, or it may have
multiple sinks, and this may also be the case for the other faces of the cube.
Fortunately, Szab{\'{o}} and Welzl have pointed out that if an orientation is
not a USO, then there is a succinct witness of this fact~\cite{SzaboW01}.

Let $v, u \in C$ be two distinct vertices. We have that $\udir$ is a USO if and
only if there exists some dimension $i$ such that $v_i \ne u_i$ and $\udir(v)_i
\ne \udir(u)_i$. Put another way, this means that if we restrict the
orientation only to the sub-cube defined by any two vertices $v$ and $u$, then
the orientations of $v$ and $u$ must be different on that sub-cube. If one uses
$\oplus$ to denote the XOR operation on binary strings, and $\cap$ to denote
the bit-wise and-operation, then
this condition can be written concisely as
\begin{equation*}
(v \oplus u) \cap (\udir(v) \oplus \udir(u)) \ne 0^n.
\end{equation*}
Note that this condition also ensures that the orientation is consistent, since
if $v$ and $u$ differ in only a single dimension, then the conditions states
that they must agree on the orientation of the edge between them.

We use this condition to formulate the USO problem as a problem in \TFNP.
\begin{definition}[\USO]
\label{def:uso}
Given an orientation function $\udir: \{0, 1\}^n \rightarrow \{0, 1\}^n \cup
\{\vblank\}$ find
one of the following.
\begin{enumerate}[label=(US\arabic*), wide=0pt]
\item \solnlabel{US1} A point $v \in \{0,1\}^n$ such that $\udir(v)_i = 0$ for all $i$.
\end{enumerate}
\begin{enumerate}[label=(USV\arabic*), wide=0pt]
\item \solnlabel{USV1} A point $v \in \{0,1\}^n$ such that $\udir(v) = \vblank$.
\item \solnlabel{USV2} Two points $v, u \in \{0,1\}^n$ such that $v \ne u$ and 
$(v \oplus u) \cap (\udir(v) \oplus \udir(u)) = 0^n$.
\end{enumerate}
\end{definition}

Note that this formulation of the problem allows the orientation function to
decline to give an orientation for some vertices, and this is indicated
by setting $\udir(v) = \vblank$. Any such vertex is a violation of type \solnref{USV1}.
While this adds nothing interesting to the USO problem, 
we will use this in Section~\ref{sec:PLCPtoEOPL} when we reduce P-LCP 
\emph{to} USO, since in some cases the reduction may
not be able to produce an orientation at a particular vertex.

Assuming that there are no violations of type \solnref{USV1}, it is easy to see that the
problem is total. This is because
every USO has a sink, giving a solution of type \solnref{US1}, while every orientation
that is not a USO has a violation of type \solnref{USV2}.

\paragraph{\bf Placing the problem in \UEOPLc.}

We show that the problem lies in \UEOPLc by providing a promise-preserving
reduction from \USO to \OPDC. The reduction is actually not too difficult,
because when the point set for the \OPDC instance is a hypercube, the \OPDC
problem can be viewed as a less restrictive variant of USO. Specifically, 
USO demands that \emph{every} face has a unique sink, while \OPDC only requires
that the $i$-slices should have unique sinks.

The reduction creates an OPDC instances on the same set of points, meaning that
$P = C$. The direction functions simply follow the orientation given by $\udir$.
Specifically, for each $v \in P$ and each dimension $i$ we define
\begin{equation*}
D_i(v) = \begin{cases}
\zero & \text{if $\udir(v)_i = 0$,} \\
\up & \text{if $\udir(v)_i = 1$ and $v_i = 0$,} \\
\down & \text{if $\udir(v)_i = 1$ and $v_i = 1$.}
\end{cases}
\end{equation*}
If $\udir(v) = \vblank$, then we instead set $D_i(v) = \zero$ for all $i$.

To prove that this is correct, we show that every solution of the OPDC instance
can be mapped back to a solution of the USO instance. Any fixpoint of the OPDC
instance satisfies $D_i(v) = \zero$ for all $i$, which can only occur if $v$ is
a sink, or if $\udir(v) = \vblank$. The violation solutions of OPDC can be used
to generate a pair of vertices that constitute a \solnref{USV2} violation. We defer the
details to Appendix~\ref{app:uso}, where the following lemma is proved.

\begin{lemma}
\label{lem:uso}
There is a polynomial-time promise-preserving reduction from \USO to \OPDC.
\end{lemma}

Thus we have shown the following theorem.

\begin{theorem}
\label{thm:uso}
\USO is in \UEOPLc under promise-preserving reductions.
\end{theorem}

This proves the following new facts about the complexity of finding the sink of a USO.

\begin{corollary}
\USO is in \PPAD, \PLS, and \CLS.
\end{corollary}

\subsection{Piecewise Linear Contraction Maps}
\label{sec:lcm2eopl}

In this section, we show that finding a fixpoint of a \defineterm{piecewise
linear} contraction map lies in \UEOPLc. Specifically, we study contraction maps
where the function $f$ is given as a \defineterm{\LinearFIXP circuit}, which is
an arithmetic circuit comprised of $\max, \min, +, -$, and $\mz$ (multiplication
by a constant) gates~\cite{EtessamiY10}. Hence, a \LinearFIXP circuit defines a
\emph{piecewise linear} function.

\paragraph{\bf Violations.}

Not every function $f$ is contracting, and the most obvious way to prove that
$f$ is not contracting is to give a pair of points $x$ and $y$ that satisfy 
$\|f(x) - f(y)\|_p > c \cdot \| x - y \|_p$, which directly witness the fact
that $f$ is not contracting.

However, when we discretize Contraction in order to to reduce it to OPDC, there
are certain situations in which we have a convincing proof that $f$ is not
contracting, but no apparent way to actually produce a violation of contraction.
In fact, the discretization itself is non-trivial, so we will explain that
first, and then define the type of violations that we will use.

\paragraph{\bf The reduction.}

We are given a function $f : [0, 1]^n \rightarrow [0, 1]^n$, that is purported
to be contracting with contraction factor $c$ in the $\ell_p$ norm. We will
produce an \OPDC instance by constructing the point set $P$,  and a family of
direction functions $\mathcal{D}$. 

The most complex step of the reduction is to produce an appropriate set of
points $P$ for the \OPDC instance. This means we need to choose integers $k_1$,
$k_2$, through $k_d$ in order to define the point set $P(k_1, k_2, \dots, k_d)$,
where we recall that this defines a grid of integers, where each dimension~$i$
can take values between $0$ and $k_i$.
We will describe the method for picking $k_1$ through $k_d$ after we have
specified the rest of the reduction.

The direction functions will simply follow the directions given by $f$.
Specifically, for every point $p \in P(k_1, k_2, \dots, k_d)$, let $p'$ be the
corresponding point in $[0, 1]^n$, meaning that $p'_i = p_i / k_i$ for all $i$.
For and every
dimension $i$  we define the direction function $D_i$ so that
\begin{itemize}
\item if $f(p')_i > p'_i$ then $D_i(p) = \up$,
\item if $f(p')_i < p'_i$ then $D_i(p) = \down$, and
\item if $f(p')_i = p'_i$ then $D_i(p) = \zero$.
\end{itemize}
In other words, the function $D_i$ simply checks whether $f(p')$ moves up, down,
or not at all in dimension $i$. This completes the specification of the family
of direction functions $\mathcal{D}$. 

We must carefully choose $k_1$ through $k_d$ to ensure that the fixpoint of $f$
is contained within the grid. In fact, we need a stronger property: 
for every $i$-slice of the grid, if $f$ has a fixpoint in that $i$-slice, then
it should also appear in the grid.
Recall that $p \in P$ is a fixpoint of some slice $s$ if $D_i(p) = \zero$ for
every $i$ for which $s_i = \blank$. We can extend this definition to the
continuous function $f$ as follows: a point $x \in [0, 1]^d$ is a fixpoint of
$s$ if $(x - f(x))_i = 0$ for all $i$ for which $s_i = \blank$,
where we now interpret $s$ as specifying that $x_i = s_i/k_i$ whenever $s_i \ne
\blank$.
We are able to
show the following lemma, whose proof appears in Appendix~\ref{app:points2}.

\begin{lemma}
\label{lem:points2}
There exists integers $(k_1, k_2, \dots, k_d)$ such that for every $i$-slice $s$
we have that if $x \in [0, 1]^d$ is a fixpoint of $s$ according to $f$, then
there exists a point $p \in P(k_1, k_2, \dots, k_d)$ such that 
\begin{itemize}
\item $p$ is a fixpoint of $\mathcal{D}$, and
\item $p_i = k_i \cdot x_i$ for all $i$ where $s_i = \blank$.
\end{itemize}
Moreover, the number of bits needed to write down each $k_i$ is polynomial
in the number of bits needed to write down $f$.
\end{lemma}

This lemma states that we can pick the grid size to be fine enough so that all
fixpoints of $f$ in all $i$-slices are contained within the grid. The proof of
this is actually quite involved, and relies crucially on the fact that we have
access to a \LinearFIXP representation of $f$. From this, we can compute upper
bounds on the bit-length of any point that is a fixpoint of $f$. We also rely on
the fact that we only need to consider $i$-slices, because our proof fixes the
grid-widths one dimension at a time, starting with dimension $d$ and working
backwards.

\paragraph{\bf The extra violation.}

The specification of our reduction is now complete, but we have still not
fully defined the original problem, because we need to add an extra violation.
The issue arises with solutions of type \solnref{OV2}, where we have 
an $i$-slice $s$ and two points $p, q$ in $s$ such that
\begin{itemize}
\item $D_j(p) = D_j(q) = \zero$ for all $j < i$, 
\item $p_i = q_i + 1$, and
\item $D_i(p) = \down$ and $D_i(q) = \up$.
\end{itemize}
This means that $p$ and $q$ are both fixpoints of their respective slices
$(i-1)$-slices, and are directly adjacent to each other in dimension $i$. We are
able to show, in this situation, that if $f$ is contracting, then $f$ has a
fixpoint for the slice $s$, and it must lie between $p$ and $q$. The following
lemma is shown in Appendix~\ref{app:ov2violation}.

\begin{lemma}
\label{lem:ov2violation}
If $f$ is contracting, and we have two points $p$ and $q$ that are a violation
of type \solnref{OV2}, then there exists a point $x \in [0, 1]^n$ in the slice $s$ that
satisfies all of the following.
\begin{itemize}
\item $(x - f(x))_j = 0$ for all $j \le i$, meaning that $x$ is a fixpoint of
the slice $s$, and
\item $q_i < k_i \cdot x_i < p_i$, meaning that $x$ lies between $p$ and $q$ in
dimension $i$.
\end{itemize}
\end{lemma}

So if we have an \solnref{OV2} violation, and if $f$ is contracting, then
Lemma~\ref{lem:ov2violation} implies that there is a fixpoint $x$ of the slice
$s$ that lies strictly between $p$ and $q$ in dimension $i$. However,
Lemma~\ref{lem:points2} says that all fixpoints of $s$ lie in the grid, and
since $p$ and $q$ are directly adjacent adjacent in the grid in dimension $i$,
there is no room for $x$, so it cannot exist. The only way that this
contradiction can be resolved is if $f$ not actually contracting. 

Hence, \solnref{OV2} violations give us a concise witness that $f$ is not contracting. But
the points $p$ and $q$ themselves may satisfy the contraction property. While we
know that there must be a violation of contraction \emph{somewhere}, we are not
necessarily able to compute such a violation in polynomial time. To resolve
this, we add the analogue of an \solnref{OV2} violation to the contraction problem.

\begin{definition}[\LCM]
\label{def:LCM}
Given a \LinearFIXP circuit computing 
$f : [0, 1]^d
\rightarrow [0,1]^d$, a constant $c \in (0, 1)$, and a $p \in \Natural \cup
\{\infty\}$, find one of the following.
\begin{enumerate}[label=(CM\arabic*), wide=0pt, leftmargin=\parindent]
\item \solnlabel{CM1} A point $x \in [0, 1]^d$ such that $f(x) = x$.
\end{enumerate}
\begin{enumerate}[label=(CMV\arabic*), wide=0pt, leftmargin=\parindent]
\item \solnlabel{CMV1} Two points $x, y \in [0, 1]^d$ such that $\|f(x) - f(y)\|_p > c \cdot \| x - y \|_p$.
\item \solnlabel{CMV2} A point $x \in [0, 1]^d$ such that $f(x) \not\in [0, 1]^d$.
\item \solnlabel{CMV3} An $i$-slice $s$ and two points $x, y \in [0, 1]^d$ in $s$ such that 
\begin{itemize}
\item $(f(x) - x)_j = (f(y) - y)_j = 0$ for all $j < i$, 
\item $k_i \cdot x_i = k_i \cdot y_i + 1$, where $k_i$ is the integer given by
Lemma~\ref{lem:points2} for the \LinearFIXP circuit that computes $f$, and
\item $f(x)_i < x_i$ and $f(y)_i > y_i$.
\end{itemize}
\end{enumerate}
\end{definition}

Solution type \solnref{CM1} asks us to find a fixpoint of the map $f$, and there are two
types of violation. Violation type \solnref{CMV1} asks us to find two points $x$ and $y$
that prove that $f$ is not contracting with respect to the $\ell_p$ norm.
Violation type \solnref{CMV2} asks us to find a point that $f$ does not map to 
$[0, 1]^d$. Note that this second type of violation is necessary to make the
problem total, because it is possible that $f$ is a contraction map, but the
unique fixpoint of $f$ does not lie in $[0, 1]^d$.

Violations of type \solnref{CMV3} are the direct translation of \solnref{OV2} violations to
contraction. Note that if $f$ actually is contracting, and has a fixpoint in
$[0, 1]^n$, then no violations can exist. For \solnref{CMV3} violations, this fact is a
consequence of Lemmas~\ref{lem:points2} and~\ref{lem:ov2violation}.

\paragraph{\bf Correctness of the reduction.}

To prove that the reduction is correct, we must show that all solutions of the
\OPDC instance given by $P$ and $\mathcal{D}$ can be mapped back to solutions of
the original instance. Solutions of type \solnref{O1} give us a point $p$ such that
$D_i(p) = \zero$ for all $i$, which by definition means that the point
corresponding to $p$ is a fixpoint
of $f$. Violations of type \solnref{OV1} give us two points that are both fixpoints of the
same slice $s$, which also means that they are both fixpoints of the slice $s$
according to $f$, and it is not difficult to show that these two points violate
contraction
in an $\ell_p$ norm. Violations of type \solnref{OV3} are points that attempt
to leave the $[0, 1]^d$, and so give us a solution of type \solnref{CMV2}.
Violations of type \solnref{OV2} map directly to violations of type \solnref{CMV3}, as we have
discussed.
So we have the following lemma, which is proved in Appendix~\ref{app:lcm2opdc}.

\begin{lemma}
\label{lem:lcm2opdc}
There is a polynomial-time promise-preserving reduction from \LCM to \OPDC.
\end{lemma}

\begin{theorem}
\label{thm:lcm}
\LCM is in \UEOPLc under promise-preserving reductions.
\end{theorem}

\subsection{The P-Matrix Linear Complementarity Problem}
\label{sec:PLCPtoEOPL}

In this section, we reduce \PLCP to \USO and, separately, \PLCP to \UEOPL. 
Given that we show that \USO reduces to \UEOPL (via \OPDC) and is thus in
\UEOPLc, our direct reduction from \PLCP to \UEOPL is not needed to show that
\PLCP is contained in \UEOPLc. However, by reducing directly we can produce a
\UEOPL instance with size linear in the size of our \PLCP instance, which is
needed to obtain the algorithmic result in Section~\ref{sec:aldous}.

The direct reduction to \UEOPL relies heavily on the application of Lemke's
algorithm to P-matrix LCPs, and our reduction to \USO relies on the computation
of \emph{principal pivot transformations} of LCPs. Next we introduce the 
required concepts.
Let $[d]$ denote the set $\{1,\dots,d\}$.

\begin{definition}[LCP $(M, \qq)$]
\label{def:lcp}
Given a matrix $M \in \Real^{d \times d}$ and vector $\qq\in \Real^{d\times 1}$,
find a~{$\yy\in \Real^{d \times 1}$} s.t.:
\begin{equation}\label{eq:lcp}
\ww = M\yy + \qq \ge 0;\ \ \ \ \yy\ge 0;\ \ \ \ \\y_i\cdot \\w_i =0,\ \forall i \in [d]. 
\end{equation}
\end{definition}

In general, deciding whether an LCP has a solution is \NP-complete~\cite{chung1989np},
but if $M$ is a P-matrix, as defined next, then the LCP $(M,\qq)$ has a \emph{unique} solution 
for all $\qq\in \Real^{d\times 1}$.

The problem of checking if a matrix is a P-matrix is \coNP-complete~\cite{coxson1994p}, 
so we cannot expect to be able to verify that an LCP instance $(M,\qq)$ is actually 
defined by a P-matrix $M$. Instead, we use succinct witnesses that $M$ is not a P-matrix
as a violation solution, which allows us to define total variants of the P-LCP problem
that lie in \TFNP, as first done by Megiddo~\cite{megiddo1988note, megiddo1991total}.
This approach has previously used to place the P-matrix problem in \PPAD and \CLS~\cite{papadimitriou1994complexity,daskalakis2011continuous}.

Our paper is about problems with unique solutions. It is well known that a
matrix $M$ is a P-matrix \emph{if and only if} for all $\qq \in \Real^{d\times 1}$, 
the LCP $(M,\qq)$ has
a unique solution~\cite{cottle2009linear}. However, this characterization of a
P-matrix is not directly useful for defining succinct violations: 
while two distinct solutions would be a succinct violation, there is no
corresponding succinct witness for the case of no solutions. Next we 
introduce the three well-known succinct witnesses for $M$ not being a P-matrix that we
will use.

First, we introduce some further required notation.
Restating \eqref{eq:lcp}, the LCP problem $(M,\qq)$ seeks a pair of non-negative vectors $(\yy,\ww)$ such that:
\begin{equation}
\label{eq:IwminusMy}
I\ww - M\yy = \qq \quad \text{and } \forall i\in [d]\  \text{we have } \yy_i\cdot \ww_i=0. 
\end{equation}
If $\qq \ge 0$, then $(\yy,\ww) = (0,\qq)$ is a \emph{trivial solution}.
We identify a solution $(\yy,\ww)$ with the set of components of $\yy$ that are positive:
let $\alpha = \{i \ | \ \yy_i > 0, \ i \in [d]\}$ denote such a set of ``basic variables''.
Going the other way, to check if there is a solution that corresponds to a
particular $\alpha \subseteq [d]$, we
try to perform a \emph{principal pivot transformation}, re-writing the LCP by 
writing certain variables $\yy_i$ as~$\ww_i$, and checking if in this re-written 
LCP there exists the trivial solution $(\yy',\ww') = (0,\qq')$.
To that end, we construct an $d \times d$ matrix $A_\alpha$, 
where the $i$th column of $A_\alpha$ is defined as follows.
Let $e_i$ denote the $i$th 
unit column vector in dimension $d$, and let $M_{\cdot i}$ denote the $i$th column of the matrix $M$.
\begin{equation}
\label{eq:Aalpha}
(A_\alpha)_{\cdot i} \ :=
\begin{cases}
	-M_{\cdot i} & \quad \text{if } i \in \alpha,\\
	\ \ e_i & \quad \text{if } i \notin \alpha.
\end{cases} 
\end{equation}
Then $\alpha$ corresponds to an LCP solution if 
$A_\alpha$ is non-singular and, the ``new \qq'' i.e.,
$(A_\alpha^{-1}\qq)$, is non-negative.
For a given $\alpha \subseteq [n]$, we define $\out(\alpha) := \vblank$ if $\det(A_\alpha) = 0$;
note that this will not happen if $M$ is really a P-matrix, but our general treatment here is made
to deal with the non-promise problem, in which case a zero determinant will correspond to a
violation\footnote{We could also define $\out(\alpha) := \vblank$ if $\det(A_\alpha) < 0$, since then we also get
a P-matrix violation, however we choose to still perform a principal pivot transformation
in this case since, ceteris paribus, it seems reasonable to prefer an LCP solution to a P-matrix violation
when we reduce P-LCP to USO.}.
If $\det(A_\alpha) \ne 0$, we define $\out(\alpha)$ as a bit-string in $\{0,1\}^d$ as follows:
\begin{equation}
\label{def:outK}
(\out(\alpha))_i = 
\begin{cases}
1 & \quad \text{if } (A_\alpha^{-1}\qq)_i < 0,\\
0 & \quad \text{if } (A_\alpha^{-1}\qq)_i \ge 0.
\end{cases}
\end{equation}
With this notation, a subset $\alpha \subseteq [d]$ corresponds to a solution of
the LCP if $\out(\alpha) = 0^d$.
We will use $\out(\alpha)$ both to define a succinct violation of the P-matrix property, and 
in our promise-preserving reduction from \PLCP to \USO.

Next, we introduce three types of succinct violations 
that prove that a matrix $M$ is not a P-matrix. Let $\cha(v)$ for $v \subseteq [d]$ denote 
the characteristic vector of $v$, i.e., $(\cha(v))_i = 1$  if $i\in v$ and 0 otherwise.
As in the section on USOs, when we write $\oplus$ and $\cap$, here we mean 
the bit-wise XOR and bit-wise and-operations on bit-strings.
\begin{definition}[\PLCP violations]
For a given LCP $(M,\qq)$ in dimension $d$,
each of the following provides a polynomial-size witness that $M$ is not a P-matrix:
\begin{enumerate}[label=(PV\arabic*), wide=0pt] 
	\item \solnlabel{PV1} A set $\alpha \subseteq [d]$ such that the corresponding principal minor is non-positive, i.e., $\det(M_{\alpha\alpha}) \le 0$.
	\item \solnlabel{PV2} A vector $x \ne 0$ whose sign is reversed by $M$, that is, for all $i \in [d]$ we have $x_i(Mx)_i \le 0$.
	\item \solnlabel{PV3} Two distinct sets $\alpha, \beta \subset [d]$, with $(\cha(\alpha) \oplus \cha(\beta)) \cap (\out(\alpha) \oplus \out(\beta)) = 0^d$.
\end{enumerate}
\end{definition}
The violation PV1 corresponds to the standard definition of a P-matrix as having all
positive principal minors. 
Megiddo~\cite{megiddo1988note, megiddo1991total} used this violation to 
place to place the \PLCP problem in \TFNP.
The same violation was then used by Papadimitriou to put \PLCP in \PPAD, because Lemke's
algorithm is a \PPAD-type complementary pivoting algorithm, which inspired
\PPAD, and it will return a non-positive principal minor if it fails to find an
LCP solution.

The characterization of P-matrices as those that do not reverse the sign of any
non-zero vector, as used for violation PV2, was first discovered by Gale and
Nikaido~\cite{GaleN65}. The final violation, PV3, follows from the work
of Stickney and Watson~\cite{StickneyW78}, who showed that P-matrix LCPs give
rise to USOs, and from the work of Tzabo and Welzl~\cite[Lemma 2.3]{SzaboW01},
who showed that this condition characterizes that ``outmap'' of USOs, as
discussed in Section~\ref{sec:uso}, hence the name of the function being ``out''.

Several other characterizations of P-matrices are known, some of which would
provide alternative succinct violations~\cite{cottle2009linear,JT95}. We have
given the violations PV1--PV3 above, since we use the three for our promise
preserving reductions from \PLCP to \UEOPL and to \USO. In particular, for our
reduction from \PLCP to \USO we need violations of types PV1 and PV3, and for
our reduction from \PLCP directly to \UEOPL we need violations of types PV1 and
PV2. It is not immediately apparent how to convert violations of one type to
another in polynomial time, and it is conceivable that allowing different sets
of violations changes the complexity of the problem. We leave it as further work
to further explore this.

 
\begin{definition}[\PLCP] \label{def:plcp} Given an LCP $(M, \qq)$, find either:
\begin{enumerate}[label=(Q\arabic*), wide=0pt] 
\item \solnlabel{Q1} $\yy \in \Real^{d\times 1}$ that satisfies (\ref{eq:lcp}).
\item \solnlabel{Q2} one of the violations \solnref{PV1}--\solnref{PV3}. 
\end{enumerate}
In the promise version, we are promised that $M$ is a P-matrix and seek a solution of type \PLo.
\end{definition}

\paragraph{Reduction from \PLCP to \USO.}
We are now ready to present our reduction to USO.
The reduction is simple, and we present it in full detail here.
For the LCP instance $\CI = (M,\qq)$ in dimension $d$, we produce an instance
$\mathcal{U}$ of \USO also in dimension $d$.

We first need to deal with the possibility that $\qq$ is degenerate.
A P-matrix LCP has a degenerate~$\qq$ if $A_\alpha^{-1}\cdot \qq$ has a zero entry for some
$\alpha \subseteq [d]$.
To ensure that this does not present a problem for our reduction, we use a standard
technique known as lexicographic perturbation~\cite[Section 4.2]{cottle2009linear}:
In the reduction that follows we assume that we are using such a degeneracy resolution scheme.

The reduction associates each vertex $v$ of the resulting USO with a set a $\alpha(v)$ of basic variables
for the LCP, and then uses $\out(\alpha)$ as the outmap at $v$.
In detail,
for a vertex $v \in \{0,1\}^d$ of $\mathcal{U}$, we define $\alpha(v) = \{i \ | \ v(i) = 1\}$,
and set $\udir(v) = \out(\alpha(v))$.
It immediately follows that:
\begin{itemize}
\item A solution of type \solnref{US1} in $\mathcal{U}$ is a solution of type \solnref{Q1} in \CI.
\item A solution of type \solnref{USV1} in $\mathcal{U}$ is a solution of type \solnref{PV1} in \CI.
\item A solution of type \solnref{USV2} in $\mathcal{U}$ is a solution of type \solnref{PV3} in \CI.
\end{itemize}
If $M$ is actually a P-matrix then $\mathcal{U}$ will have exactly one solution of type \solnref{US1},
and no violation solutions~\cite{StickneyW78}. 
Thus our reduction is promise preserving, and we obtain the following.

\begin{theorem}
\label{thm:plcp2uso}
There is a polynomial-time promise-preserving reduction from \PLCP with violations
of type \solnref{PV1} and \solnref{PV3} to \USO.
\end{theorem}

\begin{figure}[tbp]
   \centering
\begin{center}
\scalebox{0.99}{ 
\begin{tikzpicture}

	\path[use as bounding box] (-3,-3) rectangle (3.7,3.9);
	\clip (-3,-3) rectangle (3.7,3.9);

	\tikzstyle{colvec} = [-{>[width=2mm,length=2mm]}, thick]

	\draw[draw=none, fill=shade1] (3,1.5)  -- (3.5,2) -- (1.3,3.9) --  (0,0) -- cycle {};
	\draw[draw=none, fill=shade2] (3,1.5) -- (3.7,1) -- (3.2, 0) -- (1.5,-2) -- (0,-1.5) -- (0,0) -- cycle {};
	\draw[draw=none, fill=shade3] (-1.75,0) -- (-1.3,3.6) -- (1.3,3.9) -- (0,0) -- cycle {};
	\draw[draw=none, fill=shade4] (-2.2,0) -- (-3,-2) -- (-2.5,-2.5) -- (-2,-3) -- (0,-1.5) -- (0,0) -- cycle {};

	\draw[colvec, darkpastelblue] (0,0) -- (2,1) {};
	\node[darkpastelblue] at (2.3,0.8) {$M_1$};
	\node[darkpastelblue] at (2.9,0.8) {$\begin{bmatrix} 2\\ 1\\ \end{bmatrix}$};
	\draw[colvec, darkpastelblue] (0,0) -- (1,3) {};
	\node[darkpastelblue] at (0.5,2.9) {$M_2$};
	\node[darkpastelblue] at (1.4,2.9) {$\begin{bmatrix} 1\\ 3\\ \end{bmatrix}$};
	\draw[colvec, darkpastelred] (0,0) -- (-1,0) {};
	\node[darkpastelred] at (-1,0.3) {$-e_1$};
	\draw[colvec, darkpastelred] (0,0) -- (0,-1) {};
	\node[darkpastelred] at (0.45,-0.85) {$-e_2$};

	\draw[-, darkpastelbrown, thick] (-2,-1) -- (1,1);
	\node[fill=darkpastelbrown, circle, inner sep=1.5pt] at (0.142,0.429) {};
	\node[fill=darkpastelbrown, circle, inner sep=1.5pt] at (-0.5,0) {};


	\node[fill=darkpastelmagenta, circle, inner sep=1.5pt] at (1,1) {};
	\node[darkpastelmagenta] at (1.35,1.55) {$-\qq = \begin{bmatrix} 1\\ 1\\ \end{bmatrix}$};

	\node[fill=darkpastelbrown, circle, inner sep=1.5pt] at (-2,-1) {};
	\node[darkpastelbrown] at (-2,-1.8) {$-\cov = \begin{bmatrix} -2\\ -1\\ \end{bmatrix}$};


\end{tikzpicture}
}
\hskip -1.75cm
\scalebox{0.96}{ 
\tikzset{arrowstyle/.style={draw=black, fill=black, single arrow, minimum height=#1, minimum width=6ex, single arrow head extend=.4cm,}}
\newcommand{\tikzfancyarrow}[3]{\node [arrowstyle=1cm, rotate=-90] at (#1,#2) {};}

\newcommand{\zuparrow}[2]{
	\draw[->,line width=0.1mm] ([yshift=2mm]$(#1)!0.5!(#2)$) -- node[left] {\small $z$} ([yshift=5mm]$(#1)!0.5!(#2)$);
}
\newcommand{\zdnarrow}[2]{
	\draw[->,line width=0.1mm] ([yshift=5mm]$(#1)!0.5!(#2)$) -- node[left] {\small $z$} ([yshift=2mm]$(#1)!0.5!(#2)$);
}

\pgfdeclarelayer{background}
\pgfsetlayers{background,main}

\tikzstyle{vertex}=[draw,ultra thick,circle,fill=white,minimum size=10pt,inner sep=0pt]
\tikzstyle{red vertex} = [vertex, fill=pastelred]
\tikzstyle{blue vertex} = [vertex, fill=pastelblue]
\tikzstyle{edge} = [draw,thick,->]
\tikzstyle{weight} = []

\begin{tikzpicture}[scale=1,swap]

\path[use as bounding box] (-1.7,-3) rectangle (7.7,3);
\clip (-1.7,-3) rectangle (7.7,3);

\node (ray) at (-1.5,1) {};

\foreach \pos/\name in {
{(0,1)/start},
{(1,1)/1},
{(2,1)/2},
{(3,1)/3}, 
{(4,1)/4}, 
{(5,1)/5}, 
{(6,1)/6}, 
{(7,1)/end}}
	\node[vertex] (\name) at \pos {};

\path[edge] (ray) -- node[below,xshift=-1mm, align=center, fontscale=-1] {primary\\ ray} node () {} (start);

\zdnarrow{ray}{start}
\zdnarrow{start}{1}
\zdnarrow{1}{2}
\zuparrow{2}{3}
\zuparrow{3}{4}
\zuparrow{4}{5}
\zdnarrow{5}{6}
\zdnarrow{6}{end}



\foreach \source/ \dest /\weight in
	{
	start/1/,
	1/2/,
	2/3/,
	3/4/,
	4/5/,
	5/6/,
	6/end/
	}
	\path[edge] (\source) -- node[weight] {$\weight$} (\dest);


\node[red vertex] (10) at (-1,-1) [label=above:{$\mathbf 0$}] {};

\foreach \pos/\name in {
{(0,-1)/11},
{(1,-1)/12},
{(2,-1)/13}, 
{(3,-1)/14}, 
{(4,-1)/15},
{(5,-1)/16},
{(6,-1)/17},
{(7,-1)/18}}
	\node[vertex] (\name) at \pos {};

\draw[edge,->] (10) edge [loop left] node [below]  {\tiny $P$} ();
\draw[edge,->] (10) edge [bend left] node [above] {\tiny $S$} (11);
\draw[edge,->] (11) edge [bend left] node [below] {\tiny $P$} (10);

\draw[edge,->] (11) edge [bend left] node [above] {\tiny $S$} (12);
\draw[edge,->] (12)	edge [bend left] node [below] {\tiny $P$} (11);

\draw[edge,->] (12) edge [bend left] node [above] {\tiny $S$} (13);
\draw[edge,->] (13)	edge [bend left] node [below] {\tiny $P$} (12);
\draw[edge,->] (13) edge [loop right] node [above]  {\tiny $S$} ();

\draw[edge,->] (14) edge [loop below] node [right]  {\tiny $S,P$} ();
\draw[edge,->] (15) edge [loop below] node [right]  {\tiny $S,P$} ();

\draw[edge,->] (16) edge [loop left] node [below]  {\tiny $P$} ();
\draw[edge,->] (16) edge [bend left] node [above] {\tiny $S$} (17);
\draw[edge,->] (17)	edge [bend left] node [below] {\tiny $P$} (16);

\draw[edge,->] (17) edge [bend left] node [above] {\tiny $S$} (18);
\draw[edge,->] (18)	edge [bend left] node [below] {\tiny $P$} (17);
\draw[edge,->] (18) edge [loop right] node [above]  {\tiny $S$} ();

\node[blue vertex] at (18) {};
\node[blue vertex] at (16) {};
\node[blue vertex] at (13) {};

\draw[-{>[scale=2.5,
          length=2,
  		  width=8]}, line width=2mm] (3,0.4) -- (3,-0.4); 

\node[align=center] at (3,2.25) {Lemke path showing increase/decrease of $z$};
\node[align=center] at (3,-2.25) {\textsc{EndOfPotentialLine} instance}; 

\end{tikzpicture}
}
\end{center}
\caption{
Left: geometric view of Lemke's algorithm as inverting a piecewise linear map.
Right: construction of $S$ and $P$ for \EOPL instance from the Lemke
path, where the arrows on its edges indicate
whether $z$ increases or decreases along the edge.
This figure should be viewed in color.
}
\label{fig:lemke}
\end{figure}
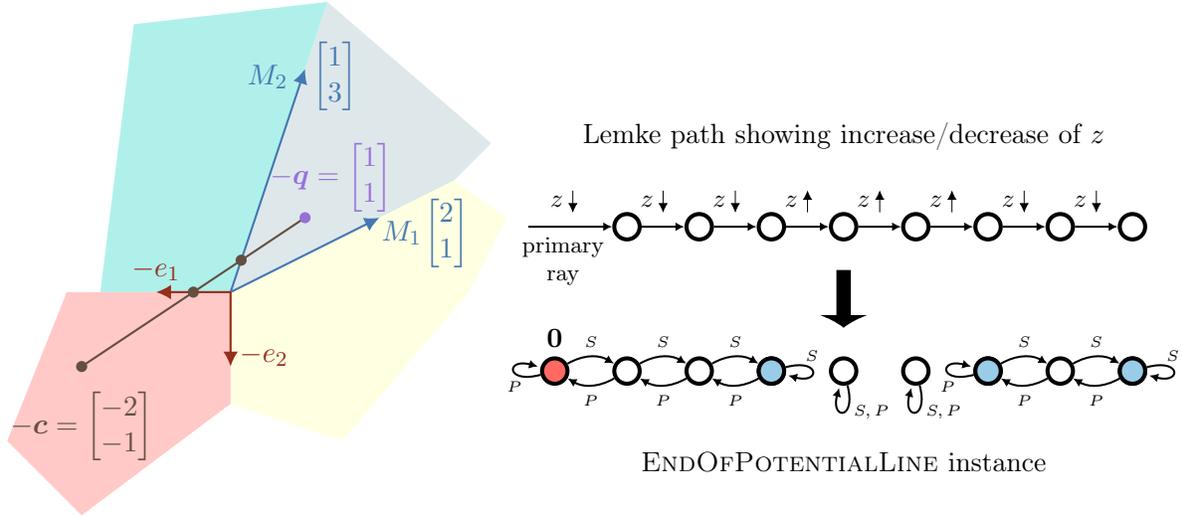   

\medskip

\noindent \textbf{Overview of reductions from \PLCP to \EOPL and \UEOPL.}
Comparing \UEOPL and \EOPL we see that:
\solnref{U1} and \solnref{UV2} correspond to \solnref{R1}, 
and \solnref{UV1} corresponds to \solnref{R2}, so the only difference
is that \UEOPL has the extra violation solution \solnref{UV3}.
Thus there is some extra work to do for our reduction to \UEOPL, to 
map \solnref{UV3} solutions back to a P-LCP solution.

Our reduction from P-LCP to the two problems produces the same instance.
The reduction in both cases is based on Lemke's algorithm, which we 
describe in detail in Appendix~\ref{app:lemke}.
For \EOPL, we only need to use \solnref{PV1} violations.
For \UEOPL, we only need to use \solnref{PV2} violations.
Next we give a high-level description of the reduction, where for simplicity
we just refer to a resulting \EOPL instance.
Full details of both reductions appear in Appendix~\ref{app:full_plcp_reduction}.

Lemke's algorithm introduces to the LCP an extra variable $z$ and an extra positive vector~$\cov$,
called a covering vector.
It follows a path along edges of the new LCP polyhedron 
based on a complementary pivot rule that maintains an almost-complementary solution.  
In Figure~\ref{fig:lemke}, we give an example with $\cov = (2,1)^\top$; in
our reduction we take $\cov$ to be the all ones vector $\ones$. 
Geometrically, solving an LCP is equivalent to finding a \emph{complementary
cone}, corresponding to a subset of columns of $M$ and the complementary 
unit vectors, that contains $-\qq$. This is depicted on the left
in~Figure~\ref{fig:lemke}, which also shows Lemke's algorithm as inverting a
piecewise linear map along the line from $-\cov$ to $-\qq$.
The algorithm pivots between the brown vertices at the
intersections of complementary cones and terminates at $-\qq$. The extra variable
$z$ can be normalized and then describes how far along the line from $\cov$ to $-\qq$
we are. P-matrix LCPs are exactly those where the complementary cones cover the
whole space with no overlap, and then $z$ decreases monotonically as Lemke's
algorithm proceeds.

Each vertex along the Lemke path corresponds to a subset of $[d]$ of 
basic variables, and a possibly a unique duplicate label. 
A label $l \in [d]$ is said to be duplicate if $y_l=0$ as well as $w_l=0$. 
The vertices without a duplicate label has $z=0$ and correspond to a solution of the LCP. 
To encode these subsets and the duplicate label, we consider bit strings of
length $n=2d$ that represent vertices in the \EOPL instance. 
The first $d$ bits encode the subset, and bits $(d+1,\dots, 2d)$ bits encode the
duplicate label, where bit $(d+l)$ is one if $l$ is the duplicate label.
Thus, ``valid'' vertices in \EOPL instance can have at most one bit set to one
among bits $(d+1)$ through $2d$.
We ensure that ``invalid'' bit configurations form self-loops in the \EOPL
instance, and hence do not give rise to any solutions.

We use the following key properties of Lemke's algorithm as applied to a P-matrix:
\begin{enumerate}
\itemsep0mm
\item If Lemke's algorithm does not terminate with a LCP solution \PLo, it provides a \PLt solution.
\item For a P-matrix $M$, the extra variable $z$ strictly decreases in each step of the algorithm.
\item Given a subset of $[d]$ and duplicate label $l$, we can efficiently decide if it corresponds to a vertex on a Lemke path.
\item By a result of Todd~\cite[Section 5]{todd1976orientation}, the Lemke path can be efficiently locally oriented.
\end{enumerate}

Recall that LCP \eqref{eq:lcp} has a trivial solution, namely $\yy=0$, if $\qq\ge 0$. 
Therefore, wlog assume that $\min_{i \in [d]} q_i <0$. 
The starting vertex of the Lemke path is the vertex $\xx^0 = (\yy^0,\ww^0,z^0)$
with $\yy^0=0$, $z^0=|\min_{i\in [d]} q_i|$, and
$\ww=\qq+z^0\ones$. 
So $0^n$ is a start of line in the \EOPL
instance, we point the successor of $0^n$ to the bit configuration corresponding to $\xx^0$. 
We then follow the line of bit configurations corresponding to the vertices
traversed by Lemke's algorithm, updating $z$ in each step.
We use $(z^0-z+1)$ as the potential
function -- $(z^0-z)$ to ensure increasing potential along the line, and $+1$ to
ensure that only $0^n$ has zero potential.
If we start with a \PLCP instance where~$M$ is actually a P-matrix 
then this reduction will produce a single line from $\xx^0$ to the solution of the
\PLCP, and~$z$ will monotonically decrease along this line. 
The main difficulty of the reduction is dealing with the case where~$M$ is not a
P-matrix. This may cause Lemke's algorithm to terminate without an LCP solution.
Another issue is that, even when Lemke's algorithm does find a
solution,~$z$ may not decrease monotonically along the line. 

In the former case, the first property above gives us a \PLt solution for the
\PLCP problem. In the latter case, we define any point on the line where $z$
increases to be a self-loop, breaking the line at these points.
Figure~\ref{fig:lemke} shows an example, where the two vertices at which $z$
increases are turned into self loops, thereby introducing two new solutions before
and after the break.
Both of these solutions give us a \PLt solution for the \PLCP instance. The 
full details of the reduction are
involved and appear in Appendix~\ref{app:full_plcp_reduction}. It is worth
noting that, in the case where the input is actually a P-matrix, the resulting
\EOPL instance has a unique line, so our reduction in promise preserving.
Moreover, our \EOPL instance is a valid \UEOPL instance, and we can map back
all violations so as to obtain the following.

\begin{theorem}
\label{thm:plcp2ueopldirectly}
There are a polynomial-time promise-preserving reduction from 
\PLCP with violations of type \solnref{PV1} to \EOPL,
and from \PLCP with violations of type \solnref{PV2} to \UEOPL, and thereby also to \EOPL.
Both reductions only incur a linear blowup in the size of the instance.
\end{theorem}


\section{Algorithms}

\subsection{Algorithms for Contraction Maps}
\label{sec:algorithms}

\paragraph{\bf An algorithm for \LCM.}

The properties that we observed in our reduction from \LCM to \EOPL can also be
used to give polynomial time algorithms for the case where the number of
dimensions is constant. In our two-dimensional example, we relied on the fact
that each dimension-two slice has a unique point on the blue surface, and 
that the direction function at this point tells us the direction of the overall
fixpoint. 

This suggests that a nested binary search approach can be used to find the
fixpoint. The outer binary search will work on dimension-two coordinates, and
the inner binary search will work on dimension-one coordinates. For each fixed
dimension-two coordinate $y$, we can apply the inner binary search to find the
unique point $(x, y)$ that is on the blue surface. Once we have done so, $D_2(x,
y)$ tells us how to update the outer binary search to find a new candidate
coordinate $y'$. 

This can be generalized to $d$-dimensional instances, by
running $d$ nested instances of binary search. Moreover, our algorithm can detect violations in the course of performing the binary search and is able to produce witnesses to the given function not being a contraction map. Thus, our algorithm solves the non-promise problem \LCM, giving the following theorem, whose proof appears in Appendix~\ref{subsec:exact_algo_details}.

\begin{theorem}
Given a $\LinearFIXP$ circuit $C$ purporting to encode a contraction map $f : [0,1]^d\to
[0,1]^d$ with respect to any $\ell_p$ norm, there is an algorithm to find a
fixpoint of $f$ or return a pair of points witnessing that $f$ is not a contraction map in time that is polynomial in $\size(C)$ and exponential in $d$.
\end{theorem}

\paragraph{\bf An algorithm for \CM.}

We are also able to generalize this to the more general \CM problem, where the
input is given as an arbitrary (non-linear) arithmetic circuit. Here
the key issue is that the fixpoint may not be rational, and so we must find a
suitably accurate approximate fixpoint. Our nested binary search
approach can be adapted to do this.

Since we now deal with approximate fixpoints, we must cut off each of our nested
binary search instances at an appropriate accuracy. Specifically, we must ensure
that the solution is accurate enough so that we can correctly update the outer
binary search. 
Choosing these cutoff points turns out to be quite involved, as we must choose
different cutoff points depending on both the norm and the level of recursion,
and moreover the $\ell_1$ case requires a separate proof.

Again, the algorithm is able to detect violations of contraction during the course of the binary search, and thus solves the more general problem of either finding a fixpoint when the circuit defines a contraction map, or returning a pair of points that are not contracting.
The details of this are deferred to
Appendix~\ref{subsec:approx_algo_details}, where the following theorem is shown.

\begin{theorem}
\label{thm:alggeneral}
For a contraction map $f:[0,1]^d\to [0,1]^d$ under $\Norm{\cdot}_p$ for $2 \leq
p < \infty$, there is an algorithm to compute a point $v\in [0,1]^d$ such that
$\Norm{f(v) - v}_p < \eps$ or return a pair of points $(x,y)$ such that $\Norm{f(x) - f(y)}_p > c\Norm{x - y}_p$ in time $O(p^{d^2}\log^d(1/\eps)\log^d(p))$.
\end{theorem}

Actually, our algorithm treats the function as a black-box, and so it can be
applied to any contraction map, with Theorem~\ref{thm:alggeneral} giving the
number of queries that need to be made.

\subsection{Aldous' algorithm for PLCP}
\label{sec:aldous}


Aldous~\cite{Aldous83} analysed a simple randomized algorithm for solving local
search problems. The algorithm randomly samples a large number of candidate
solutions and then performs a local search from the best sampled solution.
Aldous' algorithm can solve any problem in \PLS and thus any problem in \UEOPLc.
In~\cite{GHHKMS18} it was noted that, because our reduction from our reduction
from \PLCP to \UEOPL only incurs a linear blowup, that is, from an LCP in dimension 
$n$ we produce an \UEOPL
instance with $O(2^n)$ vertices, when we apply Aldous' algorithm to the
resulting instance, the expected time is $2^{n/2}\cdot \poly(n)$ in the worst
case, which gives the fastest known running time for a randomized algorithm for
\PLCP.
Thus, we get the following corollary of Theorem~\ref{thm:plcp2ueopldirectly}.

\begin{corollary}
There is a randomized algorithm for \PLCP that runs in expected time $O(1.4143^n)$.
\end{corollary}

\newpage

\bibliography{paper}

\newpage
\appendix

\section{Proofs for Section~\ref{sec:EOPL}: Equivalence of EOPL and EOML}
\label{app:eoml2eopl}

%

First we recall the definition of \EOML, which was first defined
in~\cite{hubavcek2017hardness}. 
It is close in spirit to the problem \EOL that
is used to define \PPAD~\cite{papadimitriou1994complexity}. 

\begin{definition}[\EOML~\cite{hubavcek2017hardness}]
Given circuits $S,P: \{0,1\}^n \rightarrow \{0,1\}^n$, and $V:\{0,1\}^n\rightarrow \{0,\dots, 2^n\}$ such that $P(0^n) =0^n\neq S(0^n)$ and $V(0^n)=1$, find a string $\xx \in \{0,1\}^n$ satisfying one of the following 
\begin{enumerate}[label=(T\arabic*)]
\itemsep0mm
\item either $S(P(\xx))\neq \xx \neq 0^n$ or $P(S(\xx))\neq \xx$,
\item $\xx\neq 0^n, V(\xx)=1$,
\item either $V(\xx)>0$ and $V(S(\xx))-V(\xx)\neq 1$, or $V(\xx)>1$ and $V(\xx)-V(P(\xx))\neq 1$. 
\end{enumerate}
\end{definition}
\EOML is actually quite similar to \EOPL. The main difference is that in \EOML,
each edge must increase the potential by exactly 1, which is enforced by
solution type T3. 

\subsection{\EOML to \EOPL}
\label{sec:EOMLtoEOPL}

Given an instance \CI\  of \EOML defined by circuits $S,P$ and $V$ on vertex
set $\{0,1\}^n$ we are going to create an instance $\CI'$ of \EOPL with circuits
$S',P'$, and $V'$ on vertex set $\{0,1\}^{(n+1)}$, i.e., we introduce one extra bit.  
This extra bit is essentially to take care of the difference in the value of potential 
at the starting point in \EOML and \EOPL, namely $1$ and $0$ respectively. 

Let $k=n+1$, then we create a potential function $V':\{0,1\}^k \rightarrow
\{0,\dots,2^k-1\}$. 
The idea is to make $0^k$ the starting point with potential zero as required,
and to make all other vertices with first bit $0$ be dummy vertices with self
loops. The real graph
will be embedded in vertices with first bit $1$, i.e., of type $(1,\uu)$. Here
by $(b,\uu)\in \{0,1\}^k$, where $b\in \{0,1\}$ and $\uu\in \{0,1\}^n$, we mean
a $k$ length bit string with first bit set to $b$ and for each $i\in[2:k]$ bit $i$ 
set to bit $u_i$. 

\medskip
\medskip

\noindent{\bf Procedure $V'(b,\uu)$:} If $b=0$ then Return $0$, otherwise Return $V(\uu)$. 
\medskip
\medskip

\noindent{\bf Procedure $S'(b,\uu)$:}
\vspace{-0.3cm}

\begin{enumerate}
\itemsep1mm
\item If $(b,\uu)=0^k$ then Return $(1,0^n)$
\item If $b=0$ and $\uu\neq 0^n$ then Return $(b,\uu)$ (creating self loop for dummy vertices)
\item If $b=1$ and $V(\uu)=0$ then Return $(b,\uu)$ (vertices with zero potentials have self loops)
\item If $b=1$ and $V(\uu)>0$ then Return $(b,S(\uu))$ (the rest follows $S$)
\end{enumerate}

\noindent{\bf Procedure $P'(b,\uu)$:}
\vspace{-0.3cm}

\begin{enumerate}
\itemsep1mm
\item If $(b,\uu)=0^k$ then Return $(b,\uu)$ (initial vertex points to itself in $P'$).
\item If $b=0$ and $\uu\neq 0^n$ then Return $(b,\uu)$ (creating self loop for dummy vertices)
\item If $b=1$ and $\uu=0^n$ then Return $0^k$ (to make $(0,0^n)\rightarrow (1,0^n)$ edge consistent)
\item If $b=1$ and $V(\uu)=0$ then Return $(b,\uu)$ (vertices with zero potentials have self loops)
\item If $b=1$ and $V(\uu)>0$ and $\uu \neq 0^n$ then Return $(b,P(\uu))$ (the rest follows $P$)
\end{enumerate}

Valid solutions of \EOML of type T2 and T3 requires the potential to be strictly greater than zero, while solutions of \EOPL may have zero potential. However, a solution of \EOPL can not be a self loop, so we've added self-loops around vertices with zero potential in the \EOPL instance.
By construction, the next lemma follows:
\begin{lemma}\label{lem:m2p-valid}
$S'$, $P'$, $V'$ are well defined and polynomial in the sizes of $S$, $P$, $V$ respectively. 
\end{lemma}

Our main theorem in this section is a consequence of the following three lemmas.

\begin{lemma}\label{lem:m2p-sl}
For any $\xx=(b,\uu)\in \{0,1\}^k$, $P'(\xx)=S'(\xx)=\xx$ (self loop) iff $\xx\neq 0^k$, and $b=0$ or $V(\uu)=0$.
\end{lemma}
\begin{proof}
This follows by the construction of $V'$, the second condition in $S'$ and $P'$, and third and fourth conditions in $S'$ and $P'$ respectively. 
\end{proof}

\begin{lemma}\label{lem:m2p-r1}
Let $\xx=(b,\uu)\in \{0,1\}^k$ be such that $S'(P'(\xx))\neq \xx \neq 0^k$ or $P'(S'(\xx))\neq \xx$ (an \solnref{R1} type solution of \EOPL instance $\CI'$), then $\uu$ is a solution of \EOML instance $\CI$.
\end{lemma}
\begin{proof}
The proof requires a careful case analysis. 
By the first conditions in the descriptions of $S',P'$ and $V'$, we have $\xx \neq 0^k$. 
Further, since $\xx$ is not a self loop, Lemma \ref{lem:m2p-sl} implies $b=1$  and $V'(1,\uu)=V(\uu)>0$.
\medskip

\noindent{\em Case I.}
If $S'(P'(\xx))\neq \xx\neq 0^k$ then we will show that either $\uu$ is a genuine start of a line other than $0^n$ giving a T1 type solution of \EOML instance $\CI$, or there is some issue with the potential at $\uu$ giving either a T2 or T3 type solution of $\CI$. Since $S'(P'(1,0^n))=(1,0^n)$, $\uu \neq 0^n$. Thus if $S(P(\uu))\neq \uu$ then we get a T1 type solution of $\CI$ and proof follows. If $V(\uu)=1$ then we get a T2 solution of $\CI$ and proof follows. 

Otherwise, we have $S(P(\uu))=\uu$ and $V(\uu)>1$. Now since also $b=1$ $(1,\uu)$ is not a self loop (Lemma \ref{lem:m2p-sl}). 
Then it must be the case that $P'(1,\uu)=(1,P(\uu))$. However, $S'(1,P(\uu))\neq (1,\uu)$ even though $S(P(\uu))=\uu$. This happens only when $P(\uu)$ is a self loop because of $V(P(\uu))=0$ (third condition of $P'$).
Therefore, we have $V(\uu)-V(P(\uu))>1$ implying that $\uu$ is a T3 type solution of $\CI$. 
\medskip

\noindent{\em Case II.}
Similarly, if $P'(S'(\xx))\neq \xx$, then either $\uu$ is a genuine end of a line of $\CI$, or there is some issue with the potential at $\uu$. If $P(S(\uu))\neq \uu$ then we get T1 solution of $\CI$. Otherwise, $P(S(\uu))=\uu$ and $V(\uu)>0$. Now as $(b,\uu)$ is not a self loop and $V(\uu)>0$, it must be the case that $S'(b,\uu)=(1,S(\uu))$. However, $P'(1, S(\uu))\neq (b,\uu)$ even though $P(S(\uu))=\uu$. This happens only when $S(\uu)$ is a self loop because of $V(S(\uu))=0$. Therefore, we get $V(S(\uu))-V(\uu)<0$, i.e., $\uu$ is a type T3 solution of $\CI$. 
\end{proof}

\begin{lemma}\label{lem:m2p-r2}
Let $\xx=(b,\uu)\in \{0,1\}^k$ be an \solnref{R2} type solution of the constructed \EOPL instance $\CI'$, then $\uu$ is a type T3 solution of \EOML instance~$\CI$.
\end{lemma}
\begin{proof}
Clearly, $\xx\neq 0^k$. Let $\yy = (b',\uu') = S'(\xx) \neq \xx$, and observe that $P(\yy) = \xx$. This also implies that $\yy$ is not a self loop, and hence $b=b'=1$ and $V(\uu)>0$ (Lemma \ref{lem:m2p-sl}). Further, $\yy = S'(1,\uu)=(1,S(\uu))$, hence $\uu'=S(\uu)$. Also, $V'(\xx)=V'(1,\uu)=V(\uu)$ and $V'(\yy)=V'(1,\uu')=V(\uu')$. 

Since $V'(\yy)-V'(\xx)\le 0$ we get $V(\uu')-V(\uu)\le 0 \Rightarrow V(S(\uu)) - V(\uu) \le 0\Rightarrow V(S(\uu)) - V(\uu)\neq 1$. Given that $V(\uu)>0$, $\uu$ gives a type T3 solution of \EOML.
\end{proof}

\begin{theorem}\label{thm:m2p}
An instance of \EOML can be reduced to an instance of \EOPL in linear time such that a solution of the former can be constructed in a linear time from the solution of the latter. 
\end{theorem}

\subsection{\EOPL to \EOML}
\label{sec:eopl2eoml}

In this section we give a linear time reduction from an instance $\CI$ of \EOPL to an instance $\CI'$ of \EOML. Let the given \EOPL instance $\CI$ be defined on vertex set $\{0,1\}^n$ and with procedures $S,P$ and $V$, where $V:\{0,1\}^n\rightarrow \{0,\dots,2^m-1\}$. 
\medskip

\noindent{\bf Valid Edge.} We call an edge $\uu \rightarrow \vv$ valid if $\vv=S(\uu)$ and $\uu=P(\vv)$. 
\medskip

We construct an \EOML instance $\CI'$ on $\{0,1\}^k$ vertices where $k=n+m$. 
Let $S',P'$ and $V'$ denotes the procedures for $\CI'$ instance. 
The idea is to capture value $V(\xx)$ of the potential in the $m$ least significant bits of vertex description itself, so that it can be gradually increased or decreased on valid edges. For vertices with irrelevant values of these least $m$ significant bits we will create self loops. Invalid edges will also become self loops, e.g., if $\yy=S(\xx)$ but $P(\yy)\neq \xx$ then set $S'(\xx,.)=(\xx,.)$. We will see how these can not introduce new solutions. 

In order to ensure $V'(0^k)=1$, the $V(S(0^n))=1$ case needs to be discarded. For
this, we first do some initial checks to see if the given instance $\CI$ is not
trivial.  If the input \EOPL instance is trivial, in the sense that either
$0^n$ or $S(0^n)$ is a solution, then we can just return it.

\begin{lemma}
\label{lem:valid-edges}
If $0^n$ or $S(0^n)$ are not solutions of \EOPL instance $\CI$ then $0^n
\rightarrow S(0^n) \rightarrow S(S(0^n))$ are valid edges, and $V(S(S(0^n))\ge 2$. 
\end{lemma}

\begin{proof}
Since both $0^n$ and $S(0^n)$ are not solutions, we have
	$V(0^n)<V(S(0^n))<V(S(S(0^n)))$, $P(S(0^n))=0^n$, and for $\uu = S(0^n)$,
	$S(P(\uu))=\uu$ and $P(S(\uu))=\uu$. In other words, $0^n \rightarrow S(0^n)
	\rightarrow S(S(0^n))$ are valid edges, and since $V(0^n)=0$, we have
	$V(S(S(0^n))\ge 2$. 
\end{proof}

Let us assume now on that $0^n$ and $S(0^n)$ are not solutions of $\CI$, and
then by Lemma \ref{lem:valid-edges}, we have $0^n \rightarrow S(0^n) \rightarrow
S(S(0^n))$ are valid edges, and $V(S(S(0^n))\ge 2$. We can avoid the need to check
whether $V(S(0))$ is one all together, by making $0^n$ point directly to
$S(S(0^n))$ and make $S(0^n)$ a dummy vertex. 

We first construct $S'$ and $P'$, and then construct $V'$ which will give
value zero to all self loops, and use the least significant $m$ bits to give a
value to all other vertices.
Before describing $S'$ and $P'$ formally, we first describe the underlying
principles. Recall that in $\CI$ vertex set is $\{0,1\}^n$ and possible potential values are $\{0,\dots,2^m-1\}$, while in $\CI'$ vertex set is $\{0,1\}^k$ where $k=m+n$. 
We will denote a vertex of $\CI'$ by a tuple $(\uu,\pi)$, where $\uu \in
\{0,1\}^n$ and $\pi\in \{0,\dots,2^m-1\}$. 
Here when we say that we introduce an {\em edge $\xx\rightarrow \yy$} we mean
that we introduce a valid edge from $\xx$ to $\yy$, i.e., $\yy=S'(\xx)$ and $\xx=P(\yy)$. 
\begin{itemize}
\item Vertices of the form $(S(0^n),\pi)$ for any $\pi \in \{0,1\}^m$ and the vertex $(0^n,1)$ are
dummies and hence have self loops.
\item If $V(S(S(0^n))=2$ then we introduce an edge $(0^n,0)\rightarrow(S(S(0^n)),2)$, otherwise 
\begin{itemize}
\item for $p=V(S(S(0^n))$, we introduce the edges $(0^n,0)\ra (0^n,2)\ra (0^n, 3)\dots (0^n,p-1)\ra (S(S(0^n)),p)$.
\end{itemize}
\item If $\uu \ra \uu'$ valid edge in $\CI$ then let $p=V(\uu)$ and $p'=V(\uu')$
\begin{itemize}
\item If $p=p'$ then we introduce the edge $(\uu,p)\ra (\uu',p')$. 
\item If $p<p'$ then we introduce the edges $(\uu,p)\ra (\uu,p+1)\ra \dots\ra (\uu,p'-1)\ra (\uu',p')$.
\item If $p>p'$ then we introduce the edges $(\uu,p)\ra (\uu,p-1)\ra \dots\ra (\uu,p'+1)\ra (\uu',p')$.
\end{itemize}
\item If $\uu\neq 0^n$ is the start of a path, i.e., $S(P(\uu))\neq \uu$, then
make $(\uu,V(\uu))$ start of a path by ensuring $P'(\uu,V(\uu))=(\uu,V(\uu))$.
\item If $\uu$ is the end of a path, i.e., $P(S(\uu))\neq \uu$, then make
$(\uu,V(\uu))$ end of a path by ensuring $S'(\uu,V(\uu))=(\uu,V(\uu))$.
\end{itemize}

Last two bullets above remove singleton solutions from the system by making them
self loops. However, this can not kill all the solutions since there is a path
starting at $0^n$, which has to end somewhere. Further, note that this entire process ensures that no new start or end of a paths are introduced. 
\medskip
\medskip

\noindent{\bf Procedure $S'(\uu,\pi)$.} 
\vspace{-0.2cm}

\begin{enumerate}
\itemsep1mm
\item If ($\uu=0^n$ and $\pi=1$) or $\uu=S(0^n)$ then Return $(\uu,\pi)$. 
\item If $(\uu,\pi)=0^k$, then let $\uu'=S(S(0^n))$ and $p'=V(\uu')$. 
\begin{enumerate}
\itemsep1mm
\item If $p'=2$ then Return $(\uu',2)$ else Return $(0^n,2)$.
\end{enumerate}
\item If $\uu=0^n$ then
\begin{enumerate}
\itemsep1mm
\item If $2\le \pi<p'-1$ then Return $(0^n,\pi+1)$.
\item If $\pi=p'-1$ then Return $(S(S(0^n)),p')$.
\item If $\pi\ge p'$ then Return $(\uu,\pi)$.
\end{enumerate}
\item Let $\uu'=S(\uu)$, $p'=V(\uu')$, and $p=V(\uu)$. 
\item \label{itm:Sfive} If $P(\uu')\neq \uu$ or $\uu'=\uu$ then Return $(\uu,\pi)$
\item If $\pi=p=p'$ or ($\pi=p$ and $p'=p+1$) or $(\pi=p$ and $p'=p-1$) then Return $(\uu',p')$.
\item If $\pi<p\le p'$ or $p\le p'\le \pi$ or $\pi>p\ge p'$ or $p\ge p'\ge \pi$ then Return $(\uu,\pi)$
\item If $p<p'$, then if $p\le \pi<p'-1$ then Return $(\uu,\pi+1)$. If $\pi=p'-1$ then Return $(\uu',p')$.
\item If $p>p'$, then if $p \ge \pi>p'+1$ then Return $(\uu,\pi-1)$. If $\pi=p'+1$ then Return $(\uu',p')$.
\end{enumerate}
\medskip

\noindent{\bf Procedure $P'(\uu,\pi)$.} 
\vspace{-0.2cm}

\begin{enumerate}
\itemsep1mm
\item If ($\uu=0^n$ and $\pi=1$) or $\uu=S(0^n)$ then Return $(\uu,\pi)$. 
\item If $\uu=0^n$, then 
\begin{enumerate}
\itemsep1mm
\item If $\pi=0$ then Return $0^k$.
\item If $\pi<V(S(S(0^n)))$ and $\pi\notin \{1,2\}$ then Return $(0^n,\pi-1)$.
\item If $\pi<V(S(S(0^n)))$ and $\pi=2$ then Return $0^k$.
\end{enumerate}
\item If $\uu=S(S(0^n))$ and $\pi=V(S(S(0^n))$ then 
\begin{enumerate}
\itemsep1mm
\item If $\pi=2$ then Return $(0^n,0)$, else Return $(0^n,\pi-1)$. 
\end{enumerate}
\item If $\pi=V(\uu)$ then 
\begin{enumerate}
\itemsep1mm
\item Let $\uu'=P(\uu)$, $p'=V(\uu')$, and $p=V(\uu)$. 
\item If $S(\uu')\neq \uu$ or $\uu'=\uu$ then Return $(\uu,\pi)$
\item If $p=p'$ then Return $(\uu',p')$ 
\item If $p'<p$ then Return $(\uu',p-1)$ else Return $(\uu',p+1)$
\end{enumerate}
\item Else \% when $\pi \neq V(\uu)$
\begin{enumerate}
\itemsep1mm
\item Let $\uu'=S(\uu)$, $p'=V(\uu')$, and $p=V(\uu)$
\item If $P(\uu')\neq \uu$ or $\uu'=\uu$ then Return $(\uu,\pi)$
\item If $p'=p$ or $\pi<p< p'$ or $p<p'\le \pi$ or $\pi>p> p'$ or $p>p'\ge \pi$ then Return $(\uu,\pi)$
\item If $p<p'$, then If $p<\pi\le p'-1$ then Return $(\uu,\pi-1)$. 
\item If $p>p'$, then if $p> \pi\ge p'+1$ then Return $(\uu,\pi+1)$. 
\end{enumerate}
\end{enumerate}

As mentioned before, the intuition for the potential function procedure $V'$ is to return zero for self loops, return $1$ for $0^k$, and return the number specified by the lowest $m$ bits for the rest. 
\medskip
\medskip

\noindent{\bf Procedure $V'(\uu,\pi)$.} Let $\xx=(\uu,\pi)$ for notational convenience.
\vspace{-0.2cm}

\begin{enumerate}
\itemsep1mm
\item If $\xx=0^k$, then Return $1$. 
\item If $S'(\xx) = \xx$ and $P'(\xx)=\xx$ then Return $0$.
\item If $S'(\xx) \neq \xx$ or $P'(\xx)\neq \xx$ then Return $\pi$.
\end{enumerate}

The fact that procedures $S'$, $P'$ and $V'$ give a valid \EOML instance follows from construction.
\begin{lemma}\label{lem:p2m-valid}
Procedures $S'$, $P'$ and $V'$ gives a valid \EOML instance on vertex set $\{0,1\}^k$, where $k=m+n$ and $V':\{0,1\}^k\ra \{0,\dots, 2^k-1\}$.
\end{lemma}

The next three lemmas shows how to construct a solution of \EOPL instance $\CI$ from a type T1, T2, or T3 solution of constructed \EOML instance $\CI'$.
The basic idea for next lemma, which handles type T1 solutions, is that we never create spurious end or start of a path. 
\begin{lemma}\label{lem:p2m-t1}
Let $\xx=(\uu,\pi)$ be a type T1 solution of constructed \EOML instance $\CI'$. Then $\uu$ is a type \solnref{R1} solution of the given \EOPL instance $\CI$.
\end{lemma}

\begin{proof}
Let $\Delta=2^m-1$.
In $\CI'$, clearly $(0^n,\pi)$ for any $\pi \in {1,\dots, \Delta}$ is not a start or end of a path, and $(0^n,0)$ is not an end of a path. Therefore, $\uu\neq 0^n$. Since $(S(0^n),\pi), \forall \pi\in \{0,\dots,\Delta\}$ are self loops, $\uu \neq S(0^n)$.

If to the contrary, $S(P(\uu))=\uu$ and $P(S(\uu))=\uu$. If $S(\uu)=\uu=P(\uu)$ then $(\uu,\pi),\ \forall \pi\in\{0,\dots,\Delta\}$ are self loops, a contradiction. 
\spencer{I don't understand this line. Basically, the point found can't have
corresponded to a self-loop because it would then have been a self-loop too.}

For the remaining cases, let $P'(S'(\xx))\neq \xx$, and let $\uu'=S(\uu)$. 
\spencer{There must have been a edge in the original if there is a successor edge
in the new instance}. There is a valid edge from $\uu$ to $\uu'$ in $\CI$. Then
we will create valid edges from $(\uu,V(\uu))$ to $(S(\uu),V(S(\uu))$ with
appropriately changing second coordinates. The rest of $(\uu,.)$ are self loops,
a contradiction. 

Similar argument follows for the case when $S'(P'(\xx))\neq \xx$. 
\end{proof}

The basic idea behind the next lemma is that a T2 type solution in $\CI'$ has
potential $1$. Therefore, it is surely not a self loop. Then it is either an end of a path or near an end of a path, or else near a potential violation. 

\begin{lemma}\label{lem:p2m-t2}
Let $\xx=(\uu,\pi)$ be a type T2 solution of $\CI'$. Either $\uu \neq 0^n$ is start of a path in $\CI$ (type \solnref{R1} solution), or $P(\uu)$ is an \solnref{R1} or \solnref{R2} type solution in $\CI$, or $P(P(\uu))$ is an \solnref{R2} type solution in $\CI$.
\end{lemma}

\begin{proof}
Clearly $\uu \neq 0^n$, and $\xx$ is not a self loop, i.e., it is not a dummy vertex with irrelevant value of $\pi$. Further, $\pi=1$. If $\uu$ is a start or end of a path in $\CI$ then done. 

Otherwise, if $V(P(\uu))>\pi$ then we have $V(\uu)\le \pi$ and hence $V(\uu)-V(P(\uu))\le 0$ giving $P(\uu)$ as an \solnref{R2} type solution of $\CI$. 
If $V(P(\uu))<\pi=1$ then $V(P(\uu))=0$. Since potential can not go below zero, either $P(\uu)$ is an end of a path, or for $\uu''=P(P(\uu))$ and $\uu'=P(\uu)$ we have $\uu'=S(\uu'')$ and $V(\uu')-V(\uu'')\le 0$, giving $\uu''$ as a type \solnref{R2} solution of $\CI$.
\end{proof}

At a type T3 solution of $\CI'$ potential is strictly positive, hence these solutions are not self loops. If they correspond to potential violation in $\CI$ then we get a type \solnref{R2} solution. But this may not be the case, if we made $S'$ or $P'$ self pointing due to end or start of a path respectively. In that case, we get a type \solnref{R1} solution. The next lemma formalizes this intuition. 

\begin{lemma}\label{lem:p2m-t3}
Let $\xx=(\uu,\pi)$ be a type T3 solution of $\CI'$. If $\xx$ is a start or end of a path in $\CI'$ then $\uu$ gives a type \solnref{R1} solution in $\CI$. Otherwise $\uu$ gives a type \solnref{R2} solution of $\CI$.
\end{lemma}

\begin{proof}
Since $V'(\xx)>0$, it is not a self loop and hence is not dummy, and $\uu\neq 0^n$. If $\uu$ is start or end of a path then $\uu$ is a type \solnref{R1} solution of $\CI$. Otherwise, there are valid incoming and outgoing edges at $\uu$, therefore so at $\xx$. 

If $V((S(\xx))-V(\xx)\neq 1$, then since potential either remains the same or increases or decreases exactly by one on edges of $\CI'$, it must be the case that $V(S(\xx))-V(\xx)\le 0$. This is possible only when $V(S(\uu))\le V(\uu)$. Since $\uu$ is not an end of a path we do have $S(\uu)\neq \uu$ and $P(S(\uu))=\uu$. Thus, $\uu$ is a type T2 solution of $\CI$.

If $V((\xx)-V(P(\xx))\neq 1$, then by the same argument we get that for $(\uu'',\pi'')=P(\uu)$, $\uu''$ is a type \solnref{R2} solution of $\CI$. 
\end{proof}

Our main theorem follows using Lemmas \ref{lem:p2m-valid}, \ref{lem:p2m-t1}, \ref{lem:p2m-t2}, and \ref{lem:p2m-t3}.

\begin{theorem}\label{thm:p2m}
An instance of \EOPL can be reduced to an instance of \EOML in polynomial time such that a solution of the former can be constructed in a linear time from the solution of the latter. 
\end{theorem}


\section{Proofs for Section~\ref{sec:opdc2ufeopl}: OPDC to UEOPL}

\subsection{Proof of Lemma~\ref{lem:dcm2ufeopl}}
\label{app:dcm2ufeopl}

Throughout this proof, we will fix $\mathcal{D} = (D_i)_{i = 1, \dots, d}$ to be
the direction functions, and $P = P(k_1, k_2, \dots, k_d)$ to be the set of
points used in the \OPDC instance. We will produce a \UFEOPL instance $L = (S,
V)$.

\paragraph{\bf The circuit $S$.}

A vertex of the line is a tuple $(p_0, p_1, p_2, \dots, p_d)$, where each $p_i \in P
\cup \{\vblank\}$ is either a point or a special symbol, $\vblank$, that is used
to indicate an unused element.
We use
$\vert = (P \cup \{\vblank\})^{d+1}$ to denote the set of possible vertices.
Only some of the tuples are valid encodings for a vertex. 
To be valid, a vertex $(p_0, p_1, p_2, \dots, p_d)$ must obey the following rules:
\begin{enumerate}
\item If $p_i \ne \vblank$, then $D_j(p_i) = \zero$ for all
$j \le i$. This means that if $p_i$ is a point, then
it must be a point on the $i$-surface.

\item If $p_i \ne \vblank$, then we must have $D_{i+1}(p_i) \ne \down$.

\item If $p_i \ne \vblank$ and $p_j = \vblank$ and $i < j$, then we must have $(p_i)_{j+1} =
0$.

\item If $p_i \ne \vblank$ and $p_j \ne \vblank$ and $i < j$, then we must have
$(p_i)_{j+1} = (p_j)_{j+1} + 1$. 


\end{enumerate}
We define
the function $\isvertex : \vert \rightarrow \{\mathtt{true}, \mathtt{false}\}$
that determines whether a given $v \in \vert$ is a valid encoding of a vertex, by
following the rules laid out above. This can clearly be computed in polynomial
time.

The initial vertex, which will be mapped to the bit-string $0^n$, will be
$(p_\text{init}, \vblank, \dots, \vblank)$, where $p_\text{init} = (0, 0, \dots,
0)$ is the all zeros point in $P$.

Given a vertex encoding $v = (p_0, p_1, p_2, \dots, p_d) \in \vert$, the circuit
$S$ carries out the following operations.
If $\isvertex(v)$ is \texttt{false}, then $S(v) = v$, indicating that $v$
is indeed not a vertex. Otherwise, we use the following set of rules to
determine the successor of $v$.
Let $i$ be the smallest index such that $p_i \ne \vblank$. 

\begin{enumerate}
\item If $i = d$ then our vertex has the form $v = (\vblank, \dots, \vblank,
p_d)$, and $p_d$ is on the $d$-surface, meaning that it is a solution to the
discrete contraction map. So we set $S(v) = v$ to ensure that this
is a solution.

\item If $D_{i+1}(p_i) = \zero$, then we define $S(v) = v'$ where
\begin{equation*}
v'_j = \begin{cases}
\vblank & \text{if $j < i+1$,} \\
p_i & \text{if $j = i+1$,} \\
p_j & \text{if $j > i+1$.}
\end{cases}
\end{equation*}
This operation overwrites the point in position $i+1$ with $p_i$, and sets
position $i$ to $\vblank$. All other components of $v$ are unchanged.

\item If
$D_{i+1}(p_i) \ne \zero$ and
$i > 0$ then let $q$ be the point such that 
\begin{equation*}
(q)_j = \begin{cases}
0 & \text{if $j < i+1$,} \\
(p_i)_{i+1}+1 & \text{if $j = i + 1$,} \\
(p_i)_j & \text{if $j > i + 1$.}
\end{cases}
\end{equation*}

\begin{enumerate}
\item If $q$ is a point in $P$, then we define $S(v) = (q, p_1, p_2, \dots, p_d)$.  
\item Otherwise, we must have that $(p_i)_{i+1} = k_{i+1}$, meaning that $p_i$
is the last point of the grid. This means that we have a solution of type
\solnref{OV3}, since the fact that $\isvertex(v) = \mathtt{true}$ implies that
$D_{i+1}(p_i) = \up$. So we set $S(v) = v$.
\end{enumerate}

\item 
If $D_{i+1}(p_i) \ne \zero$ and $i = 0$ then let $q$ be the point such that $(q)_j = (p_0)_j$ for all $j >
1$, and $(q)_1 = (p_0)_1 + 1$.
\begin{enumerate}
\item 
If $q$ is in the point set $P$, then we define $S(v) = (q, p_1, p_2, \dots, p_d)$.

\item If $q$ is not in $P$, then we again have a solution of type \solnref{OV3},
since $(p_0)_1 = k_1$, and $D_1(p_0) = \up$ from the fact that $\isvertex(v) =
\mathtt{true}$. So we set $S(v) = v$.
\end{enumerate}
\end{enumerate}
 
\paragraph{\bf The potential function.}

To define the potential function, we first define an intuitive potential that
uses a tuple of values ordered lexicographically, and then translate this into a
circuit $V$ that produces integers. We'll call the tuple of values the
lexicographic potential associated with a vertex~$v$, and denote it using $\lexpot(v)$. To define the lexicographic potential, we'll need to introduce an auxiliary function $\potf : (P\cup \Set{\vblank})\times \Set{0,\dotsc,d+1} \to \Z$ given by

\begin{equation*}
\potf(p, i) =  \begin{cases}
(p)_{i+1} + 1& \text{if $p \ne \vblank$,} \\
0 & \text{otherwise.}
\end{cases}
\end{equation*}
The lexicographic potential of $v = (p_0,p_1,\dotsc,p_d) \in \vert$ is the following:
\[ \lexpot(v) = \Paren{\potf(p_0,0),\potf(p_1,1),\dotsc,\potf(p_{d-1},d-1)}\text{.} \]
Note that the $\lexpot(v) \in \Z^{d}$, since the definition ignores $p_d$.

Let $\prec_{d}$ be the ordering on tuples from 
$\Z^{d}$ where they are compared lexicographically from right to left, so that $(0, 0) \prec_2 (1, 0) \prec_2 (0, 1) \prec_2 (1, 1)$, for example.
Our potential function will be defined by the tuples given by $\lexpot$ and the
order $\prec_{d+1}$. We omit the subscript from $\prec$ whenever it is clear
from the context.

Given a vertex $v = (p_0, p_1, \dots, p_d) \in \vert$, let $\lexpot(v) = (l_0,\dotsc,l_{d-1})$. To translate from lexicographically ordered tuples to integers in a way that preserves the ordering, we pick some integer $k$  such that $k > k_i$ for
all $i$, meaning that $k$ is larger than the grid-width used in every dimension,
which implies that $l_j < k$ for all $j$. We now take a weighted sum of the coordinates of $\lexpot(v)$ where the weight for coordinate $i$ is $k^i$, so that the $i$th coordinate dominates coordinates $0$ through $i-1$. The final potential value $V : \vert \to \Z$ is then given by

\begin{equation*}
V(v) = \sum_{i = 0}^{d-1} k^il_i,
\end{equation*} which can be easily computed in polynomial time. This completes the definition of the \UFEOPL problem $(S, V)$.

\begin{lemma}
Every solution of the \UFEOPL instance $(S, V)$ can be used to find a solution
of the \OPDC instance given by $\mathcal{D}$ and $P$. Furthermore, solutions
of type \solnref{O1} are only ever mapped onto solutions of type \solnref{UF1}.
\end{lemma}
\begin{proof}
We begin by considering solutions of type \solnref{UF1}. Let $x =
(p_0, p_1, \dots, p_d)$ and let $y = S(x)$, and suppose that $x$ is a \solnref{UF1}
solution. This means that 
$S(x) \ne x$ and either $S(y) = y$ or
$V(y) \le V(x)$. We first suppose that $S(y) = y$, and note that in this case we
must have $\isvertex(x) = \mathtt{true}$ while $\isvertex(y) = \mathtt{false}$. 
We have the following cases based on the rule used to determine $S(x)$.
\begin{enumerate}
\item If $S(x)$ is determined by the first rule in the definition of $S$, then
this means that $p_d \ne \vblank$. Since $\isvertex(x) = \mathtt{true}$, this
means that $D_i(p_d) = \zero$ for all $i$, which means that $p_d$ is a solution
of type \solnref{O1}. 

\item If $S(x)$ is determined by the second rule in the
definition of $S$, then there are two cases.  Let $i$ be the smallest index such
that $p_i \ne \vblank$.

\begin{enumerate}
\item If $D_{i+2}(p_i) = \down$ then we have a solution of type \solnref{OV2} or \solnref{OV3}.
\begin{enumerate}
\item If $p_{i+2} \ne \vblank$ then $p_i$ and $p_{i+2}$ are solution of type
\solnref{OV2}. Specifically, this holds because $D_{i+2}(p_{i+2}) = \up$ while
$D_{i+2}(p_i) = \down$, and $(p_i)_{i+2} = (p_{i+2})_{i+2} + 1$ holds because
$\isvertex(x) = \mathtt{true}$. Moreover, $D_j(p_i) = D_j(p_{i+2}) = \zero$, for
all $j < i+2$, where in particular the fact that $D_{i+1}(p_i) = \zero$ is given
by the fact that we are in the second case of the definition of $S$.

\item If $p_{i+2} = \vblank$ then this means that $(p_i)_{i+2} = 0$. Since 
$D_{i+2}(p_i) = \down$ this gives us a solution of type \solnref{OV3}.

\end{enumerate}

\item If $D_{i+2}(p_i) \ne \down$, then we argue that this case is impossible.
Specifically, we will show that $\isvertex(y) = \mathtt{true}$, meaning that
$S(y) \ne y$. To do this, we will prove that the four conditions of $\isvertex$
hold for $y$. Note that $y$ differs from $x$ only in positions $i$ and $i+1$,
and that position $i$ of $y$ is $\vblank$. So we only need to consider the
conditions imposed by $\isvertex$ when the point $p_{i}$ is placed in position
$i+1$.
\begin{enumerate}
\item The first condition of $\isvertex$ is that $p_{i}$ should be on the
$(i+1)$-surface, which is true because the second rule of $S$ explicitly checks
that $D_{i+1}(p_i) = \zero$, while the fact that $\isvertex(x) = \mathtt{true}$
guarantees that $D_{j}(p_i) = \zero$ for all $j < i+1$. 

\item The second condition requires that $D_{i+2}(p_i) \ne \down$, which is true
by assumption.

\item Every constraint imposed by the third and fourth conditions also holds for
$p_i$ in $x$, and so the fact that $\isvertex(x) = \mathtt{true}$ implies that
these  conditions hold for $y$.
\end{enumerate}


\end{enumerate}

\item If $S(x)$ is determined by the third rule defining $S$, then we 
have two cases. Since the third rule was used, we know that $y = (q, p_1, p_2,
\dots, p_d)$, with the definition of $q$ being given in the third rule.

\begin{enumerate}

\item If $q$ is not in $P$, then we have a solution of
type \solnref{OV3}, as described in the algorithm for~$S$.

\item If $q \in P$ and $D_1(q) = \down$, then we have a solution of type \solnref{OV3}, since $q_1 = 0$
by definition.

\item If $q \in P$ and $D_1(q) \ne \down$ and $q \in P$, then we argue that the
case is impossible, and we prove this by showing that $\isvertex(y) =
\mathtt{true}$. Note that $y$ differs from $x$ only in the position occupied by
$q$, and so this is the only point for which we need to prove the conditions,
since all the other points satisfy the conditions by the fact that $\isvertex(x)
= \mathtt{true}$.
\begin{enumerate}
\item The first requirement of $\isvertex(y)$ holds trivially, since the only
new requirement is that $q$ is on the $0$-surface, and every point is on the
$0$-surface by definition. 

\item The second requirement is that $D_1(q) \ne \down$, which is true by
assumption.

\item The third and fourth conditions place constraints on certain coordinates
of $q$. For coordinates $j < i$, the third condition requires that $q_j = 0$,
which is true by definition, while the fourth condition is inapplicable. For coordinates $j \ge i$, the constraints imposed by the third and fourth conditions 
hold because $q_j = (p_i)_j$ in these coordinates, and $p_i$ also satisfies
these constraints.
\end{enumerate}

\end{enumerate}

\item If $S(x)$ is determined by the fourth rule, then we have two cases. Let
$y = (q, p_1, p_2, \dots, p_d)$ be the value of $y$ produced by the fourth rule,
where the definition of $q$ is given in that rule. 

\begin{enumerate}
\item If $q$ is not in $P$, then we have a solution of
type \solnref{OV3}, as described in the algorithm for~$S$.

\item If $q \in P$ and $D_1(q) = \down$ then we have a solution of type \solnref{OV2}. Specifically, the
points $p_0$ and $q$ provide the violation since $(q)_1 = (p_0)_1 + 1$, while
$D_1(p_0) = \up$ and $D_1(q) = \down$. The fact that both $q$ and $p_0$ belong
to the same $1$-slice is guaranteed by the definition of $q$.

\item If $q \in P$ and $D_1(q) \ne \down$ then we again argue that $\isvertex(y)
= \mathtt{true}$, making this case impossible. The reasoning is the same as the
reasoning used in case 3b.
\end{enumerate}

\end{enumerate}

We now proceed to the case where we have a solution $x = (p_0, p_1, \dots, p_d)$
of type \solnref{UF1} and the vertex $y = S(x)$ satisfies $S(y) \ne y$. In this case, we
must have $V(y) \le V(x)$. We argue that this is impossible, and again we will
do a case analysis based on the rule used to determine the output of the circuit
$S$.

\begin{enumerate}
\item If $S(x)$ is determined by the first rule then we have $\isvertex(y) =
\mathtt{false}$, which is not possible in this case.

\item If $S(x)$ is determined by the second rule, then we can prove that $V(y) >
V(x)$. This is because $y$ differs from $x$ only in positions $i$ and $i+1$, and
because $(p_i)_{i+1} = (p_{i+1})_{i+1} + 1$ by the fourth rule of
$\isvertex(x)$. Hence
\begin{equation*}
k^i \cdot
\potf(\vblank, i) + k^{i+1} \cdot \potf(p_{i}, i+1) > k^i \cdot \potf(p_i, i) + k^{i+1} \cdot \potf(p_{i+1}, i+1),
\end{equation*}
where we are using the fact that $k^{i+1} > k^i\cdot \potf(p_i,i)$. Thus, $V(y) > V(x)$.

\item If $S(x)$ is determined by the third rule, then note that we must have $q
\in P$. We again argue that $V(y) >
V(x)$. Specifically, observe that $V(y) = V(x) + \potf(q, 0) = V(x) + 1$. 

\item If $S(x)$ is determined by the fourth rule, then note that we must have $q
\in P$, and we also have $V(y) >
V(x)$. Specifically
\begin{align*}
V(y) &= V(x) - \potf(p_0, 0) + \potf(q, 0) \\
&= V(x) - (p_0)_1 + q_1 \\
&= V(x) + 1.
\end{align*}

\end{enumerate}

Finally, we move to the case where we have a solution of type \solnref{UFV1}. In this
case we have two vertices $x = (p_0, p_1, \dots, p_d)$ and
$y = (q_0, q_1, \dots, q_d)$ for which $\isvertex(x) = \isvertex(y) =
\mathtt{true}$, and for which $V(x) \le V(y) < V(S(x))$.  
We will once again perform a case analysis over the possible cases of $S$.
\begin{enumerate}
\item $S(x)$ cannot be determined by the first case in the definition of $S$,
because that case only applies when $\isvertex(y) = \mathtt{false}$.

\item If $S(x)$ is determined by the second case in the definition of $S$, then
we observe that we have $\lexpot(x) = (0, \dots, 0, v_i, v_{i+1}, \dots, v_d)$
and $\lexpot(S(x)) = (0, \dots, 0, 0, v_{i+1}+1, v_{i+2} \dots, v_d)$, ie., the
$i$th element of $\lexpot(x)$ is replaced by $0$, and the $(i+1)$th element is
replaced by $v_{i+1}+1$. Note also that elements $i+2$ through $d$ of
$\lexpot(S(x))$ agree with the corresponding elements of $\lexpot(x)$.

Since we have $V(x) \le V(y) < V(S(x))$ this means that $\lexpot(x) \prec
\lexpot(y) \prec \lexpot(S(x))$. Hence, $\lexpot(y)$ must have the form $(v'_0, v'_1, \dots, v'_i,  v_{i+1},
v_{i+2}, \dots, v_d)$, meaning that it agrees with $\lexpot(x)$ and
$\lexpot(S(x))$ on elements $i+1$ through $d$. Furthermore, we must have $v'_i
\ge v_i$.

Let $j$ be the largest index satisfying $j \le i$ and
$v'_j \ne v_j$. Note that such a $j$ must exist, since otherwise we would have
$x = y$, which would contradict the fact that $x \ne y$ in any solution of type
UFV1. We claim that $p_j$ and $q_j$ form a solution of type \solnref{OV1}. 

Let $s$ be the $j$-slice that satisfies $s_l = \blank$ for all $l \le j$ and
$s_l = (p_j)_l$ for all $l > j$. The point $p_j$ lies in the slice $s$ by
definition. We claim that $q_j$ also lies in this slice, which follows from the
fact that $v_l = v'_l$ for all $l > j$, combined with the third and fourth
properties of $\isvertex(x)$ and $\isvertex(y)$. Note also that $(p_j)_j \ne
(q_j)_j$, and hence $p_j \ne q_j$.

Finally, the first property of 
$\isvertex(x)$ and $\isvertex(y)$ imply that $D_l(p_j) = \zero$ 
and $D_l(q_j) = \zero$ for all $l \le j$. So $p_j$ and $q_j$ are two distinct
fixpoint of the $j$-slice $s$, meaning that we have a solution of type \solnref{OV1}.

\item We claim that $S(x)$ cannot be determined by the third rule in the
definition of $S$. In this case we have $\lexpot(x) = (0, v_1, v_2, \dots,
v_d)$, since the case is only applicable when $p_0 = \vblank$. Observe that
$\lexpot(S(x)) = (1, v_1, v_2, \dots, v_d)$, and that there is no possible tuple
$t$ that satisfies $\lexpot(x) \prec t \prec \lexpot(S(x))$. This means that we
must have $V(y) = V(x)$, but this is only possible if $y = x$, due to the
constraints placed by the third and fourth properties of $\isvertex$. Hence this
case is not possible, since $y = x$ is specifically ruled out in a solution of
type UVF1.

\item For similar reasons, we claim that $S(x)$ cannot be determined by the
fourth rule. In this case we have 
$\lexpot(x) = (v_0, v_1, v_2, \dots, v_d)$,
and
$\lexpot(S(x)) = (v_0+1, v_1, v_2, \dots, v_d)$, which again means that there
cannot be a tuple $t$ satisfying 
$\lexpot(x) \prec t \prec \lexpot(S(x))$. Hence we would have $V(x) = V(y)$ and
$x = y$ as in the previous case, which is impossible.
\end{enumerate}
To complete the proof, we note that solutions of type \solnref{O1} are only ever mapped
onto solutions of type \solnref{UF1}.
\end{proof}

This proves that our reduction from \OPDC to \UFEOPL is correct. The fact that
the reduction is promise-preserving follows from the fact that all solutions of
type \solnref{O1} are mapped onto solutions of type \solnref{UF1}. Hence, if we promise that there
are no violations in the OPDC instance, then the resulting UFEOPL instance can
only have solutions of type \solnref{UF1}. This completes the proof of
Lemma~\ref{lem:dcm2ufeopl}.

\subsection{Proof of Lemma~\ref{lem:uf2ufp1}}
\label{app:uf2ufp1}

We provide a polynomial-time
promise-preserving reduction from \UFEOPL to \UFEOPLp1. This reduction uses
essentially the same idea as the reduction from \EOPL to \EOML
given in Theorem~\ref{thm:eoml2eopl}. 

Let $L = (S, V)$ be an instance of \UFEOPL, and let $n$ be the bit-length of the
strings used to represent vertices in $L$. If $x$ is a vertex in $L$ such that
$V(S(x)) > V(x) + 1$, then we introduce new vertices between $x$ and $S(x)$,
each of which increases in potential by 1.

Formally, our \UFEOPLp1 instance will be denoted as $L' = (S', V')$. Each vertex
in this instance will be a pair $(v, i)$, where $v$ is a vertex from $L$, and
$i$ is an integer between $0$ and $2^n$. The circuit $S'$ is defined in the
following way. For each vertex $(v, i)$ we use the following algorithm.

\begin{enumerate}
\item If $S(v) = v$, then $S'(v, i) = (v, i)$, meaning that if $v$ is not a
vertex in $L$, then $(v, i)$ is not a vertex in $L'$ for all $i$.


\item If $S(v) \ne v$ and $V(S(v)) > V(v) + i + 1$, then $S'(v, i) = (v, i+1)$.

\item If $S(v) \ne v$ and $V(S(v)) = V(v) + i + 1$, then $S'(v, i) = (S(v), 0)$.

\item If $S(v) \ne v$ and $V(S(v)) < V(v) + i + 1$, then $S'(v, i) = (v, i)$.
\end{enumerate}
The last three conditions add a new sequence of vertices between any edge $(x,
y)$ where $V(y) > V(y) + 1$. Specifically, this sequence is
\begin{equation*}
(x, 0) \rightarrow (x, 1) \rightarrow \dots \rightarrow (x, V(y) - V(x) - 1)
\rightarrow (y, 0).
\end{equation*}
Any pair $(v, i)$ that is not used in such a sequence has $S'(v, i) = (v, i)$,
meaning that it is not a vertex.

The potential function is defined as follows. For each vertex $(v, i)$, we use
the following algorithm.
\begin{enumerate}
\item If $S(v, i) = (v, i)$, then the potential of $(v, i)$ is irrelevant.

\item Otherwise, we set $V'(v, i) = V(v) + i$.
\end{enumerate}
This completes the specification of the reduction. 

Clearly the reduction can be carried out in polynomial time. We now prove that
this is a promise-preserving reduction.

\begin{lemma}
Every solution of $L'$ can be mapped to a solution of $L$. Furthermore,
solutions of type \solnref{UF1} of $L$ are only ever mapped onto solutions of type \solnref{UFP1} in
$L'$.
\end{lemma}
\begin{proof}
We start by considering solutions of type \solnref{UFP1}. Let $x = (v, i)$ be a vertex,
and let $y = (u, j)$ be the vertex with $y = S(x)$. In a solution of type \solnref{UFP1}
we have that $S'(x) \neq x$ and either $S'(y) = y$ or $V'(y) \ne V'(x) + 1$.
Hence, there are two cases to consider.
\begin{enumerate}
\item If $S'(y) = y$ then we must have that $u = S(v)$ and $j = 0$, since if
$(v, i)$ is a vertex, then $S'$ will always either give $(v, i+1)$ or $(S(v),
0)$, and it suffices to note that if $S'(v, i) = (v, i+1)$ then $S'(v, i+1) \ne
(v, i+1)$.

So, we have $y = (S(v), 0)$, and $S'(y) = y$. This can only occur in the case
where either $S(S(v)) = S(v)$, or $V(S(S(v))) < V(S(v))$. Both of these cases
yield a solution of type \solnref{UF1} for $L$.

\item We claim that the case  $V'(y) \ne V'(x) + 1$ is not possible. Since we
have already dealt with the case where $S'(y) = y$, we can assume that $y$ is a
vertex. 
Note that $S'(x)$ cannot be determined by cases 1 or 4 of the algorithm, since
in those cases $x$ would not be a vertex.
If $S'(x)$ is determined by case 2 of the algorithm, then we have that
$V'(y) = V'(x) + 1$ by definition. If 
$S'(x)$ is determined by case 3 of the algorithm, then we have 
\begin{equation*}
V'(x) = V(v) + (V(u) - V(v) - 1) = V'(y) - 1
\end{equation*}
Hence, this case is impossible by construction.
\end{enumerate}
Thus, we have dealt with all possible solutions of type \solnref{UFP1}.

We now consider a violation of type \solnref{UFPV1}. In this case we have two vertices $x
= (v, i)$ and $y = (u, j)$ such that $x \ne y$, $x \ne S'(x)$, $y \ne S'(y)$,
and $V'(x) = V'(y)$. We claim that this gives us a solution of type \solnref{UFV1}.
If $i = j$ then $V(v) = V(u)$, and so we have a solution of type \solnref{UFV1} in $L$.
Assume, without loss of generality, that $i \le j$. We have
\begin{equation*}
V(v) + i = V'(v, i) = V'(u, j) = V(u) + j,
\end{equation*}
which implies that $V(u) \le  V(v)$. Furthermore, we have that 
\begin{equation*}
V(S(u)) > V(u) + j \ge V(v) + i \ge V(v),
\end{equation*}
where the first inequality arises from the fact that $(u, j)$ is a valid vertex
in $L'$. Hence, we have shown that $V(u) \le V(v) < V(S(u))$, which is exactly a
solution of type \solnref{UFV1} in $L$.
\end{proof}

The lemma implies that our reduction is correct. The fact that it is
promise-preserving follows from the fact that solutions of type \solnref{UF1} are only
ever mapped onto solutions of type \solnref{UFP1}. Hence, if it is promised that $L$ only
has \solnref{UF1} solutions, then $L'$ must only have \solnref{UFP1} solutions. This completes the
proof of Lemma~\ref{lem:uf2ufp1}.


\subsection{Proof of Lemma~\ref{lem:ufeopl2eopl}}
\label{app:ufeopl2eopl}

The reduction of Hub\'{a}\v{c}ek and Yogev~\cite{hubavcek2017hardness} relies on
the \defineterm{pebbling game} technique that was first applied by Bitansky et
al~\cite{BPR15}.
Let $L = (S, W, x_s, T)$ be an \SOVL instance.
The pebbling game is played by placing \emph{pebbles} on the vertices of this
instance according to the following rules.
\begin{itemize}
\item A pebble may be placed or removed from the starting vertex $x_s$ at any time.
\item A pebble may be placed or removed on a vertex $x \ne x_s$ if and only if
there is a pebble on a vertex $y$ with $S(y) = x$.
\end{itemize}

Given~$n$ pebbles, how far can we move along the line by following these rules?
The answer is that we can place a pebble on vertex $2^n-1$ by applying the
following \emph{optimal strategy}. The strategy is recursive. In the base case,
we can place a pebble on vertex $2^1-1 = 1$ by placing our single pebble on the
successor of $x_s$. In the recursive step, where we have $n$ pebbles, we use the following approach:
\begin{enumerate}
\item Follow the optimal strategy for $n-1$ pebbles, in order to place a pebble
on vertex $2^{n-1}-1$.

\item Place pebble $n$ on the vertex $2^{n-1}$.

\item Follow the optimal strategy for $n-1$ pebbles \emph{backwards} in order to
reclaim the first $n-1$ pebbles.

\item Follow the optimal strategy for $n-1$ pebbles forwards, but starting from
the vertex $2^{n-1}$. This ends by placing a pebble on vertex $2^{n-1} + 2^{n-1}
- 1 = 2^n-1$.
\end{enumerate}
Step 3 above relies on the fact that the pebbling game is \emph{reversible},
meaning that we can execute any sequence of moves backwards as well as forwards.
At the beginning of Step 4, we have a single pebble on vertex $2^{n-1}$, and we
follow the optimal strategy again, but using $2^{n-1}$ as the starting point.

\paragraph{\bf The reduction from \UFEOPL to \UEOPL.}

To reduce \UFEOPL to \UEOPL, we play the optimal strategy for the pebbling game.
Note that, since that since every step of the pebbling game is reversible, this
gives us a predecessor circuit. 
We will closely follow the reduction given by Bitansky et al~\cite{BPR15} from
\SOVL to \EOL. Specifically, we will reduce an instance $L = (S, V)$ of \UFEOPL
to an instance $L' = (S', P', V')$ of \UEOPL.

A vertex in $L'$ will be a tuple of pairs $((v_1, a_1), (v_2, a_2), \dots, (v_n,
a_n))$ describing
the state of the pebbling game. Each $v_i$ is a bit-string, while each $a_i$ is
either 
\begin{itemize}
\item the special symbol $\vblank$, implying that pebble $i$ is not used and
that the bit-string $v_i$ should be disregarded, or
\item an integer such that $a_i = V(v_i)$, meaning that pebble $i$ is
placed on the vertex $v_i$.
\end{itemize}
Bitansky et al~\cite{BPR15} have produced circuits $S'$ and $P'$ that implement
the optimal strategy of the pebbling game for pebbles encoded in this way. The
only slight difference is that they reduce from \SOVL, but we can apply their
construction by creating the circuit $W$ so that $W(v, a) = 1$ if and only if
$V(v) = a$. We refer the reader to their work for the full definition of these
two circuits.

Hub\'{a}\v{c}ek and Yogev~\cite{hubavcek2017hardness} built upon this reduction
by showing that it is possible to give a potential function $V'$ for the
resulting instance. Specifically, their potential encodes how much progress we
have made along the optimal strategy, which it turns out, can be computed purely
from the current configuration of the pebbles. Their construction also
guarantees that the potential at each vertex always increases by $1$, meaning
that we have $V(S(x)) = V(x) + 1$ whenever $S(x)$ and $x$ are both vertices. We
refer the reader to their work for the full definition of the circuit $V'$.

\paragraph{\bf Violations.}

So far, we have a reduction from the promise version of \UFEOPL to \UEOPL, which
entirely utilizes prior work. Specifically, every solution of type \solnref{U1} in $L'$
will map back to a solution of type \solnref{UFP1} in $L$.

Our contribution is to handle the violations,
thereby giving a promise-preserving reduction from the non-promise version of
\UFEOPL to the non-promise version of \UEOPL.

\begin{lemma}
Every violation in $L'$ can be mapped back to a violation of $L$.
\end{lemma}
\begin{proof}
There are three types of violation in $L'$.
\begin{enumerate}
\item Violations of type \solnref{UV1}, which are edges where the potential decreases, are not possible, since the reduction of 
Hub\'{a}\v{c}ek and Yogev ensures that $V'(S'(x)) =
V'(x) + 1$ whenever $x$ and $S'(x)$ are both vertices. 

\item In violations of type \solnref{UV2} we have a vertex $((v_1, a_1), (v_2, a_2),
\dots, (v_n, a_n))$ that is the start of a second line. This means that, for
some reason, we are not able to execute the optimal strategy backwards from this
vertex. There are two possibilities 

\begin{enumerate}
\item The optimal strategy needs to place a pebble on the successor of some
vertex $v_i$, but it cannot because $v_i$ is the end of a line. This means that
either $S(v_i) = v_i$ or that $V(S(v_i)) \ne V(v_i)+1$, and in either case we
have a solution of type \solnref{UFP1} for $L$.

\item The optimal strategy needs to remove the pebble $v_i$, but it cannot,
because it does not have a pebble on a vertex $u$ with $S(u) = u$. By
construction, there will be some pebble $v_j$ with $a_j = a_i - 1$, but in this
case we have $S(v_j) \ne v_i$. This means that we have two lines, and
specifically we have that $v_i$ and $S(v_j)$ are two vertices with the same
potential, since $V(v_j) = a_j$ and $V(S(v_j)) = V(v_j) + 1$. This gives us a
solution of type \solnref{UFPV1}.
\end{enumerate}

\item In violations of type \solnref{UV3} we have two distinct vertices $x = ((v_1, a_1), (v_2,
a_2), \dots, (v_n, a_n))$ and $y = ((u_1, b_1), (u_2, b_2), \dots, (u_n, b_n))$
with $V'(x) \le V'(y) < V'(S'(x))$. Since the reduction ensures that $V'(S'(x))
= V'(x)
+ 1$ this means that $x$ and $y$ have the same potential. The reduction of
Hub\'{a}\v{c}ek and Yogev ensures that, if two vertices have the same
potential, then they refer to the same step of the optimal strategy, meaning
that
$a_i = b_i$ for all $i$. This means that any pair of vertices $v_i$ and
$u_i$ with $a_i \ne \vblank$ is a pair of vertices with $V(v_i) = V(u_i)$, and so
a solution of type \solnref{UFPV1}. To see that such a pair must exist, it suffices to
note that the only vertex with $a_i = \vblank$ for all $i$ is the start of the
line, and there cannot be two distinct vertices with this property.

\end{enumerate}
\end{proof}

The lemma above proves that the reduction is correct, but it does not directly
prove that it is promise-preserving. Specifically, in case 2a of the proof we
show that some violations of \solnref{UV2} are mapped back to solutions of type \solnref{UFP1}.
This, however, is not a problem, because we can argue that case 2a of the proof
can only occur if there is more than one line in $L'$.

Specifically, if we are at some vertex 
$x = ((v_1, a_1), (v_2, a_2), \dots, (v_n, a_n))$ and $P(x)$ needs to place a
pebble on $S(v_i)$, then $v_i$ cannot be the furthest point in the pebble
configuration, meaning that there is some $v_j$ with $a_j \ne \vblank$ and $a_j
> a_i$. This can be verified by inspecting the recursive definition of the
optimal strategy.
But note that if $v_i$ is the end of a line, and $v_j$ is a vertex with $V(v_j)
> V(v_i)$, then $L'$ must contain more than one line.

This allows us to argue that the reduction is promise-preserving, since if $L$
is promised to have no violations, then it must contain exactly one line, and if
$L$ contains exactly one line, then all proper solutions of $L$ are mapped onto
proper solutions of $L'$. Thus $L'$ will contain no violations. This completes
the proof of Lemma~\ref{lem:ufeopl2eopl}.

\section{Proofs for Section~\ref{sec:opdc-complete}: \UEOPL to \OPDC}

\subsection{Proof of Lemma~\ref{lem:2k}}
\label{app:2k}

This construction is very similar to the one used in the proof of
Lemma~\ref{lem:uf2ufp1} given in Appendix~\ref{app:uf2ufp1}, although here we
must deal with the predecessor circuit, and ensure that the end of each line has
the correct potential. Let $L = (S, P, V)$ be an instance of \UEOPL, and let $k$
be the bit-length of the vertices used in $L$. 

We will create an instance $L' = (S', P', V')$ in the following way. Each
vertex in $L'$ will be a pair $(v, i)$, where $v$ is a vertex in $L$, and $i$ is
an integer satisfying $i \le 2^n$, where $n = k+1$. Hence, a vertex in $L'$
can be represented by a bit-string of length $n + k$.

The circuit $S'$ is defined as follows. Given a vertex $(v, i)$, we execute the
following algorithm.
\begin{enumerate}
\item If $S(v) = v$, then $S'(v, i) = (v, i)$, meaning that if $v$ is not a
vertex in $L$, then $(v, i)$ is not a vertex in $L'$ for all $i$.

\item If $v \ne P(S(v))$ and $V(v) + i > 2^{n}-1$, then $S'(v, i) = S'(v, i)$,
meaning that $(v, i)$ is not a vertex in the instance.

\item If $v \ne P(S(v))$ and $V(v) + i = 2^{n}-1$, then $S'(v, i) = 0^k$, which
makes $(v, i)$ an end of line solution.

\item If $v \ne P(S(v))$ and $V(v) + i < 2^{n}-1$, then $S'(v, i) = (v, i+1)$.

\item If $S(v) \ne v = P(S(v))$ and $V(S(v)) > V(v) + i + 1$, then $S'(v, i) = (v, i+1)$.

\item If $S(v) \ne v = P(S(v))$ and $V(S(v)) = V(v) + i + 1$, then $S'(v, i) = (S(v), 0)$.

\item If $S(v) \ne v = P(S(v))$ and $V(S(v)) < V(v) + i + 1$, then $S'(v, i) =
(v, i)$, meaning that $(v, i)$ is not a vertex.
\end{enumerate}
The last three conditions add a new sequence of vertices between any edge $(x,
y)$ where $V(y) > V(y) + 1$. Specifically, this sequence is
\begin{equation*}
(x, 0) \rightarrow (x, 1) \rightarrow \dots \rightarrow (x, V(y) - V(x) - 1)
\rightarrow (y, 0).
\end{equation*}
Conditions 2 through 4 introduce a new line starting at $(x, 0)$, whenever $x$
is the end of a line in the original instance. This line has the form
\begin{equation*}
(x, 0) \rightarrow (x, 1) \rightarrow \dots \rightarrow (x, 2^n - V(x)-1).
\end{equation*}
where $(x, 2^n - V(x)-1)$ is the new end of line. Observe that this line is
always non-empty, since $2^n-1 =  2^{k+1}-1 > 2^k \ge V(x)$.

The predecessor circuit $P'$ walks the line backwards, which is easy to
construct, since we can either follow the predecessor circuit of $L$, or we can
walk backwards along any of the lines that we have introduced. Specifically,
given a vertex $(v, i)$, we use the following algorithm to implement $P'$.
\begin{enumerate}
\item If $S(v) = v$, then the value returned by $P'(v, i)$ is irrelevant, since
$(v, i)$ is not a vertex. 

\item If $v \ne P(S(v))$ and $V(v) + i > 2^{n}-1$, then again the value of $P'(v,
i)$ is irrelevant, since $(v, i)$ not a vertex. 

\item If $v \ne P(S(v))$ and $V(v) < V(v) + i \le 2^{n}-1$, then $P'(v, i) = (v,
i-1)$, meaning that if we are on the new line at a solution, then we walk the line backwards.

\item If $v \ne P(S(v))$ and $V(v) + i = V(v)$, then $P'(v, i) = (P(v), V(v) -
V(P(v)))$, meaning that if we are at $(v, 0)$, then we move to the end of the
line between $(P(v), 0)$ and $(v, 0)$.

\item If $S(v) \ne v = P(S(v))$ and $i = 0$, then $P'(v, i) = (P(v), V(v) -
V(P(v)))$.




\item If $S(v) \ne v = P(S(v))$ and $V(S(v)) < V(v) + i + 1$, then the value
returned by $P'(v, i)$ is irrelevant, since $(v, i)$ is not a vertex. 

\item If $S(v) \ne v = P(S(v))$ and cases 5 and 6 do not apply, then
$P'(v, i) = (v, i-1)$.
\end{enumerate}
Cases 2 through 4 deal with the new line introduced at the end of each line.
while cases 5 through 7 deal with the new lines introduced between the vertices
$(x, 0)$ and $(S(x), 0)$.

The potential function $V'$ is defined so that $V'(v, i) = V(v) + i$.

The two properties that we need to satisfy hold by construction for $L'$. Every
edge is constructed so that $V'(S(x)) = V'(x) + 1$, whenever $x$ is a vertex,
and the new lines starting at vertices $(v, 0)$, where $v$ is the end of a line
in $L$, ensure that $V'(x) = 2^n-1$ if and only if $x$ is the end of a line in $L'$. The
following lemma shows that the reduction is correct.

\begin{lemma}
Every solution of $L'$ can be mapped back to a solution of $L$.
\end{lemma}
\begin{proof}
We enumerate the types of solutions in \UEOPL. 
\begin{enumerate}
\item A solution of type \solnref{U1} in $L'$ is a vertex $x = (v, i)$ such that
$P'(S'(x)) \ne x$. By construction this can only occur if $P(S(v)) \ne v$, so
this gives us a solution of type \solnref{U1} for $L$.

\item A solution of type \solnref{UV1} gives us a vertex $x = (v, i)$ where $P'(S'(x)) = x$
and $V'(S'(x)) \le V'(x)$. By construction, this can only occur if $V(S(v)) \le
V(v)$, so we have a solution of type \solnref{UV1} for $L$.

\item A solution of type \solnref{UV2} gives us a vertex $x = (v, i)$ where $S'(P'(x)) \ne
x$. By construction, this can only happen if $S(P(x)) \ne p$ and $i = 0$. This
gives us a solution of type \solnref{UV2} for $L$.

\item A solution of type \solnref{UV3} gives us a pair of vertices $x = (v, i)$ and $y =
(u, j)$ such that
$x \ne y$, and either $V(x) = V(y)$, or $V(x) < V(y) < V(S(x))$. Note that the
latter case is impossible here, since by construction we have $V(S(x)) = V(x) +
1$ whenever $x$ is a vertex, so we must have $V(x) = V(y)$.

If $i = j$, then $V(v) = V(u)$, and so $v$ and $u$ form a solution of type \solnref{UV3}
for $L$. Otherwise, we will 
suppose, without loss of generality, that $i > j$, which means that $V(v) <
V(u)$. Moreover, we have $V(S(v)) > V(v) + i = V(u) + j \ge V(u)$. Hence we have
shown that $V(v) < V(u) < V(S(v))$, and so we have that $v$ and $u$ form a
solution of type \solnref{UV3} for $L$.
\end{enumerate}
\end{proof}

The lemma above implies that the reduction is correct. To see that it is
promise-preserving, it suffices to note that proper solutions of $L$ are only
ever mapped onto proper solutions of $L'$. Therefore, if it is promised that $L$
has no violations, then $L'$ will also have no violations. This completes the
proof of Lemma~\ref{lem:2k}.

\subsection{Proof of Lemma~\ref{lem:points}}
\label{app:points}

The proof requires us to define two polynomial-time algorithms. 

\paragraph{\bf Splitting lines around a vertex.}

We begin by defining the $\subline$ function. 
We will use two sub-functions that split an instance in two based on a
particular vertex. Let 
$L = (S, P, V)$ be a \UEOPL instance, and
$v$ be some vertex of $L$.

\begin{enumerate}
\item We define the function $\firsthalf(v, L)$ to return an \UEOPL instance
$(S', P', V')$ by removing every vertex with potential greater than or equal to
$V(v)$. Specifically, the $S'$ and $P'$ circuits will check whether $V(x) \ge
2^{n-1}$ for each vertex $x$, and if so, then they will set $S'(x) = P'(x) = x$,
which ensures that $x$ is not on the line. For any vertex $x$ with $V(x) <
2^{n-1}$, we set $S'(x) = S(x)$, $P'(x) = P(x)$ and $V'(x) = V(x)$. Note that
$S'$, $P'$, and $V'$ can all be produced in polynomial time if we are given $(S,
P, V)$.

\item We define the function $\secondhalf(v, L)$ to return a
\UEOPL instance
$(S', P', V')$ by removing every vertex with potential strictly less than
$V(v)$. 
For the circuits $S'$ and $P'$, this is done in the same way as the previous case, but this time the circuits
will check whether $V(x) < 2^{n-1}$. The function $V'$ is defined so that $V'(x)
= V(x) - 2^{n-1}$, thereby ensuring that the potentials are in the range $[0,
2^{n-1})$. Finally, the string $0^n$ is remapped to represent the vertex $v$,
which is the start of the second half of the line. In this case, we are able to
compute $S'$, $P'$, and $V'$ in polynomial time if we are given $(S, P, V)$ and
the bit-string $v$.
\end{enumerate}
We remark that, although we view these functions as intuitively splitting a
line, the definitions still work if $L$ happens to contain multiple lines.
Each line in the instance will be split based on the potential of $v$.

\paragraph{\bf The subline function.}

The subline function is defined recursively, based on the number of bit-strings
that are given to the function. In the base case we are given a line $L$ and
zero bit-strings, in which case $\subline(L) = L$. 

For the recursive step, assume that we are given bit-strings $v_i$ through
$v_n$. Let $L' = (S', P', V') = \subline(v_{i+1}, v_{i+2}, \dots, v_n, L)$. 
\begin{itemize}
\item If $V'(v_i) \ne 2^{i-1}$ then we set $\subline(v_i, v_{i+1}, \dots, v_n,
L) = \firsthalf(L')$.
\item If $V'(v_i) = 2^{i-1}$ then we set $\subline(v_i, v_{i+1}, \dots, v_n, L) = \secondhalf(L')$.
\end{itemize}
Note that since $\firsthalf$ and $\secondhalf$ can be computed out in polynomial
time, this means that $\subline$ can also be computed in polynomial time.

An important property of the reduction is that the output of $\subline(v_i,
v_{i+1}, \dots, v_n)$ is a \UEOPL instance in which the longest possible line
has length $2^{i-1}$. This can be proved by induction. For the base case note
that $\subline(L)$ allows potentials between $0$ and $2^n-1$. Each step of the
recursion cuts this in half, meaning that $\subline(v_i, v_{i+1}, \dots, v_n,
L)$ allows potentials between $0$ and $2^{i} - 1$. Since each edge in a line
must always strictly increase the potential, this means that the longest
possible line in 
$\subline(v_i, v_{i+1}, \dots, v_n, L)$ has length $2^{i-1}$. This holds even if
the instance has multiple lines.


\paragraph{\bf The decode function.}

Given a point $(v_1, v_2, \dots, v_n)$, let $L' = \subline(v_1, v_2, \dots, v_n,
L)$. As we have argued above, we know that $L'$ is an instance in which the
longest possible line has length $2^{1 - 1} = 1$. Hence, the starting vertex of
$L'$, which is by definition $v_1$, is the end of a line. So we set
$\decode(v_1, v_2, \dots, v_n) = v_1$. Since $\subline$ can be computed in
polynomial time, this means that $\decode$ can also be computed in polynomial
time. This completes the proof of Lemma~\ref{lem:points}.

\subsection{The formal definition of the direction functions.}
\label{app:direction}

We will extend the approach given in the main body to all dimensions. Let $p =
(v_1, v_2, \dots, v_n)$ be a point. For the dimensions $j$ in the range $m(i-1)
+ 1 \le j \le mi$ we use the following procedure to determine $D_j$. Let $L' =
(S', P', V') = 
\subline(v_{i+1}, v_{i+2}, \dotsc, v_n, L)$.

\begin{enumerate}
\item 
In the case where $V'(v_i) \ne 2^{i-1}$, meaning that $\decode(p)$ is a vertex in
the first half of $L'$, then there are two possibilities.

\begin{enumerate}
\item If $V'(\decode(p)) = 2^{i-1} - 1$, meaning that $p$ is the last
vertex on the first half $L'$, then
we orient the direction function towards the bit-string given by $S(\decode(p))$.  
So we set
\begin{equation*}
D_j(p) = \begin{cases}
\up & \text{if $p_j = 0$ and $S(\decode(p)) = 1$,}\\
\down & \text{if $p_j = 1$ and $S(\decode(p)) = 0$,}\\
\zero & \text{otherwise.}\\
\end{cases}
\end{equation*}

\item If the rule above does not apply, then we orient everything towards $0$ by
setting
\begin{equation*}
D_i(p) = \begin{cases}
\down & \text{if $p_j = 1$,}\\
\zero & \text{if $p_j = 0$.}
\end{cases}
\end{equation*}
\end{enumerate}

\item If $V'(v_i) = 2^{i-1}$, then we are in the second half of the line. In
this case we set $D_i(p) = \zero$. 
\end{enumerate}
Observe that this definition simply extends the idea presented in the main body
to all dimensions, using exactly the same idea.

\subsection{Proof of Lemma~\ref{lem:completeness}}
\label{app:completeness}

To prove this lemma, we must map every solution of the \OPDC instance given by
$P$ and $\mathcal{D}$ to a solution of the \UEOPL instance $L = (S, P, V)$.
We do this by enumerating the solution types for OPDC.

\begin{enumerate}
\item In solutions of type \solnref{O1} we have a point $p = (v_1, v_2, \dots, v_n)$ such
that $D_i(p) = \zero$ for all $i$. Since the dimensions corresponding to $v_n$
are all zero, we know that $\decode(p)$ is in the second half of the line, and
since the dimensions corresponding to $v_{n-1}$ are all zero, we know that
$\decode(p)$ is in the last quarter of the line. Applying this reasoning for all
dimensions allows us to conclude that $V(\decode(p)) = 2^{n} - 1$.
Lemma~\ref{lem:2k} implies that this can be true if and only if $\decode(p)$ is
the end of a line, which is a solution of type \solnref{U1}.

\item In solutions of type \solnref{OV1} we have two fixed points of a single $i$-slice.
More specifically, we have 
an $i$-slice $s$ and two points $p = (v_1, v_2, \dots, v_n)$ and $q = (u_1, u_2,
\dots, u_n)$ in the slice $s$ with $p \ne q$ such that 
 $D_j(p) = D_j(q) = \zero$ for all $j \le i$.
Let $v_j$ be the bit-string that uses dimension $i$, and let $L' = (S', P', V')
= \subline(v_{j+1}, v_{j+2}, \dots, v_n, L)$. Since $p$ and $q$ both lie in $s$,
we know that $v_{l} = u_l$ for all $l > j$. Observe that, since $D_l(p) = D_l(q)
= \zero$ for all $l \le j$, we can use the same reasoning as we did in case 1 to
argue that $v_1$ through $v_{j-1}$ encode a vertex that is at the end of the
sub-line embedded into $(\blank, \blank, \dots, v_j, v_{j+1}, \dots, v_n)$, and
$u_1$ through $u_{j-1}$ encode a vertex that is at the end of the sub-line
embedded into $(\blank, \blank, \dots, u_j, v_{j+1}, \dots, v_n)$. There are
multiple cases to consider.

\begin{enumerate}
\item If $v_j = u_j$, then we have that $\subline(v_j, v_{j+1}, \dots, v_n, L) =
\subline(u_j, u_{j+1}, \dots, u_n, L)$.  Since $p \ne q$, there must be some
pair of vertices $v_l$ and $u_l$ such that $v_l \ne u_l$, and note that since
$p$ and $q$ both at the end of their respective sub-lines, we have $V'(v_l) =
V'(u_l)$, which also implies that $V(v_l) = V(u_l)$, which is a solution of type
UV3.

\item If $v_j \ne u_j$, and $D_l(p) = \zero$ for all
$l$ in the range $m(i-1)+1 \le l \le mi$, then 
$\subline(v_j, v_{j+1}, \dots, v_n, L)$ is the second half of $L'$, while
$\subline(u_j, u_{j+1}, \dots, u_n, L)$ is the first half. 
Since $q$ represents a vertex at the end of the corresponding line, we have that
$S(\decode(q))$ is the first vertex on the next half of the $L'$,
meaning that $V'(S(\decode(q))) = 2^{j-1}$. Moreover since $p$ is on the
second-half of the line, we have that $V'(v_j) = 2^{j-1}$, so we have
$V(S(\decode(q))) = V(v_j)$. This is a solution of type \solnref{UV3} so long as
$S(\decode(q)) \ne v_j$.

To prove that this is the case, recall that the direction functions for $q$ in the
dimensions corresponding to $u_j$ always point towards $S(\decode(q))$. Since
$u_j \ne v_j$, and since $u_j$ and $v_j$ lie in the same slice $s$, we know that
they disagree on some dimension $l \le i$. But we also have that $D_l(q) = \zero$
for all $l \le i$, which can only occur if $S(\decode(q))$ disagrees with $v_j$
in some dimension. Hence we have shown that 
$S(\decode(q)) \ne v_j$.

\item If $v_j \ne u_j$, and $D_l(q) = \zero$, for all $l$ in the range $m(i-1)+1
\le l \le mi$, then this is entirely symmetric to the previous case. 

\item In the last case, we have $v_j \ne u_j$, an index $l_1$ in the range 
$m(i-1)+1 \le l_1 \le mi$ with $D_{l_1}(p) \ne \zero$, and an index $l_2$
in the range 
$m(i-1)+1 \le l_2 \le mi$ with $D_{l_2}(p) \ne \zero$. Thus 
$\subline(v_j, v_{j+1}, \dots, v_n, L) = 
\subline(u_j, u_{j+1}, \dots, u_n, L)$, meaning that both points lie at the end
of the first half of $L'$. 
Thus the direction functions at $p$ point towards
$S(\decode(p))$, while the direction functions at $q$ point towards
$S(\decode(q))$. Moreover we have $V'(S(\decode(p))) = V'(S(\decode(q))$, so to
obtain a solution of type \solnref{UV3}, we need to prove that 
$S(\decode(p))) \ne S(\decode(q)$. 

This follows from the fact that $v_j \ne v_u$, and $v_j$ and $v_u$'s membership
in the slice $s$. This means that they disagree on some dimension $l < i$, but
since $D_{a}(p) = D_{a}(q) = \zero$ for all $a < i$, we must have
$S(\decode(p))) \ne S(\decode(q)$. 

\end{enumerate}

\item For solutions of type \solnref{OV2}, we have 
two points $p = (v_1, v_2, \dots, v_n)$ and $q = (u_1, u_2, \dots, u_n)$ that
lie in an 
$i$-slice $s$ with the following properties.
\begin{itemize}
\item $D_j(p) = D_j(q) = \zero$ for all $j < i$, 
\item $p_i = q_i + 1$, and
\item $D_i(p) = \down$ and $D_i(q) = \up$.
\end{itemize}
As with case 2, let $v_j$ be the bit-string that uses dimension $i$, and let $L'
= \subline(v_{j+1}, v_{j+2}, \dots, v_n) = 
\subline(u_{j+1}, u_{j+2}, \dots, u_n)$. 
Since $D_i(p) \ne \zero$ and $D_i(q) \ne \zero$, we must have that $p$ and $q$
are both points on the first half of $L'$. Furthermore, $\decode(p)$ and
$\decode(q)$ must both be at the end of the first half, since $D_l(p) = D_l(q) =
\zero$ for all $l < i$. Hence, $V'(\decode(p)) = V'(\decode(q)) = 2^{i-1} - 1$,
which also implies that $V(\decode(p)) = V(\decode(q))$. So to obtain a solution
of type \solnref{UV3}, we just need to prove that $\decode(p) \ne \decode(q)$.

To prove this, we observe that the direction function in dimension $i$ should be
oriented towards $S(\decode(p))$ for the point $p$, and $S(\decode(q))$  for the
point $q$. However, since $D_i(p)$ and $D_i(q)$ both point away from their
points, and since $p_i \ne q_i$, this must mean that $S(\decode(p)) \ne
S(\decode(q))$, which also implies that $\decode(p) \ne \decode(q)$.

\item We claim that solutions of type \solnref{OV3} are impossible. A solution of type \solnref{OV3}
requires that we have a point $p$ with either $p_i = 0$ and $D_i(p_i) = \down$,
or $p_i = 1$ and $D_i(p_i) = \up$. Our construction never sets 
$D_i(p_i) = \down$ when $p_i = 0$, and it never sets 
$D_i(p_i) = \up$ when $p_i = 1$. So solutions of type \solnref{OV3} cannot occur.
\end{enumerate}

To see that the reduction is promise preserving, it suffices to note that we
only ever map solutions of type \solnref{U1} onto solutions of type \solnref{O1}. Thus, if the
original instance has only solutions of type \solnref{U1}, then the resulting OPDC
instance will have solutions of type \solnref{O1}.

\section{Proofs for Section~\ref{sec:uso}: USO to OPDC}
\label{app:uso}

\subsection{Proof of Lemma~\ref{lem:uso}}

Every solution of the OPDC instance can be mapped back to a solution of the USO
instance. We prove this by enumerating all possible types of solution.

\begin{enumerate}
\item In solutions of type \solnref{O1} we are given a point $p \in P$ such that $D_i(p) =
\zero$ for all $i$. If $\udir(p) \ne \vblank$, then this means that $p$ is a
sink, and so it is solution of type \solnref{US1}. If $\udir(p) = \vblank$, then this
means that $p$ is a solution of type \solnref{USV1}.

\item In solutions of type \solnref{OV1}, we have
an $i$-slice $s$ and two points $p, q \in P_s$ with $p \ne q$ such that 
 $D_j(p) = D_j(q) = \zero$ for all $j \le i$. If $\udir(p) = \vblank$ or
$\udir(q) = \vblank$, then we have a solution of type \solnref{USV1}. Otherwise, this
means that $p$ and $q$ are both sinks of the face corresponding to $s$, and
specifically this means that they have the same out-map on this face. So this
gives us a solution of type \solnref{USV2}.

\item In solutions of type \solnref{OV2} we have
an $i$-slice $s$ and two points $p, q \in P_s$ such that
\begin{itemize}
\item $D_j(p) = D_j(q) = \zero$ for all $j < i$, 
\item $p_i = q_i + 1$, and
\item $D_i(p) = \down$ and $D_i(q) = \up$.
\end{itemize}
Note that in this case we must have $\udir(p) \ne \vblank$ and $\udir(q) \ne
\vblank$. If we restrict the cube to the face defined by $s$, note that
$\udir(p)_j = \udir(q)_j = 0$ for all dimensions $j \ne i$, and $\udir(p)_i =
\udir(q)_i = 1$. Hence $p$ and $q$ have the same out-map on the face defined by
$s$, which gives us a solution of type \solnref{USV2}.

\item Solutions of type \solnref{OV3} are impossible for the instance produced by our
reduction. In these solutions we have a point $p$ with $p_i = 0$ and $D_i(p) =
\down$, or a point $q$ with $q_j = 1$ and $D_j(q) = \up$. Since our
reduction never does this, solutions of type \solnref{OV3} cannot occur.
\end{enumerate}

To see that the reduction is promise-preserving, it suffices to note that
solutions of type \solnref{US1} are only ever mapped onto solutions of type \solnref{O1}. Thus, if
the USO instance has no violations, then the resulting OPDC instance also has
no violations. This completes the proof of Lemma~\ref{lem:uso}.

\section{Prerequisites for Appendix~\ref{app:plcontraction2opdc} and~\ref{app:algorithm_details}}

\subsection{Slice Restrictions of Contraction Maps}
\label{sec:slice}

Our algorithms will make heavy use of the
concept of a \defineterm{slice restriction} of a contraction map, which we describe
here. We define the set of fixed coordinates of a slice
$\ss\in \Slice_d$, $\fixed(\ss) = \Setbar{i\in [d]}{s_i \neq \blank}$ and the set
of free coordinates, $\free(\ss) = [d] \setminus \fixed(s)$. Thus, an $i$-slice is a slice $\ss \in \Slice_d$ for which $\free(\ss) = [i]$. We'll say that a slice is a \defineterm{$k$-dimensional slice} if $\Card{\free(\ss)} = k$.

We can define the \defineterm{slice restriction} of a function $f : [0,1]^d \to [0,1]^d$ with respect to a slice $\ss\in \Slice_d$, denoted $\restr{f}{\ss}$, to be the function obtained by fixing the coordinates $\fixed(\ss)$ according to $\ss$, and keeping the coordinates of $\free(\ss)$ as arguments. To simplify usage of $\restr{f}{\ss}$ we'll formally treat $\restr{f}{\ss}$ as a function with $d$ arguments, where the coordinates in $\fixed(\ss)$ are ignored. Thus, we define $\restr{f}{\ss}:[0,1]^d\to [0,1]^d$ by
\[ \restr{f}{\ss}(x) = f(y) \quad\text{where}\ y_i = \begin{cases} s_i &\ \text{if $i \in \fixed(\ss)$}\\ x_i&\ \text{if $i \in \free(\ss)$.}\end{cases} \]

Let $\free(\ss) = \Set{i_1,\dotsc, i_k}$. We'll also introduce a variant of $\restr{f}{\ss}$ when we want to consider the slice restriction as a lower-dimensional function, $\Restr{f}{\ss} : [0,1]^d \to [0,1]^{\Card{\free(\ss)}}$ defined by
\[ \Restr{f}{\ss}(x) = \Paren{\Paren{\restr{f}{\ss}(x)}_{i_1}, \dotsc, \Paren{\restr{f}{\ss}(x)}_{i_k}}\text{.} \]

We can also define slice restrictions for vectors in the natural way:
\[ \Paren{\restr{x}{\ss}}_i = \begin{cases} s_i&\ \text{if $s_i \neq \blank$}\\ x_i&\ \text{otherwise.} \end{cases}\]
Finally, we'll use $\Restr{x}{\ss}$ to denote projection of $x$ onto the coordinates in $\free(\ss)$:
\[ \Restr{x}{\ss} = \Paren{x_{i_1},\dotsc, x_{i_k}}\text{.} \]

We now extend the definition of a contraction map to a slice restriction of a function in the obvious way. We say that $\restr{\tf}{\ss}$ is a contraction map with respect to a norm $\Norm{\cdot}$ with Lipschitz constant $c$ if for any $x,y\in [0,1]^d$ we have
\[ \Norm{\Restr{f}{\ss}(x) - \Restr{f}{\ss}(y)} \leq c \Norm{\Restr{x}{\ss} - \Restr{y}{\ss}}\text{.} \]

We'll also introduce some notation to remove clutter. We'll define
\[\Delta_i(p) \triangleq f(p)_i - p_i\, \forall i \in [d]\] and use this notation when the function $f$ is clear from the context. Note that we'll only use $\Delta_i(p)$ when considering the free coordinates of a slice restriction, so that $\Delta_i(p)$ doesn't depend on whether the most recent context has us considering $\restr{f}{\ss}$ or $f$.

Slice restrictions will prove immensely useful through the following observations:

\begin{lemma}\label{lem:cm1}
Let $f:[0,1]^d\to [0,1]^d$ be a contraction map with respect to $\Norm{\cdot}_p$ with Lipschitz constant $c \in (0,1)$. Then for any slice $\ss \in \Slice_d$, $\restr{f}{\ss}$ is also a contraction map with respect to $\Norm{\cdot}_p$ with Lipschitz constant $c$.
\end{lemma}
\begin{proof}
For any two vectors $x,y \in [0,1]^{d}$ we have
\begin{align*}
  \Norm{\restr{\tf}{\ss}(x) - \restr{\tf}{\ss}(y)}_p &\leq \Norm{f(\restr{x}{\ss}) - f(\restr{y}{\ss})}_p\\
                                       &\leq c \Norm{\restr{x}{\ss} - \restr{y}{\ss}}_p\\
                                       &= c \Norm{\Restr{x}{\ss} - \Restr{y}{\ss}}_p
\end{align*}
\end{proof}

Since slice restrictions of contraction maps are themselves contraction maps in the sense defined above, they have unique fixpoints, up to the coordinates of the argument which are fixed by the slice and thus ignored. We'll nevertheless refer to the \defineterm{unique fixpoint of a slice restriction of a contraction map}, which is the unique point $x\in [0,1]^d$ such that
\[ \Restr{f}{\ss}(x) = \Restr{x}{\ss}\quad\text{and}\quad x = \restr{x}{\ss}\text{.} \]

\begin{lemma}\label{lem:cm2}
Let $f : [0,1]^d \to [0,1]^d$ be a contraction map with respect to $\Norm{\cdot}_p$ with Lipschitz constant $c\in (0,1)$. Let $\ss, \ss' \in \Slice_d$ be such that $\fixed(\ss') = \fixed(\ss) \cup \Set{i}$ and $s_j = s'_j$ for all $j \in \fixed(\ss)$. Let $x, y \in [0,1]^d$ be the unique fixpoints of $\Restr{f}{\ss}$ and $\Restr{f}{\ss'}$, respectively. Then $(x_i - y_i)\Delta_i(y) > 0$. 
\end{lemma} 
\begin{proof}
We'll prove this by contradiction. Without loss of generality, assume towards a contradiction that $x_i \leq y_i$ and that $f(y)_i > y_i$. Then we have
\begin{align*}
  \Norm{f(y) - f(x)}_p^p
  &= \Norm{\Paren{f(y)_1,\dotsc, f(y)_d} - \Paren{f(x)_1, \dotsc, f(x)_d}}_p^p \\
  &= \sum_{i\in \fixed(\ss)} \Abs{s_i - s_i}^p + \sum_{j\in \free(\ss')} \Abs{y_j - x_j}^p + \Abs{f(y)_i - x_i}^p\\
  &> \sum_{i\in \fixed(\ss)} \Abs{s_i - s_i}^p + \sum_{j\in \free(\ss')} \Abs{y_j - x_j}^p + \Abs{y_i - x_i}^p\\
  &= \Norm{y - x}_p^p
\end{align*} which contradicts the fact that $f$ is a contraction map. The lemma follows.
\end{proof}

\subsection{Approximate Fixpoint Lemmas}
\label{sec:approx_lemmas}
In this section we state and prove the key lemmas that will be used both in our reduction of \CM to \EOPL and our algorithms for finding approximate fixpoints. The lemmas reflect the intuition that sufficiently good approximate fixpoints on $(k-1)$-dimensional slices can be used to obtain slightly worse approximate fixpoints for $k$-dimensional slices for a notion of approximation to be formalized shortly. By choosing our approximation guarantees at each level appropriately we can ensure that we end up with a approximate fixpoint for the original function.

To formalize this intuition we define an approximate fixpoint of a given dimension with respect to some $\ell_p$-norm and a slice. For any fixed $\ell_p$-norm, $i$-slice $\ss$, and dimension parameter $k \leq i$ we'll say that a point $x\in \ss$ is a \defineterm{$(\ss, \ell_p, k)$-approximate fixpoint} if
\[ \Abs{\Delta_j(x)} \leq \eps_j(p,d),\quad \forall j \in [k]  \] where $\Paren{\eps_j}_{j=1}^k$ is a constant depending only on the $\ell_p$ norm and the dimension $d$. To make the dependence on $p$ and $d$ explicit, we'll write $\eps_j(p,d)$ instead of $\eps_j$.

We'll define an \defineterm{$\eps$-approximate fixpoint of a contraction map $f$ w.r.t. an $\ell_p$-norm} to be a point $x\in [0,1]^d$ such that $\Norm{f(x) - x}_p \leq \eps$.

For each different $\ell_p$ norm and dimension $d$ of our contraction map, we will choose a different sequence $\Paren{\eps_i(p,d)}_{i=1}^d$ such that if a point $x$ on a $(k-1)$-slice satisfies $\Abs{\Delta_j(x)} \leq \eps_j(p,d)j$ for all $j < k$, either $\Abs{\Delta_k(x)} \leq \eps_k(p,d)$ or the sign of $\Delta_k(x)$ indicates the direction of the unique fixpoint of that slice. 
There are two distinct cases to consider for $\Paren{\eps_i(p,d)}_{i=1}^d$:
\begin{itemize}
\item For $p=1$, we choose $\eps_i(1,d)^{(1)}_i \triangleq \eps/2^{2(d+1-i)}$.
\item For $2 \leq p < \infty$, we choose $\eps_i(p,d) \triangleq \eps^{p^{(d+1-i)}} \Paren{dp}^{-2\sum_{j=0}^{d+1-i} p^j}$.
\end{itemize}

The key property satisfied by the choices of $\eps$ is captured in the following:
\begin{lemma} For any $p \in \N$, and any $d\in \N, k\in \N$ with $k \leq d$,
  \[ \sum_{i=1}^{k-1}p\eps_i(p,d) \leq {\eps_k(p,d)}^p\text{.} \] \label{lem:approx_cm_progress_helper}
\end{lemma}
\begin{proof}
  There are two cases to consider. For $p=1$,
  \[ \sum_{i=1}^{k-1}\eps_i(1,d) = \sum_{i=1}^{k-1} \eps 2^{-2(d+1-i)} < \eps 2^{-2(d+1-k)} = \eps_k(1,d)\text{.} \] The interesting case is where $p \geq 2$. Here we observe that
\begin{align*}
  \sum_{i=1}^{k-1}p\eps_i(p,d)
  &= \sum_{i=1}^{k-1} p\frac{\eps^{p^{d+1-i}}}{\Paren{dp}^{2p^0 + 2p^1 + \dotsb + 2p^{d+1-i}}}\\
  &< \sum_{i=1}^{k-1} p\frac{\eps^{p^{d+1-(k-1)}}}{\Paren{dp}^{2p^0+2p^1+\dotsb + 2p^{d+1-(k-1)}}}\\
  &< dp \frac{\eps^{p^{d+1-(k-1)}}}{\Paren{dp}^{2p^0 + 2p^1 + \dotsb + 2p^{d+1-(k-1)}}}\\
  &= \frac{\eps^{p^{d+1-(k-1)}}}{\Paren{dp}^{2p^1 + 2p^2 + \dotsb + 2p^{d+1-(k-1)} + 1}}\\
  &= \Paren{\frac{\eps^{p^{d+1-k}}}{\Paren{dp}^{2p^0 + 2p^1 + \dotsb + 2p^{d+1-k}}}}^p/\Paren{dp}\\
  &= \Paren{\eps_k(p,d)}/\Paren{dp}\\
  &< \eps_k(p,d)\text{.}
\end{align*}
\end{proof}

We now prove lemmas showing that these choices of $\Paren{\eps_i(p,d)}_{i=1}^d$ ensure that $k-1$-dimensional fixpoints can be used to find $k$-dimensional fixpoints. 

\begin{lemma}[Approximate Fixpoints Give Directions]\label{lem:approx_cm_progress}
  Let $\ss \in \Slice_d$ be a $k$-slice, for some $k \leq d$. Let $f:[0,1]^d\to[0,1]^d$ be a contraction map with respect to $\Norm{\cdot}_p$. Let $x^*$ be the unique fixpoint of $\Restr{f}{\ss}$. Given any $(\ss, \ell_p, k-1)$-approximate fixpoint $v\in [0,1]^d$, either $\Abs{\Delta_k(v)} \leq \eps_k(p,d)$ or $\Delta_k(v) (x^*_k - v_k) > 0$.
\end{lemma}

\begin{proof}
  By definition, $v$ satisfies $\Abs{\Delta_i(v)} \leq \eps_i$ for all $i \leq k-1$. Assume that $\Abs{\Delta_k(v)} > \eps_k$. Without loss of generality, let $\Delta_k(v) > \eps_k$. Assume towards a contradiction that $x^*_k \leq v_k$. Then
  \begin{align}
    \Norm{\Restr{f}{\ss}(v) - \Restr{x^*}{\ss}}_p^p
    &= \sum_{j=1}^{k} \Abs{f(v)_j - x^*_j}^p\\
    &> \eps_k^p + \Abs{v_k - x^*_k}^p + \sum_{j=1}^{k-1} \Abs{f(v)_j - x^*_j}^p\\
    &\geq \eps_k^p + \Abs{v_k - x^*_k}^p + \sum_{j=1}^{k-1} \Paren{{\Abs{v_j - x^*_j} - \eps_j}}^p\\
    &= \eps_k^p + \sum_{j=1}^{k-1} \Abs{v_k - x^*_k}^p - \sum_{j=1}^{k-1} \Brack{\Abs{v_j - x^*_j}^p - \Paren{{\Abs{v_j - x^*_j} - \eps_j}}^p}\\
    &= \Norm{\Restr{v}{\ss} - \Restr{x^*}{\ss}}_p^p + \eps_k^p - \sum_{j=1}^{k-1} \Brack{\Abs{v_j - x^*_j}^p - \Paren{{\Abs{v_j - x^*_j} - \eps_j}}^p} \\
    &\geq \Norm{\Restr{v}{\ss} - \Restr{x^*}{\ss}}_p^p + \eps_k^p - \sum_{j=1}^{k-1} p\eps_j \label{approx_ub_step}\\
    &> \Norm{\Restr{v}{\ss} - \Restr{x^*}{\ss}}_p^p \text{.} \label{approx_eps_step}\\
  \end{align}
  Line~\ref{approx_ub_step} uses the fact that $a^p - (a - b)^p \leq 1 - (1 - b)^p \leq pb$ for $p > 1$, $a,b\in (0,1)$. Line~\ref{approx_eps_step} follows directly from Lemma~\ref{lem:approx_cm_progress_helper}. Again we have a contradiction since $\Norm{\Restr{f}{\ss}(v) - \Restr{x^*}{\ss}}_1 \geq \Norm{\Restr{v}{\ss} - \Restr{x^*}{\ss}}_1$ and $\Restr{f}{\ss}$ is a contraction map. The lemma follows.
\end{proof}

In order for our algorithms and reductions to produce approximate fixpoints we'll need to ensure that we can find an approximate fixpoint somewhere in our discretization of the unit hypercube. We'll first establish that if we don't have an approximate fixpoint, we have witnesses to our function not being a contraction map.

\begin{lemma}[Two Nearby Opposing Pivots Give Violations]  \label{lem:approx_cm-lp_final_violations}
  Let $\ss \in \Slice_d$ be a $k$-slice, for some $k \leq d$. Let $v^{(\ell)}, v^{(h)} \in [0,1]^d$ be $(\ss, \ell_p, k-1)$-approximate fixpoints of $f$ such that $v^{(h)}_k - v^{(\ell)}_k < \eps_k(p,d)/2$ and both $\Delta_k(v^{(h)}) \geq \eps_k(p,d)$ and $\Delta_k(v^{(\ell)}) \leq -\eps_k(p,d)$. Then $\Norm{f(v^{(h)}) - f(v^{(\ell)}}_p \geq \Norm{v^{(h)} - v^{(\ell)}}_p$, and $v^{(h)},v^{(\ell)}$ witness that $f$ is not a contraction map. 
\end{lemma}

\begin{proof}
Under the assumptions of the lemma, we have
  \begin{align}
    \Norm{\Restr{f}{\ss}\Paren{v^{(h)}} - \Restr{f}{\ss}\Paren{v^{(\ell)}}}_p^p - \Norm{\Restr{v^{(h)}}{\ss} - \Restr{v^{(\ell)}}{\ss}}_p^p
    &= \sum_{j=1}^{k-1} \Brack{\Abs{f\Paren{v^{(h)}}_j - f\Paren{v^{(\ell)}}_j}^p - \Abs{v^{(h)}_j - v^{(\ell)}_j}^p}\\
    &= \Brack{\Abs{f\Paren{v^{(h)}}_k - f\Paren{v^{(\ell)}}_k}^p - \Abs{v^{(h)}_k - v^{(\ell)}_k}^p} \\
    &\quad\quad+ \sum_{j=1}^{k-1} \Brack{\Abs{f\Paren{v^{(h)}}_j - f\Paren{v^{(\ell)}}_j}^p - \Abs{v^{(h)}_j - v^{(\ell)}_j}^p} \nonumber\\
    &\geq \Brack{\Abs{f\Paren{v^{(h)}}_k - f\Paren{v^{(\ell)}}_k} - \Abs{v^{(h)}_k - v^{(\ell)}_k}}^p\\
    &\quad\quad+ \sum_{j=k+1}^d \Brack{\Paren{\Abs{v^{(h)}_j - v^{(\ell)}_j} - 2\eps_j(p,d)}^p - \Abs{v^{(h)}_j - v^{(\ell)}_j}^p} \nonumber\\
    &\geq \Paren{3\eps_k(p,d)/2}^p - \sum_{j=k+1}^d 2p\eps_j(p,d) \label{approx_ub_step2}\\
    &> 0\text{.} \label{approx_eps_step2}
  \end{align} Here, line~\ref{approx_ub_step2} mirrors line~\ref{approx_ub_step} from Lemma~\ref{lem:approx_cm_progress} (which holds for $\eps_i(p,d) < 1/2$, which is the case here). Line~\ref{approx_eps_step2} follows from Lemma~\ref{lem:approx_cm_progress_helper} when $p \geq 2$ (since $(3/2)^p > 2$ in that case) and from the observation that $\frac{3}{2}\eps_k(1,d) \geq 2\sum_{j=1}^{k-1}\eps_j(1,d)$ for $p=1$.

  Observing that \[\Norm{f\Paren{v^{(h)}} - f\Paren{v^{(\ell)}}}_p \geq \Norm{\Restr{f}{\ss}\Paren{v^{(h)}} - \Restr{f}{\ss}\Paren{v^{(\ell)}}}_p\] completes the proof. 
\end{proof}

The following lemma establishes that when we have two $(\ss, \ell_p, k-1)$-approximate fixpoints on the same $k$-slice sandwiching the true $k$-dimensional fixpoint in a sufficiently small interval, one of the two $(\ss,\ell_p, k-1)$-approximate fixpoints must be an $(\ss, \ell_p, k)$-approximate fixpoint.
\begin{lemma}[One of Two Adjacent Pivots is a Fixpoint]  \label{lem:approx_cm-lp_final}
  Let $\ss \in \Slice_d$ be a $k$-slice, for some $k \leq d$. Let $f$ be a contraction map with respect to $\Norm{\cdot}_p$. Let $x^*$ be the unique fixpoint of $\Restr{f}{\ss}$ and let $v^{(\ell)}, v^{(h)} \in [0,1]^d$ be $(\ss, \ell_p, k-1)$-approximate fixpoints such that $v^{(h)}_k - v^{(\ell)}_k < \eps_k(p,d)/2$ and the unique fixpoint $x^*$ of $\Restr{f}{\ss}$ satisfies
  \[ v^{(\ell)}_k < x_k^* < v^{(h)}_k\text{.} \] Then either $v^{(h)}$ or $v^{(\ell)}$ is an $(\ss, \ell_p, k)$-approximate fixpoint of $f$. 
\end{lemma}

\begin{proof}
  By the contrapositive of Lemma~\ref{lem:approx_cm_progress}, $\Delta_k\Paren{v^{(h)}} < \eps_k(p,d)$ and $\Delta_k\Paren{v^{(\ell)}} > -\eps_k(p,d)$. If either $\Delta_k\Paren{v^{(h)}} > -\eps_k(p,d)$ or $\Delta_k\Paren{v^{(\ell)}} < \eps_k(p,d)$ we're done since $v^{(h)}$ or $v^{(\ell)}$ is an $(\ss, \ell_p, k)$-approximate fixpoint, respectively. The only way we can fail to have a fixpoint is if both $\Delta_k\Paren{v^{(h)}} \leq -\eps_k(p,d)$ and $\Delta_k\Paren{v^{(\ell)}} \geq \eps_k(p,d)$. By Lemma~\ref{lem:approx_cm-lp_final_violations}, this cannot happen when $f$ is a contraction map, and so one of the two points must be an $(\ss, \ell_p, k)$-approximate fixpoint.
\end{proof}

Finally, by choosing $\eps_k(p,d)$ as above, an $(\ss, \ell_p, d)$-approximate fixpoint is an $\eps$-approximate fixpoint:
\begin{lemma}
  Any $(\ss, \ell_p, d)$-approximate fixpoint $x$ is an $\eps$-approximate fixpoint. \label{lem:slice_approx_to_approx}
\end{lemma}
\begin{proof}
  When $p = 1$
  \begin{align*}
    \Norm{f(v) - v}
    &= \sum_{i=1}^d \Abs{\Delta_i(v)} \\
    &\leq \sum_{i=1}^d \eps_i(p,d) \\
    &= \sum_{i=1}^d \eps/2^{2(d+1-i)} \\
    &= \sum_{i=1}^d \eps/2^{2i} \\
    &< \eps\text{.}
  \end{align*}
  When $p \geq 2$
  \begin{align*}
    \Norm{f(v) - v}_p^p
    &= \sum_{i=1}^d \Abs{\Delta_i(v)}^p \\
    &\leq \sum_{i=1}^d \eps^p_i(p,d) \\
    &= \sum_{i=1}^d \eps^{p^{(d+1-i)}} \Paren{dp}^{-2\sum_{j=0}^{d+1-i} p^j}\\
    &= \sum_{i=1}^d \eps^{p^i} \Paren{dp}^{-2\sum_{j=0}^i p^j}\\
    &< \eps^p\text{.}
  \end{align*}
  \end{proof}

\section{Proofs for Section~\ref{sec:lcm2eopl}: \LCM to \OPDC}
\label{app:plcontraction2opdc}

\subsection{Proof of Lemma~\ref{lem:points2}}
\label{app:points2}

We'll define the bit-length of an integer $n \in \Z$, denoted $b(n)$, by $b(n) \triangleq \Ceil{\log_2 n}$. We'll extend the definition to the rationals by defining the bit-length for $x\in \Q$ as the minimum number of bits needed to represent the numberator and denominator of some representation of $x$ as a ratio of integers:
\[ b(x) \triangleq \min_{\substack{p,q \in \Z\\ x = p/q}}\Paren{b(p) + b(q)}\text{.} \] We'll extend this notion of bit-length to matrices by defining the bit-length $b(M)$ of a matrix $M \in \Q^{m\times n}$ by $b(M) \triangleq \max_{i,j}b(M_{ij})$ and to vectors by defining $b(v) \triangleq \max_{i} b(v_i)$ for $v\in \Q^n$.

\paragraph{\LinearFIXP circuits and LCPs.}
The goal of this proof is to find a tuple $(k_1, k_2, \dots, k_d)$ such that the
lemma holds. 
We utilize a lemma from ~\cite{mehta} which asserts that
every \LinearFIXP circuit can be transformed in polynomial time into an LCP with bounded bit-length,
such that solutions to the LCP capture exactly the fixpoints of the circuit.

For any \LinearFIXP circuit $C:[0,1]^d\to [0,1]^d$ with gates $g_1,\dotsc,g_m$ and constants $\zeta_1,\dotsc,\zeta_q$, we'll define the size of $C$ by
\[ \size(C) \triangleq d + m + \sum_{i=1}^q b(\zeta_i)\text{.} \] 

\begin{lemma}[\cite{mehta}]
\label{lem:ruta}
Let $C:[0,1]^d \to [0,1]^d$ be a \LinearFIXP circuit. We can produce in polynomial time an LCP
defined by an $n \times n$ matrix $M_C$ and $n$-dimensional vector $\qq_C$ for some $n$ with $d \leq n \leq \size(C)$ such that there is a bijection between solutions of the LCP $(M_C, \qq_C)$ and fixpoints of $C$. Moreover, to obtain a fixpoint $x$ of $C$ from a solution $\yy$ to the LCP $(M_C, \qq_C)$, we can just set $x = (\yy_1,\dotsc, \yy_d)$. Furthermore, $b(M_C)$ and $b(\qq_C)$ are both at most $O(n\times\size(C))$.
\end{lemma}

Crucially the construction interacts nicely with fixing inputs; if $C'$ denotes a circuit where one of the inputs of $C$ is fixed to
be some number $x$, we can bound the bit-length of $M_{C'}$ and $\qq_{C'}$ in terms of the bit-lengths of $M_C$, $\qq_C$, and $x$.
\begin{observation} \label{obs:m}
$b(M_{C'}) \leq b(M_{C})$.
\end{observation}
\begin{observation} \label{obs:q}
  $b(\qq_{C'}) \leq \max\Set{b(\qq_C), b(M_C) + b(x)} + 1$.
\end{observation}

In other words, the bit-length of $M_{C'}$ does not
depend on $x$, and is in fact at most the bit-length of $M_{C}$, and the bit-length of $\qq_{C'}$ is bounded by the worse of the bit-lengths of $\qq_C$ or the sum of the bit-lengths of $M_C$ and $x$ plus an additional bit.

\paragraph{\bf Bounding the bit-length of a solution of an LCP.}
We now prove two technical lemmas about the bit-length of any solution to an
LCP. We begin with the following lemma regarding the bit-length of a matrix
inverse.

\begin{lemma}
\label{lem:Ainverse}
Let $A \in \Z^{n\times n}$ be a square matrix of full rank, and let $B = \max_{ij} \Abs{A_{ij}}$ be the largest absolute value of an entry of $A$. Then $b(A^{-1}) \leq 2n \log B + n\log n$. 
\end{lemma}
\begin{proof}
We have $A^{-1} = \frac{1}{\det(A)}\adj(A)$. Each entry of $\adj(A)$ is a cofactor of $A$, each of which is a (possibly negated) determinant of an $(n-1)\times (n-1)$ submatrix of $A$. It is a well-known corollary of Hadamard's inequality that $\Abs{\det(A)} \leq B^n n^{n/2}$. The same bound obtains for each submatrix. 
We can write entry $A^{-1}_{ij}$ as $p/q$ where $p = \adj(A)_{ij}$ and $q = \det(A)$ are both integers. We have $b(p),b(q) \leq n\log B + (n/2) \log n$. The lemma follows immediately.
\end{proof}

We now use this to prove the following bound on the bit-length of a solution
of an LCP.

\begin{lemma}
\label{lem:lcpsize}
Let $M \in \Q^{n\times n}$ and $\qq \in \Q^{n}$. Let $\yy$ be a solution to the LCP $(M, \qq)$. Then \[
b(\yy) \leq (5n+2)\log n + (4n+1)b(M) + n + b(\qq)\text{.} \]
\end{lemma}
\begin{proof}



We first note that if an LCP has a solution, then it has a vertex solution.
Let $(\yy,\ww)$ be such a vertex solution of the LCP and, as in Section~\ref{sec:PLCPtoEOPL}, let
$\alpha = \{i \ |\ \yy_i > 0 \}$, and let $A = A_\alpha$ be defined according 
to \eqref{eq:Aalpha}.
We have that $A$ is guaranteed to be invertible, and we have
that $A\xx = \qq$,
with $\yy_i = \xx_i$ for $i \in \alpha$ and $\yy_i = 0$ for $i \notin \alpha$,
so we have $b(\yy) \leq b(\xx)$.
Also note that we have~$b(A) \leq b(M)$, since the entries in columns that take the
value of $e_i$ have constant bit-length.

We must transform $A$ into an integer matrix in order to apply
Lemma~\ref{lem:Ainverse}. Let $\ell$ denote the least common multiple of the
denominators of the entries in $A$.
Note that $\ell \leq n^2 2^{b(A)}$ and hence $b(\ell) \leq b(A)+ 2\log(n)$. Our matrix equation above can be rewritten as $\ell A \xx =
\ell \qq$, where $(\ell A)$ is an integer matrix. Hence we have $\xx
= (\ell A)^{-1} (\ell\qq)$.

If $B$ denotes the maximum entry of $\ell A$, then $B \leq \ell 2^{b(A)}$.
So by Lemma~\ref{lem:Ainverse} every entry of $(\ell A)^{-1}$ can be represented
with integer numerators and denominators bounded by 
\begin{equation*}
B^nn^{n/2} \leq (\ell 2^{b(A)})^n  n^{n/2} \leq
(n^2 2^{2 b(A)})^n n^{n/2}\text{.}
\end{equation*} From this bound we obtain
\begin{align*}
  b(\Paren{\ell A}^{-1}) &\leq 2 \log\Paren{(n^2 2^{2b(A)})^n n^{n/2}}\\
                         &= 2 \Paren{n\Paren{2 \log n + 2b(A)} + (n/2)\log n }\\
                         &\leq 5n\log n + 4nb(A)\text{.}
\end{align*}

Each entry of $\yy$ consists of the sum of $n$ entries of 
$(\ell A)^{-1}$ each of which is multiplied by an entry of
$\qq$, followed at the end with a multiplication by $\ell$. We get the following bound
on the bit-length of~$\yy$.
\begin{align*}
  b(\yy) &\leq b(\Paren{\ell A}^{-1}) + b(\qq) + n + b(\ell)\\
         &\leq 5n\log n + 4nb(A) + n + b(A) + 2\log n\\
         &\leq (5n+2)\log n + (4n+1)b(A) + n + b(\qq)
\end{align*}
\end{proof}

\paragraph{\bf Fixing the grid size.} 
We shall fix the grid size iteratively, starting with $k_d$, and working
downwards. At each step, we will have a space that is partially continuous
and partially discrete. Specifically, 
after the iteration in which we fix $k_i$, the dimensions $j < i$ will allow any
number in $[0, 1]$, while the dimensions $j \ge i$ will only allow the points
with denominator $k_j$. If $I_{k}$ denotes the set of all rationals with
denominator at most $k$, then we will denote this as 
\begin{equation*}
P(k_i, k_{i+1}, \dots, k_d) = [0,1]^{i-1} \times I_{k_i} \times I_{k_{i+1}} \times
\dots \times I_{k_{d}}.
\end{equation*}
Moreover, after fixing $k_i$, we will have that the property required by
Lemma~\ref{lem:points2} holds for all $j$-slices $s$ with $j \ge i$. Specifically,
that if $x$ is a fixpoint of $s$ according to $f$, then there exists a $p \in
P(k_i, k_{i+1}, \dots, k_{d})$ that is a fixpoint of $s$ according to
$\mathcal{D}$.

We'll start by bounding the bit-length of a solution to the \LCM problem computed one coordinate at a time. Given a PL-circuit $C$, we use Lemma~\ref{lem:ruta} to produce an LCP defined by $M \in \Q^{n\times n}$ and $\qq \in \Q^{n}$. Now let $x_1,\dotsc, x_n$ be formal variables representing the inputs to our circuit $C$. We want to determine parameters $\kappa_1,\dotsc, \kappa_d$ such that such that if we fix variables $x_{i+1},..., x_{d}$ to values with bit-lengths $b(x_{i+1}),\dotsc, b(x_d)$ where $b(x_{j}) \leq \kappa_j$ for each $j \in \Set{i+1,\dotsc, d}$, then any fixpoints of $C$ with respect to the free variables $x_1,\dotsc,x_i$ will have bit-lengths $b(x_1),\dotsc,b(x_i)$ with $b(x_j) \leq \kappa_j$ for $j\in [i]$.

We'll set \[\kappa_i \triangleq (d-i+1) ((5n+2)\log n + n + (4n+2) b(M) + 1) + b(\qq)\text{.}\] Note that $\kappa_1 \geq \kappa_2 \geq \dotsb \geq \kappa_d$.

To prove that these bounds suffice, we'll use induction on $i$, starting from $i=d$. First, we observe that by Lemma~\ref{lem:lcpsize}, we have $b(\yy) \leq (5n+2)\log n + (4n+1)b(M) + b(\qq) = \kappa_d - 1 \leq \kappa_d$ for any solution $\yy$ to the LCP $(M, \qq)$. Moreover each fixpoint $x$ of $C$ corresponds to a solution $\yy$ to the LCP by Lemma~\ref{lem:ruta}, so we have $b(x_1),\dotsc,b(x_d) \leq \kappa_d$ which implies the weaker claim that all fixpoints can be found by choosing $b(x_i) \leq \kappa_i$ for all $i\in [d]$, since $\kappa_i \geq \kappa_d$ for all $i \in [d]$.

Now we'll handle the inductive case. For $i=1,\dotsc, d$, let $M^{(i)}, \qq^{(i)}$ be the pair defining the LCP after $x_{i+1}$ through $x_d$ are fixed to values with bit-lengths bounded by $\kappa_{i+1},\dotsc, \kappa_d$, respectively. This pair will of course depend on the values $x_{i+1},\dotsc, x_d$, but since the bit-length of $M^{(i)}$ and $\qq^{(i)}$ depend only on the bit-lengths of the fixed values, we can ignore the values of $x_{i+1},\dotsc,x_d$ as long as we have bounds on their bit-lengths and as long as we restrict our attention to the bit-lengths of $M^{(i)}$ and $\qq^{(i)}$ and not the values themselves.

Using Lemma~\ref{lem:lcpsize} we know that the solution to the LCP has $b(x_i) \leq (5n+2) \log n + (4n+1)b(M^{(i)}) + b(\qq^{(i)})$. Moreover, we obtained $M^{(i)}$ by repeatedly fixing inputs to $C$, so the repeated application of Observation~\ref{obs:m} implies that $b(M^{(i)}) \leq b(M^{(i+1)}) \leq b(M)$. Moreover, Observation~\ref{obs:q} gives us
\[ b(\qq^{(i)}) \leq \max\Set{b(\qq^{(i+1)}), b(M^{(i+1)}) + b(x_{i+1})} + 1\text{.} \] By a simple argument, we can also show that $\kappa_{i+1} \geq b(\qq^{(i+1)})$ so that the above bound can be simplified to 
\[ b(\qq^{(i)}) \leq b(M^{(i)}) + b(x_{i+1}) + 1 \] at which point we can use our observation and the bound $b(x_{i+1}) \leq \kappa_{i+1}$ to obtain
\[ b(\qq^{(i)}) \leq b(M) + \kappa_{i+1} + 1\text{.} \]

Now we use the bound from Lemma~\ref{lem:lcpsize} again to obtain
\begin{align*}
  b(x_i) &\leq (5n+2)\log n + (4n+1)b(M^{(i)}) + b(\qq^{(i)})\\
         &\leq (5n+2)\log n + (4n+1)b(M) + b(M) + \kappa_{i+1} + 1
\end{align*} and we now use the definition of $\kappa_{i+1}$ to conclude
\begin{align*}
  b(x_i) &\leq (5n+2)\log n + (4n+1)b(M) + b(M)\\
         &\quad+ (d-(i+1)+1) ((5n+2)\log n + n + (4n+2) b(M) + 1) + b(\qq) + 1\\
         &= (d-i+1)((5n+2)\log n + (4n+2)b(M) + 1) + b(\qq)\\
         &= \kappa_i\text{.}
\end{align*}

We conclude that all solutions to the LCP $(M^{(i)}, \qq^{(i)})$ have free coordinates with bit-length at most $\kappa_i$, and similarly for the fixpoints of $C$. Thus, every fixpoint of $C$ with respect to the free variables can be found by choosing $x_1,\dotsc, x_{i}$ to have bit-lengths at most $\kappa_1,\dotsc,\kappa_i$, respectively, when the bit-lengths of $x_{i+1},\dotsc,x_{d}$ are chosen to have bit-lengths at most $\kappa_{i+1},\dotsc,\kappa_d$. 

To conclude the proof of Lemma~\ref{lem:points2}, we need to choose $k_1,\dotsc, k_d$ so that every $i$-slice with fixed coordinates in $P(k_1,\dotsc,k_d)$ has a fixpoint also on $P(k_1,\dotsc,k_d)$ when we map points $x\in P(k_1,\dotsc,k_d)$ to $[0,1]^d$ by
\[ x \mapsto (x_1/k_1,\dotsc,x_d/k_d)\text{.} \]

We set $k_i \triangleq 2^{\kappa_i}$. Now a point $x \in P(k_1,\dotsc,k_d)$ corresponds to a point $y\in [0,1]^d$ with $b(y_i) \leq 2\kappa_i$ for all $i\in [d]$ and any $i$-slice where the fixed coordinates satisfy the bit-length bounds will have fixpoints for the remaining variables with all coordinates satisfying the bit-length bounds.

Finally, we observe that for each $i\in [d]$, $b(k_i) = \kappa_i \leq \kappa_1 = O(\poly(\size(C)))$, so each of the $k_i$ can be represented using polynomially many bits in the size of the circuit $C$.

\subsection{Proof of Lemma~\ref{lem:ov2violation}}
\label{app:ov2violation}

The proof of this lemma will make use of Lemmas~\ref{lem:cm1} and~\ref{lem:cm2},
which are proved in Appendix~\ref{sec:slice}.

The statement of the proof says that we can assume that $f$ is contracting with
respect to some $\ell_p$ norm, and that we have an $i$-slice $s$ and
two points $p, q$ in $s$ satisfying the following conditions.
\begin{itemize}
\item $D_j(p) = D_j(q) = \zero$ for all $j < i$, 
\item $p_i = q_i + 1$, and
\item $D_i(p) = \down$ and $D_i(q) = \up$.
\end{itemize}
To translate $p$ and $q$ from the grid to the $[0, 1]^n$ space, we must divide
each component by the grid length in that dimension. Specifically, we define the
point $a \in [0, 1]^n$ so that $a_i = p_i/k_i$ for all $i$, and the point $b \in
[0, 1]^n$ such that $b_i = q_i/k_i$ for all $i$.

Lemma~\ref{lem:cm1} states that if $f$ is contracting with respect to an
$\ell_p$ norm, then the restriction of $f$ to the slice $s$ is also contracting
in that $\ell_p$ norm. Hence, $f$ must have a fixpoint in the slice $s$. Let $x
\in [0, 1]^n$ denote this fixpoint.

By definition we have that $(f(x) - x)_j = 0$ for all $j \le i$, and we also
have $(f(a) - a)_j = 0$ for all $j < i$ from the fact that $p$ is a fixpoint of
its $i$-slice. So we can apply Lemma~\ref{lem:cm2} to $x$
and $a$, and from this we get that $(x_i - a_i) \cdot (f(a)_i - a_i) > 0$. Since $D_i(p)
= \down$, we have that $f(a)_i < a_i$, and hence 
$(f(a)_i - a_i) < 0$. Therefore, we can conclude that 
$(x_i - a_i) < 0$, which implies that $x_i < a_i$.

By the same reasoning we can apply Lemma~\ref{lem:cm2} to $x$ and $b$, and this gives us that 
$(x_i - b_i) \cdot (f(b)_i - b_i) > 0$. This time we have $D_i(q) = \up$, which
implies that $f(b)_i > b_i$, and hence $(f(b)_i - b_i) > 0$. So we can conclude
that $(x_i - b_i) > 0$, meaning that $x_i > b_i$.

Hence we have shown that $b_i < x_i < a_i$, and so the point $x$ satisfies the
conditions of the lemma. This completes the proof of
Lemma~\ref{lem:ov2violation}.



\subsection{Proof of Lemma~\ref{lem:lcm2opdc}}
\label{app:lcm2opdc}

We must map every possible solution of the OPDC instance back to a solution of
the \LCM instance. We will enumerate all solution types for OPDC.

\begin{itemize}
\item In solutions of type \solnref{O1}, we have a point $p \in P$ such that $D_i(p) =
\zero$ for all $i$. This means that the point $x$, where $x_i = p_i / k_i$ for
all $i$, satisfies $f(x) - x = 0$. Hence $x$ is a solution of type \solnref{CM1}.

\item In solutions of type \solnref{OV1} we have $i$-slice $s$ and two points $p, q \in
P_s$ with $p \ne q$ such that $D_j(p) = D_j(q) = \zero$ for all $j \le i$. Let
$a$ and $b$ be the two corresponding points in $[0, 1]^n$, meaning that $a_i =
p_i / k_i$ for all $i$, and $b_i = q_i / k_i$ for all $i$.

We first first consider the contraction property in the $[0, 1]^i$ space defined
by the slice $s$. Let $\| a - b \|^i_p$ denote the distance between $a$ and $b$
according to the $\ell_p$ norm restricted to the $[0, 1]^i$ space, and likewise
let $\| f(a) - f(b) \|^i_p$ be the distance between $f(a)$ and $f(b)$ in that
space.
Since $D_j(p) = D_j(q) =  \zero$ for all $j \le i$, this means that we have 
$\| f(a) - f(b) \|^i_p = \| a - b \|^i_p$, which is a clear violation of
contraction in the space $[0, 1]^i$.

To see that $a$ and $b$ also violate contraction in $[0, 1]^d$, observe that 
$\| a - b \|^i_p = \| a - b \|_p$, because $a$ and $b$ lie in the same
$i$-slice, meaning that $a_j = b_j$ for all $j > i$.
Furthermore, we have $\| f(a) - f(b) \|_p \ge \| f(a) - f(b) \|^i_p$, since
adding in extra dimensions can only increase the distance in an $\ell_p$ norm.
Hence we have $\| f(a) - f(b) \|_p \ge \| a - b \|_p$, which is a violation of
contraction, giving us a solution of type \solnref{CMV1}. 

\item Violations of type \solnref{OV2} map directly to violations of type \solnref{CMV3}, as
discussed in the main body.

\item Violations of type \solnref{OV3} give us a point $p$ such that $p_i = 0$ and $D_i(p)
= \down$, or $p_i = k_i$ and $D_i(p) = \up$. In both cases this means that $f(p)
\not\in [0, 1]^d$, and so we have a violation of type \solnref{CMV2}.

\end{itemize}

To see that the reduction is promise-preserving, it suffices to note that
solutions of type \solnref{CM1} are only ever mapped on to solutions of type \solnref{O1}. Hence, if
the input problem is a contraction map, the resulting OPDC instance only has
solutions of type \solnref{CM1}.

\section{Proofs for Section~\ref{sec:PLCPtoEOPL}: PLCP to EOPL and UEOPL}
\label{app:full_plcp_reduction}

\subsection{Background on Lemke's algorithm}
\label{app:lemke}

The explanation of Lemke's algorithm in this section is taken from \cite{GMSV}.
Recall the LCP problem from Definition \ref{def:lcp}, where given a $d\times d$ matrix $M$ and $d$-dimensional vector $q$, we want to find $\yy$ satisfying \eqref{eq:lcp}. That is ($\ww$ is a place-holder variable), 
\[
\ww= \MM \yy  + \pq, \ \ \ \  \yy \geq 0, \ \ \ \ \ww \ge 0 \ \ \ \ \mbox{and} \ \ \ \ y_iw_i = 0,\ \forall i\in[d].  \]

The problem is interesting only when $\qq \not \geq 0$, since otherwise $\yy = 0$ is a trivial solution. 
Let $\CQ$ be the polyhedron in $2d$ dimensional space defined by the first three conditions; we will assume that $\CQ$ is
non-degenerate (just for simplicity of exposition; this will not matter for our reduction).  
Under this condition, any solution to (\ref{eq:lcp}) will be a vertex of $\CQ$, since it must satisfy $2d$
equalities. Note that the set of solutions may be disconnected.
The ingenious idea of Lemke was to introduce a new variable and consider the system:
\begin{equation} \label{eq:c} \ww= \MM \yy  + \pq +z \one, \ \ \ \  \yy \geq 0, \ \ \ \ \ww \geq 0, \ \ \ \  z \geq 0  \ \
\ \ \mbox{and} \ \ \ \ y_iw_i = 0,\ \forall i\in[d].  \end{equation}
The next lemma follows by construction of (\ref{eq:c}).
\begin{lemma}\label{lem:lemke1}
Given $(\MM,\qq)$, $(\yy,\ww,z)$ satisfies \eqref{eq:c} with $z=0$ iff $\yy$ satisfies~\eqref{eq:lcp}.
\end{lemma}
Let $\CPol$ be the polyhedron in $2d + 1$ dimensional space defined by the first four conditions of \eqref{eq:c}, i.e.,
\begin{equation}\label{eq:cp}
\CPol = \{ (\yy,\ww, z) \ |\ \ww=\MM \yy  + \pq +z \one, \ \ \ \yy \geq 0, \ \ \ \ww \geq 0, \ \ \  z \geq 0\};
\end{equation}
for now, we will assume that $\CPol$ is {\em non-degenerate}.  

Since any solution to (\ref{eq:c}) must still satisfy $2d$ equalities in $\CPol$, the set of solutions, say
$S$, will be a subset of the one-skeleton of $\CPol$, i.e., it will consist of edges and vertices of $\CPol$.  Any solution to
the original system (\ref{eq:lcp}) must satisfy the additional condition $z = 0$ and hence will be a vertex of $\CPol$.

Now $S$ turns out to have some nice properties. Any point of $S$ is {\em fully labeled} in the sense that for each $i$, $y_i
= 0$ or $w_i = 0$.  We will say that a point of $S$ {\em has duplicate label i} if $y_i = 0$ and $w_i = 0$ are both satisfied
at this point. Clearly, such a point will be a vertex of $\CPol$ and it will have only one duplicate label.  Since there are
exactly two ways of relaxing this duplicate label, this vertex must have exactly two edges of $S$ incident at it.  Clearly, a
solution to the original system (i.e., satisfying $z = 0$) will be a vertex of $\CPol$ that does not have a duplicate label.  On
relaxing $z=0$, we get the unique edge of $S$ incident at this vertex.

As a result of these observations, we can conclude that every vertex of $S$ with $z >0$ has degree two within $S$, while a vertex with $z=0$ has degree one. Thus, $S$ consists of paths and cycles. Of these paths, Lemke's algorithm
explores a special one.  An unbounded edge of $S$ such that the vertex of $\CPol$ it is incident on has $z > 0$ is called a
{\em ray}.  Among the rays, one is special -- the one on which $\yy = 0$. This is called the {\em primary ray} and the rest
are called {\em secondary rays}. Now Lemke's algorithm explores, via pivoting, the path starting with the primary ray. This
path must end either in a vertex satisfying $z = 0$, i.e., a solution to the original system, or a secondary ray. In the
latter case, the algorithm is unsuccessful in finding a solution to the original system; in particular, the original system
may not have a solution.  
We give the full pseudo-code for Lemke's algorithm in Table~\ref{tab:lemke}.



\begin{table}[!htb]
\caption{Lemke's Complementary Pivot Algorithm}\label{tab:lemke}
\begin{tabular}{|l|}
\hline
\hspace{5pt} {\bf If} $\qq\ge 0$ {\bf then} {\bf Return} $\yy\leftarrow \zeros$ \\
\hspace{5pt} $\yy\leftarrow 0, z\leftarrow |\min_{i \in [d]} q_i|, \ww=\qq+z\ones$\\
\hspace{5pt} $i\leftarrow $ duplicate label at vertex $(\yy,\ww,z)$ in $\CPol$. $flag\leftarrow 1$ \\
\hspace{5pt} {\bf While} $z>0$ {\bf do}\\
\hspace{10pt} {\bf If} $flag=1$ {\bf then} set $(\yy',\ww',z')\leftarrow $ vertex obtained by relaxing $y_i=0$ at $(\yy,\ww,z)$ in $\CPol$\\
\hspace{10pt} {\bf Else} set $(\yy',\ww',z')\leftarrow $ vertex obtained by relaxing $w_i=0$ at $(\yy,\ww,z)$ in $\CPol$\\
\hspace{10pt} {\bf If} $z>0$ {\bf then}\\
\hspace{15pt} $i \leftarrow $ duplicate label at $(\yy',\ww',z')$\\
\hspace{15pt} {\bf If} $v_i>0$ and $v'_i=0$ {\bf then} $flag\leftarrow 1$. {\bf Else} $flag\leftarrow 0$\\
\hspace{15pt} $(\yy,\ww,z)\leftarrow(\yy',\ww',z')$\\
\hspace{5pt} End {\bf While} \\
\hspace{5pt} {\bf Return} $\yy$\\
\hline
\end{tabular}
\end{table}

\newpage

\subsection{Reduction from \PLCP with \solnref{PV1} violations to \EOPL}
\label{app:plcp_pv1_to_eopl}


It is well known that if matrix $M$ is a P-matrix (\PLCP), then $z$ strictly
decreases on the path traced by Lemke's algorithm \cite{cottle2009linear}.
Furthermore, by a result of Todd~\cite[Section 5]{todd1976orientation}, paths traced by
complementary pivot rule can be locally oriented.  Based on these two facts, 
we derive a polynomial-time reduction from \PLCP to \EOPL first, and then from \PLCP to \UEOPL.

Let $\CI=(M,\qq)$ be a given \PLCP instance, and let $\CL$ be the length of the 
bit representation of $M$ and $\qq$. 
We will reduce $\CI$ to an \EOPL instance $\CE$ in time $\poly(\CL)$. 
According to Definition~\ref{def:EOPL}, the instance $\CE$ is defined 
by its vertex set $\vert$, and procedures $S$ (successor), $P$ (predecessor) and $\pot$ (potential). 
Next we define each of these. 

As discussed in Section \ref{app:lemke} the linear constraints of (\ref{eq:c})
on which Lemke's algorithm operates forms a polyhedron $\CPol$ given in
(\ref{eq:cp}). We assume that $\CPol$ is non-degenerate. This is without
loss of generality since, a typical way to ensure this is by perturbing $\qq$ so
that configurations of solution vertices remain unchanged
\cite{cottle2009linear}, and since $M$ is unchanged if $\CI$ was a \PLCP instance then it remains so. 

Lemke's algorithm traces a path on feasible points of (\ref{eq:c}) which is on
$1$-skeleton of $\CPol$ starting at $(\yy^0,\ww^0,z^0)$, where:
\begin{equation}\label{eq:v0}
\yy^0=0,\ \ \ \ \ z^0= |\min_{i \in [d]} q_i|,\ \ \ \ \  \ww^0=\qq+z^0\ones
\end{equation}
We want to capture
vertex solutions of (\ref{eq:c}) as vertices in \EOPL instance $\CE$. To
differentiate we will sometimes call the latter, {\em configurations}. Vertex
solutions of (\ref{eq:c}) are exactly the vertices of polyhedron $\CPol$ with
either $y_i=0$ or $w_i=0$ for each $i\in [d]$. Vertices of (\ref{eq:c}) with
$z=0$ are our final solutions (Lemma \ref{lem:lemke1}). While each of its {\em
non-solution} vertex has a duplicate label. Thus, a vertex of this path can be
uniquely identified by which of $y_i=0$ and $w_i=0$ hold for each $i$ and its
duplicate label. This gives us a representation for vertices in the \EOPL
instance $\CE$. 

\medskip

\noindent{\bf \EOPL Instance $\CE$.}
\begin{itemize}
\item Vertex set $\vert=\{0,1\}^n$ where $n = 2d$. 
\item Procedures $S$ and $P$ as defined in Tables \ref{tab:S} and \ref{tab:P} respectively
\item Potential function $\pot:\vert \rightarrow \{0,1,\dots, 2^m-1\}$ defined in Table \ref{tab:F} for $m=\lceil ln(2\Delta^3)\rceil$, 
	  where $$\Delta=(n! \cdot I_{max}^{2d+1})+1$$ 
	  and $I_{max} = \max\{\max_{i,j\in [d]} M(i,j),\ \max_{i\in [d]} |q_i|\}$. 
\end{itemize}

For any vertex $\uu\in \vert$, the first $d$ bits of $\uu$ represent
which of the two inequalities, namely $y_i\ge 0$ and $w_i\ge 0$, is tight for
each $i \in [d]$. 
\[
\forall i \in [d],\ \ u_i=0 \Rightarrow y_i=0,\ \ \ \mbox{and} \ \ \ u_i=1 \Rightarrow w_i=0
\]

A valid setting of the second set of $d$ bits, namely $u_{d+1}$ through $u_{2d}$, will have 
at most one non-zero bit -- if none is one then $z=0$, otherwise the location of one bit indicates the duplicate label. 
Thus, there are many invalid configurations, namely
those with more than one non-zero bit in the second set of $d$ bits. 
These are dummies that we will handle separately, and we define a procedure 
$\isvalid$ to identify non-dummy vertices in Table \ref{tab:iv}.
To go between ``valid'' vertices of $\CE$ and corresponding vertices of the Lemke polytope
$\CPol$ of LCP $\CI$, we define procedures $\eti$ and $\ite$ in Table
\ref{tab:ei}.

\begin{table}[!hbt]
\caption{Procedure \isvalid(\uu)}\label{tab:iv}
\begin{tabular}{|l|}
\hline
\hspace{5pt} {\bf If} $\uu=0^{n}$ {\bf then} {\bf Return} 1\\
\hspace{5pt} {\bf Else} let $\tau = (u_{(d+1)}+\dots+u_{2d})$\\
\hspace{15pt} {\bf If} $\tau> 1$ {\bf then} {\bf Return} 0\\
\hspace{15pt} Let $S\leftarrow \emptyset$. \% set of tight inequalities. \\
\hspace{15pt} {\bf If} $\tau = 0$ {\bf then} $S=S\cup \{ z=0\}$. \\
\hspace{15pt} {\bf Else}\\
\hspace{30pt} Set $l\leftarrow $ index of the non-zero coordinate in vector $(u_{(d+1)},\dots,u_{2d})$. \\
\hspace{30pt} Set $S=\{y_l=0, w_l=0\}$.\\
\hspace{15pt} {\bf For} each $i$ from $1$ to $d$ {\bf do} \\
\hspace{30pt} {\bf If} $u_i=0$ {\bf then} $S=S\cup \{y_i=0\}$, {\bf Else} $S=S\cup \{w_i=0\}$\\
\hspace{15pt} Let $A$ be a matrix formed by l.h.s. of equalities $M\yy-\ww +\ones z=-\qq$ and that of set $S$\\
\hspace{15pt} Let $\bb$ be the corresponding r.h.s., namely $\bb=[-\qq; \zeros_{(d+1)\times 1}]$.\\
\hspace{15pt} Let $(\yy',\ww',z') \leftarrow \bb * A^{-1}$\\
\hspace{15pt} {\bf If} $(\yy',\ww',z') \in \CPol$ {\bf then} {\bf Return} 1, {\bf Else} {\bf Return} 0 \\
\hline
\end{tabular}
\end{table}

\begin{table}[!hbt]
\caption{Procedures $\ite(\uu)$ and $\eti(\yy,\ww,z)$}\label{tab:ei}
\begin{tabular}{|l|}
\hline
\begin{tabular}{l}
$\ite(\yy,\ww,z)$ \\ \hline
\hspace{5pt} {\bf If} $\exists i \in [d]$ s.t. $y_i * w_i \neq 0$ {\bf then} {\bf Return} $(\zeros_{(2d-2)\times 1};1;1)$ \% Invalid \\
\hspace{5pt} Set $\uu\leftarrow \zeros_{2d\times 1}$. Let $DL=\{i\in [d]\ |\ y_i=0\mbox{ and } w_i=0\}$.\\
\hspace{5pt} {\bf If} $|DL|>1$ {\bf then} {\bf Return} $(\zeros_{(2d-2)\times 1};1;1)$ \%Invalid \\
\hspace{5pt} {\bf If} $|DL|=1$ {\bf then} for $i\in DL$, set $u_{d+i}\leftarrow 1$\\
\hspace{5pt} {\bf For} each $i\in [d]$ {\bf If} $w_i=0$ {\bf then} set $u_{i}\leftarrow 1$\\
\hspace{5pt} {\bf Return} $\uu$
\end{tabular}
\\ \hline
\begin{tabular}{l}
$\eti(\uu)$  \\ \hline
\hspace{5pt} {\bf If} $\uu=0^n$ {\bf then} {\bf Return} $(\zeros_{d \times 1}, \qq+(z^0+1)\ones, z^0+1)$ \% This case will never happen\\
\hspace{5pt} {\bf If} \isvalid(\uu)=0 {\bf then} {\bf Return} $\zeros_{(2d+1) \times 1}$\\
\hspace{5pt} Let $\tau = (u_{(d+1)}+\dots+u_{2d})$\\
\hspace{5pt} Let $S\leftarrow \emptyset$. \% set of tight inequalities. \\
\hspace{5pt} {\bf If} $\tau = 0$ {\bf then} $S=S\cup \{ z=0\}$. \\
\hspace{5pt} {\bf Else}\\
\hspace{15pt} Set $l\leftarrow $ index of non-zero coordinate in vector $(u_{(d+1)},\dots,u_{2d})$. \\
\hspace{15pt} Set $S=\{y_l=0, w_l=0\}$.\\
\hspace{5pt} {\bf For} each $i$ from $1$ to $d$ {\bf do} \\
\hspace{15pt} {\bf If} $u_i=0$ {\bf then} $S=S\cup \{y_i=0\}$, {\bf Else} $S=S\cup \{w_i=0\}$\\
\hspace{5pt} Let $A$ be a matrix formed by lhs of equalities $M\yy-\ww +\ones z=-\qq$ and that of set $S$\\
\hspace{5pt} Let $\bb$ be the corresponding rhs, namely $\bb=[-\qq; \zeros_{(d+1)\times 1}]$.\\
\hspace{5pt} {\bf Return} $\bb * A^{-1}$\\
\end{tabular}\\
\hline
\end{tabular}
\end{table}

By construction of $\isvalid$, $\eti$ and $\ite$, the next lemma follows.

\begin{lemma}\label{lem:vert}
If $\isvalid(\uu)=1$ then $\uu=\ite(\eti(\uu))$, and the corresponding vertex $(\yy,\ww,z)= \eti(\uu)$ of $\CPol$ is feasible in (\ref{eq:c}). If $(\yy,\ww,z)$ is a feasible vertex of (\ref{eq:c}) then $\uu=\ite(\yy,\ww,z)$ is a valid configuration, {\em i.e.,} $\isvalid(\uu)=1$.
\end{lemma}

\begin{proof}
The only thing that can go wrong is that the matrix $A$ generated in $\isvalid$ and $\eti$ procedures are singular, or the set of double labels $DL$ generated in $\ite$ has more than one elements. 
\end{proof}

The main idea behind procedures $S$ and $P$, given in Tables \ref{tab:S} and
\ref{tab:P} respectively, is the following (also see Figure \ref{fig:lemke}):
Make dummy configurations in $\vert$ to point to themselves with cycles of
length one, so that they can never be solutions or violations. 
The starting vertex $0^n \in \vert$ points to the configuration that corresponds
to the first vertex of the Lemke path, namely $\uu^0=\ite(\yy^0,\ww^0,z^0)$. 
Precisely, $S(0^n)=\uu^0$, $P(\uu^0)=0^n$ and $P(0^n)=0^n$ (start of
a path). 

For the remaining cases, let $\uu\in \vert$ have corresponding representation
$\xx=(\yy,\ww,z)\in \CPol$, and suppose $\xx$ has a duplicate label. As one
traverses a Lemke path for a P-LCP, the value of $z$ monotonically decreases \cite{}.
So, for $S(\uu)$ we compute the adjacent vertex $\xx'=(\yy',\ww',z')$ of $\xx$
on Lemke path such that the edge goes from $\xx$ to $\xx'$, and if the $z'<z$,
as expected, then we point $S(\uu)$ to configuration of $\xx'$ namely
$\ite(\xx')$. Otherwise, we let $S(\uu)=\uu$. Similarly, for $P(\uu)$, we find
$\xx'$ such that edge is from $\xx'$ to $\xx$, and then we let $P(\uu)$ be
$\ite(\xx')$ if $z'>z$ as expected, otherwise $P(\uu)=\uu$. 

For the case when $\xx$ does not have a duplicate label, then we have $z=0$. This is
handled separately since such a vertex has exactly one incident edge on the Lemke
path, namely the one obtained by relaxing $z=0$. According to the direction of 
this edge, we do similar process as before. For example, if the edge goes from 
$\xx$ to $\xx'$, then, if $z'<z$, we set $S(\uu)=\ite(\xx')$ else $S(\uu)=\uu$,
and we always set $P(\uu)=\uu$.  In case the edge goes from $\xx'$ to $\xx$, we
always set $S(\uu)=\uu$, and we set $P(\uu)$ depending on whether or not $z'>z$.

%
\medskip

The potential function $\pot$, formally defined in Table \ref{tab:F},
gives a value of zero to dummy vertices and the starting vertex $0^n$. To all
other vertices, essentially it is $((z^0-z) * \Delta^2)+1$. Since value of $z$
starts at $z^0$ and keeps decreasing on the Lemke path this value will keep
increasing starting from zero at the starting vertex $0^n$. Multiplication by
$\Delta^2$ will ensure that if $z_1>z_2$ then the corresponding potential values 
will differ by at least one. This is because, since $z_1$ and $z_2$ are 
coordinates of two vertices of polytope $\CPol$, their maximum value is $\Delta$
and their denominator is also bounded above by~$\Delta$. Hence $z_1-z_2\le
1/\Delta^2$ (Lemma \ref{lem:pot}).  

To show correctness of the reduction we need to show two things: $(i)$ All the
procedures are well-defined and polynomial time. $(ii)$ We can construct a
solution of $\CI$ from a solution of $\CE$ in polynomial time. 

\begin{table}
\begin{minipage}{0.73\textwidth}
\caption{Successor Procedure $S(\uu)$}\label{tab:S}
\begin{tabular}{|l|}
\hline
\hspace{0pt}{\bf If} $\isvalid(\uu)=0$ {\bf then} {\bf Return} $\uu$\\
\hspace{0pt}{\bf If} $\uu=0^n$ {\bf then} {\bf Return} $\ite(\yy^0,\ww^0,z^0)$\\
\hspace{0pt}$\xx=(\yy,\ww,z) \leftarrow \eti(\uu)$\\
\hspace{0pt}{\bf If} $z=0$ {\bf then} \\
\hspace{5pt} $\xx^1\leftarrow$ vertex obtained by relaxing $z=0$ at $\xx$ in $\CPol$. \\
\hspace{5pt} {\bf If} Todd \cite{todd1976orientation} prescribes edge from $\xx$ to $\xx^1$ \\
\hspace{10pt} {\bf then} set $\xx'\leftarrow \xx^1$. {\bf Else Return} $\uu$ \\
\hspace{0pt}{\bf Else} set $l\leftarrow $ duplicate label at $\xx$\\
\hspace{5pt} $\xx^1\leftarrow $ vertex obtained by relaxing $y_l=0$ at $\xx$ in $\CPol$ \\
\hspace{5pt} $\xx^2\leftarrow $ vertex obtained by relaxing $w_l=0$ at $\xx$ in $\CPol$ \\
\hspace{5pt} {\bf If} Todd \cite{todd1976orientation} prescribes edge from $\xx$ to $\xx^1$ \\
\hspace{10pt} {\bf then} $\xx'=\xx^1$ \\
\hspace{5pt} {\bf Else} $\xx'=\xx^2$\\
\hspace{0pt}Let $\xx'$ be $(\yy',\ww',z')$. \\
\hspace{0pt}{\bf If} $z>z'$ {\bf then} {\bf Return} $\ite(\xx')$. {\bf Else} {\bf Return} $\uu$.\\
\hline
\end{tabular}
\end{minipage}%
\hspace{-1cm}
\begin{minipage}{0.23\textwidth}
\caption{Potential Value $\pot(\uu)$}\label{tab:F}
\begin{tabular}{|l|}
\hline
\hspace{0pt} {\bf If} $\isvalid(\uu)=0$ \\
\hspace{5pt} {\bf then} {\bf Return} $0$\\
\hspace{0pt} {\bf If} $\uu=0^n$\\
\hspace{5pt}  {\bf then} {\bf Return} $0$\\
\hspace{0pt} $(\yy,\ww,z) \leftarrow \eti(\uu)$\\
\hspace{0pt} {\bf Return} $\lfloor \Delta^2*(\Delta -z)\rfloor$\\
\hline
\end{tabular}
\end{minipage}
\end{table}

\begin{table}[!htb]
\caption{Predecessor Procedure $P(\uu)$}\label{tab:P}
\begin{tabular}{|l|}
\hline
\hspace{0pt} {\bf If} $\isvalid(\uu)=0$ {\bf then} {\bf Return} $\uu$\\
\hspace{0pt} {\bf If} $\uu=0^n$ {\bf then} {\bf Return} $\uu$\\
\hspace{0pt} $(\yy,\ww,z) \leftarrow \eti(\uu)$\\
\hspace{0pt} {\bf If} $(\yy,\ww,z)=(\yy^0,\ww^0,z^0)$ {\bf then} {\bf Return} $0^n$\\
\hspace{0pt} {\bf If} $z=0$ {\bf then} \\
\hspace{5pt} $\xx^1\leftarrow$ vertex obtained by relaxing $z=0$ at $\xx$ in $\CPol$. \\
\hspace{5pt} {\bf If} Todd \cite{todd1976orientation} prescribes edge from $\xx^1$ to $\xx$ {\bf then} set $\xx'\leftarrow \xx^1$. {\bf Else Return} $\uu$\\
\hspace{0pt} {\bf Else}\\
\hspace{5pt} $l\leftarrow $ duplicate label at $\xx$\\
\hspace{5pt} $\xx^1\leftarrow $ vertex obtained by relaxing $y_l=0$ at $\xx$ in $\CPol$ \\
\hspace{5pt} $\xx^2\leftarrow $ vertex obtained by relaxing $w_l=0$ at $\xx$ in $\CPol$ \\
\hspace{5pt} {\bf If} Todd \cite{todd1976orientation} prescribes edge from $\xx^1$ to $\xx$ {\bf then} $\xx'=\xx^1$ {\bf Else} $\xx'=\xx^2$\\
\hspace{0pt} Let $\xx'$ be $(\yy',\ww',z')$. {\bf If} $z<z'$ {\bf then} {\bf Return} $\ite(\xx')$. {\bf Else} {\bf Return} $\uu$.\\
\hline
\end{tabular}
\end{table}

\begin{lemma}\label{lem:PSF}
Functions $P$, $S$ and $\pot$ of instance $\CE$ are well defined, making $\CE$ a valid \EOPL instance. 
\end{lemma}

\begin{proof}
Since all three procedures are polynomial-time in $\CL$, they can be defined
by $\poly(\CL)$-sized Boolean circuits. Furthermore, for any $\uu \in \vert$,
we have that $S(\uu),P(\uu) \in \vert$. For~$\pot$, 
since the value of $z \in [0,\ \Delta-1]$, we
have $0\le \Delta^2(\Delta-z)\le \Delta^3$. Therefore, $\pot(\uu)$ is an
integer that is at most $2 \cdot \Delta^3$ and hence is in set $\{0,\dots, 2^m-1\}$. 
\end{proof}

There are two possible types of solutions of an \EOPL instance 
(see Definition~\ref{def:EOPL}). One indicates
the beginning or end of a line \solnref{R1}, and the other is a vertex with locally optimal
potential \solnref{R2}. 
First we show that \solnref{R2} never arises. 
For this, we need the next lemma, which shows that potential differences in two
adjacent configurations adheres to differences in the value of~$z$ at
corresponding vertices.

\begin{lemma}\label{lem:pot}
Let $\uu \neq \uu'$ be two valid configurations, i.e.,
	$\isvalid(\uu)=\isvalid(\uu')=1$, and let $(\yy,\ww,z)$ and $(\yy',\ww',z')$
	be the corresponding vertices in $\CPol$. Then the following holds: 
	\begin{enumerate}
	\item[$(i)$] $\pot(\uu)=\pot(\uu')$ iff $z=z'$. 
	\item[$(ii)$] $\pot(\uu)>\pot(\uu')$ iff $z<z'$.
	\end{enumerate}
\end{lemma}

\begin{proof}
Among the valid configurations all except $\zeros$ has positive $\pot$ value.
Therefore, wlog let $\uu,\uu'\neq \zeros$. For these we have $\pot(\uu)=\lfloor
\Delta^2 \cdot (\Delta -z)\rfloor$, and $\pot(\uu')=\lfloor \Delta^2 \cdot (\Delta
-z')\rfloor$. 

Note that since both $z$ and $z'$ are coordinates of vertices of $\CPol$, whose
description has highest coefficient of $\max\{\max_{i,j\in [d]}
M(i,j),\max_{i\in [d]} |q_i|\}$, and therefore their numerator and denominator
both are bounded above by $\Delta$. Therefore, if $z< z'$ then we have 
\[
z'-z\ge \frac{1}{\Delta^2} \Rightarrow ((\Delta-z) - (\Delta - z')) \cdot \Delta^2 \ge 1 \Rightarrow \pot(\uu)-\pot(\uu') \ge 1.
\]

For $(i)$, if $z=z'$ then clearly $\pot(\uu)=\pot(\uu')$, and from the above argument it also follows that if $\pot(\uu)= \pot(\uu')$ then it can not be the case that $z\neq z'$. Similarly for $(ii)$, if $\pot(\uu)>\pot(\uu')$ then clearly, $z'>z$, and from the above argument it follows that if $z'>z$ then it can not be the case that $\pot(\uu')\ge \pot(\uu)$. 
\end{proof}

Using the above lemma, we will next show that instance $\CE$ has no local maximizer. 

\begin{lemma}\label{lem:t}
Let $\uu,\vv \in \vert$ s.t. $\uu\neq \vv$, $\vv=S(\uu)$, and $\uu=P(\vv)$. Then $\pot(\uu)< \pot(\vv)$.
\end{lemma}
\begin{proof}
Let $\xx=(\yy,\ww,z)$ and $\xx'=(\yy',\ww',z')$ be the vertices in polyhedron $\CPol$ corresponding to $\uu$ and $\vv$ respectively. From the construction of $\vv=S(\uu)$ implies that $z'<z$. Therefore, using Lemma \ref{lem:pot} it follows that $\pot(\vv)<\pot(\uu)$.
\end{proof}

Due to Lemma \ref{lem:t} the only type of solutions available in $\CE$ is \solnref{R1} where $S(P(\uu))\neq \uu$ and $P(S(\uu))\neq \uu$. Next two lemmas shows how to construct solution of P-LCP instance $\CI$ or a \solnref{PV1} type violation (non-positive principle minor of matrix $M$) from these. 

\begin{lemma}\label{lem:t1}
Let $\uu \in \vert$, $\uu \neq 0^n$. 
If $P(S(\uu))\neq \uu$ or $S(P(\uu))\neq \uu$, then $\isvalid(\uu)=1$. Futhermore, for $(\yy,\ww,z)=\eti(\uu)$ if $z=0$, then $\yy$ is a $\PLo$ type solution of \PLCP instance $\CI=(M,\qq)$. 
\end{lemma}
\begin{proof}
By construction, if $\isvalid(\uu) = 0$, then $S(P(\uu))=\uu$ and $P(S(\uu))=\uu$, therefore $\isvalid(\uu)=0$ when $\uu$ has a predecessor or successor different from $\uu$.
Given this, from Lemma \ref{lem:vert} we know that $(\yy,\ww,z)$ is a feasible vertex in (\ref{eq:c}). 
Therefore, if $z=0$, then by Lemma \ref{lem:lemke1} we have a solution of the LCP (\ref{eq:lcp}), {\em i.e.,} a type $\PLo$ solution of our \PLCP instance $\CI=(\MM,\qq)$.
%
%
\end{proof}

\begin{lemma}\label{lem:t2}
Let $\uu \in \vert$, $\uu \neq 0^n$ such that $P(S(\uu))\neq \uu$ or
$S(P(\uu))\neq \uu$, and let $\xx=(\yy,\ww,z)=\eti(\uu)$. 
If $z\neq 0$ then $\xx$ has a duplicate label, say $l$. And for directions
$\sigma_1$ and $\sigma_2$ obtained by relaxing $y_l=0$ and $w_l=0$ respectively
at $\xx$, we have $\sigma_1(z) \cdot \sigma_2(z)\ge 0$, where $\sigma_i(z)$ is
the coordinate corresponding to $z$. 
\end{lemma}
\begin{proof}
From Lemma \ref{lem:t1} we know that $\isvalid(\uu)=1$, and therefore from Lemma \ref{lem:vert}, $\xx$ is a feasible vertex in (\ref{eq:c}).
From the last line of Tables \ref{tab:S} and \ref{tab:P} observe that $S(\uu)$ points to the configuration of vertex next to $\xx$ on Lemke's path only if it has lower $z$ value otherwise it gives back $\uu$, and similarly $P(\uu)$ points to the previous only if value of $z$ increases.


First consider the case when $P(S(\uu))\neq \uu$. Let $\vv=S(\uu)$ and corresponding vertex in $\CPol$ be $(\yy',\ww',z')=\eti(\vv)$. 
If $\vv\neq \uu$, then from the above observation we know that $z'>z$, and in that
case again by construction of $P$ we will have $P(\vv)=\uu$, contradicting
$P(S(\uu))\neq \uu$. Therefore, it must be the case that $\vv=\uu$.
Since $z\neq 0$ this happens only when the next vertex on Lemke path after $\xx$ has
higher value of $z$ (by above observation). As a consequence of $\vv=\uu$, we also have $P(\uu)\neq \uu$. By construction of $P$ this implies for 
$(\yy'',\ww'',z'')=\eti(P(\uu))$, $z''>z$. Putting both together we get 
increase in $z$ when we relax $y_l=0$ as well as when we relax $w_l=0$ at
$\xx$.

For the second case $S(P(\uu))\neq \uu$ similar argument gives that value of $z$ decreases when we relax $y_l=0$ as well as when we relax $w_l=0$ at
$\xx$. The proof follows.
\end{proof}

Finally, we are ready to prove our main result of this section using Lemmas
\ref{lem:t}, \ref{lem:t1} and \ref{lem:t2}. Together with Lemma \ref{lem:t2},
we will use the fact that on Lemke path $z$ monotonically decreases if $M$ is a
P-matrix or else we get a \solnref{PV1} type witness that $M$ is not a
P-matrix~\cite{cottle2009linear}. 

\begin{theorem}
\label{thm:plcp-eopl}
There is a polynomial-time promise-preserving reduction from \PLCP with \solnref{PV1} violations to \EOPL. 
\end{theorem}
\begin{proof}
Given an instance of $\CI=(\MM,\qq)$ of \PLCP, where $M\in \Real^{d\times d}$
and $\qq\in \Real^{d\times 1}$ reduce it to an instance $\CE$ of \EOPL as
described above with vertex set $\vert=\{0,1\}^{2d}$ and procedures $S$, $P$ and
$\pot$ as given in Table \ref{tab:S}, \ref{tab:P}, and \ref{tab:F} respectively.

Among solutions of \EOPL instance $\CE$, there is no local potential maximizer,
	i.e., $\uu\neq \vv$ such that $\vv=S(\uu)$, $\uu=P(\vv)$ and $\pot(\uu)>\pot(\vv)$
	due to Lemma \ref{lem:t}. We get a solution $\uu \neq 0$ such that either
	$S(P(\uu))\neq \uu$ or $P(S(\uu))\neq \uu$, then by Lemma \ref{lem:t1} it is
	valid configuration and has a corresponding vertex $\xx=(\yy,\ww,z)$ in
	$\CPol$. Again by Lemma~\ref{lem:t1} if $z=0$ then $\yy$ is a $\PLo$ type solution
	of our \PLCP instance $\CI$. On the other hand, if $z>0$ then from Lemma
	\ref{lem:t2} we get that on both the two adjacent edges to $\xx$ on Lemke
	path the value of $z$ either increases or deceases. This gives us a minor of $M$
	which is non-positive~\cite{cottle2009linear}, 
	i.e., a \PLt type solution of the \PLCP instance~\CI\ with \solnref{PV1} violation.

	The reduction is promise preserving because if the LCP instance is promised to be \PLCP then $z$ monotonically decreases along the Lemke's path, and all feasible complementary vertices are on this path. Therefore, the corresponding \EOPL instance will have exactly one path ending in a solution where the corresponding vertex $\xx=(\yy,\ww,z)$ of the LCP has $z=0$ mapping to the \PLCP solution. 
\end{proof}

\subsection{Reduction from \PLCP with \solnref{PV2} violations to \UEOPL}
\label{app:plcp_pv2_to_ueopl}

Next we show that the above construction also implies \PLCP with \solnref{PV2}
violations is in \UEOPL, and thereby also in \EOPL.  
We start with a simple well-known lemma that turns two solutions of an LCP $(M,
\qq)$ into a non-zero sign-reversing vector for $M$, i.e. a \solnref{PV2}
violation.

\begin{lemma}\cite[Theorem 3.3.7]{cottle2009linear}
\label{lem:2solTOPV2}
If for some $d$-dimensional vector $\qq'$, LCP $(M,\qq')$ has more than one solution then there exists a sign-reversing vector $\xx$ w.r.t. $M$, i.e., $x_i(Mx)_i \le 0,\ \ \forall i \in [d]$.
\end{lemma}
\begin{proof}
Let $(\yy',\ww')$ and $(\yy^*,\ww^*)$ be two distinct solutions of the LCP defined by $(M,\qq')$. That is $\yy'\neq \yy^*$. Then,
\[
\ww^* = M\yy^*+\qq' \ \ \mbox{and} \ \ \ww' = M\yy'+\qq' \Rightarrow (\ww^* - \ww') = M(\yy^* - \yy')
\]
Furthermore, for each $i\in [d]$ we have $w^*_i y^*_i =0, w'_i y'_i =0, w^*_i y'_i \ge 0$ and $w'_i y^*_i \ge 0$. This together with $(\ww^* - \ww')_i = (M(\yy^* - \yy'))_i$ gives,
\[
\forall i \in [d], \ \ \ (\yy^* -\yy')_i (\ww^* - \ww')_i \le 0\ \ \  \Rightarrow \ \ \ (\yy^* -\yy')_i (M(\yy^* - \yy'))_i \le 0
\]
Thus, $\xx = \yy^* -\yy'$ is our desired vector. Note that $\xx \neq 0$ since $\yy' \neq \yy^*$.
\end{proof}

\UEOPL has four types of solutions. Out of these, \solnref{UV1} is ruled out by
Lemma \ref{lem:t}. Next we show that any ``extra'' end of lines as well as
\solnref{UV3} type solutions map to a \solnref{PV2} violation.

\begin{lemma}\label{lem:plcp-2sol}
Given either of the following, we can construct two distinct solutions of LCP $(M,\qq')$ for some~$\qq'$:
\begin{itemize}
\item[$(a)$] $\uu\in \vert$ is a \solnref{U1} or \solnref{UV2} type solution of instance $\CE$ such that corresponding vertex $\xx=(\yy,\ww,z)=\eti(\uu)$ has $z^*>0$. 
\item[$(b)$] $\uu,\vv \in \vert$ forms a \solnref{UV3} type solution of instance $\CE$.
\end{itemize}
\end{lemma}
\begin{proof}
The common idea to go from $(a)$ or $(b)$ to two solutions of some LCP with matrix $M$ is to create two (or more) solutions of \eqref{eq:c} with same $z$ value. Suppose $(\yy^*,\ww^*,a)$ and $(\yy',\ww',a)$ with $\yy \neq \yy'$ are feasible in \eqref{eq:c} for some $a\in \mathbb R$, then clearly for $\qq'=\qq+a$, $(\yy^*,\ww^*)$ and $(\yy',\ww')$ are solutions of LCP \eqref{eq:lcp} with matrix $M$ and vector $\qq'$.

For $(a)$, let $\xx^*=(\yy^*,\ww^*,z^*)=\eti(\uu)$ with $z^*>0$, and let $l$ be the duplicate label at vertex $\xx^*$ in $\PLo$. Then from Lemma \ref{lem:t2} we know that for directions $\sigma_1$ and $\sigma_2$ obtained by relaxing $y_l=0$ and $w_l=0$ respectively at $\xx^*$, we have $\sigma_1(z)*\sigma_2(z)\ge 0$, where $\sigma_i(z)$ is the coordinate corresponding to $z$. Suppose $\sigma_1(z), \sigma_2(z)<0$, and for $i=1,2$, let $z_i$ be the value of $z$ at the vertex adjacent to $\xx^*$ in direction of $\sigma_i$; set $z_i = -\infty$ if no vertex encountered in direction $\sigma_i$. Let $\epsilon>0$ be small enough so that $\epsilon < (z^* - z_i),\ i=1,2$, and consider the points $\xx^i = (\xx^* + \frac{\epsilon}{|\sigma_i(z)|} \sigma_i)$ on the edge corresponding to $\sigma_i$ adjacent to $\xx^*$. It is easy to check that by choice of $\epsilon$, both $\xx^1$ and $\xx^2$ are feasible. We will next show that these are solutions of an LCP defined by $(M,\qq')$ for some $\qq'$.

Note that by construction the $z$ coordinate at both $\xx^1$ and $\xx^2$ is $(z^*-\epsilon)$, giving us desired two solutions of \eqref{eq:c} with the same $z$ value. Similar, argument holds when $\sigma_1(z), \sigma_2(z)>0$ where the corresponding $z$ value is $(z^* + \epsilon)$. If either $\sigma_1(z)$ or $\sigma_2(z)$ is zero, then $z$ remains unchanged on the entire corresponding edge. 

For $(b)$, $\xx^* = (\yy^*,\ww^*,z^*)= \eti(\uu)$ and $\xx' = (\yy',\ww',z') = \eti(\vv)$. If $\pot(\uu) = \pot(\vv)$, then clearly $z^*=z'$ (Lemma \ref{lem:pot}) and we get the desired two points feasible in \eqref{eq:c} with the same $z$ value. If $\pot(\uu) < \pot(\vv) < \pot(S(\uu))$, then there exists a point on the edge joining $\xx^*$ with $\eti(S(\xx^*))$ with the same $z$ value as $z'$. 
\end{proof}

Now we are ready to show our main result of \PLCP with PV2 in \UEOPL using Lemmas \ref{lem:t}, \ref{lem:t1}, \ref{lem:2solTOPV2} and \ref{lem:plcp-2sol}

\begin{theorem}
\label{thm:plcp-ueopl}
There is a polynomial-time promise-preserving reduction from \PLCP with \solnref{PV2} violations to \UEOPL,
and thereby also to \EOPL. 
\end{theorem}
\begin{proof}
Given an instance of $\CI=(\MM,\qq)$ of \PLCP, where $M\in \Real^{d\times d}$
and $\qq\in \Real^{d\times 1}$ reduce it to an instance $\CE$ of \UEOPL as
described above with vertex set $\vert=\{0,1\}^{2d}$ and procedures $S$, $P$ and
$\pot$ as given in Table \ref{tab:S}, \ref{tab:P}, and \ref{tab:F} respectively.

Lemma \ref{lem:t} rules out UV1 violation in $\CE$. If we get U1 solution or UV2
violation $\uu$ of $\CE$, then corresponding vertex $\xx=(\yy,\ww,z)$ is
feasible in \eqref{eq:c} by Lemma \ref{lem:t1}. Furthermore, if $z=0$ then $\yy$
is a $\PLo$ type solution of our \PLCP instance $\CI$. On the other hand if
$z>0$, then by Lemmas \ref{lem:plcp-2sol} and \ref{lem:2solTOPV2} we can
construct a PV2 violation of our \PLCP instance $\CI$. Similarly, Lemmas
\ref{lem:plcp-2sol} and \ref{lem:2solTOPV2} also map any UV3 violation of
\UEOPLc instance $\CE$ to a PV2 violation of \CI. 

By construction, if $\CI$ is a promise \PLCP instance, then instance $\CE$ of
\UEOPLc will have exactly one U1 solution corresponding to the unique solution
of the $\CI$.
\end{proof}


\section{Proofs for Section~\ref{sec:algorithms}: Algorithms for Contraction Maps}
\label{app:algorithm_details}


In this section, we provide an exact algorithm for solving \LCM, i.e., we either return a rational fixpoint of polynomial bit-length or a pair of points that prove (indirectly) that the given function is not a contraction map. 
which is guaranteed to have a rational fixpoint
of polynomial bit-length or two points that prove (indirectly) that the given function is not contracting.
Then we extend this algorithm to find an approximate fixpoint of general 
contraction maps for which there may not be an exact solution of
polynomial bit length. In both cases, the problems solved by our algorithm are not promise problems and we always return either a solution or a violation.
Our algorithms work for any $\ell_p$ norm with $p \in \Natural$, and are polynomial for constant dimension $d$.
These are the first such algorithms for $p \ne 2$. 
Such algorithms were so far only known for the $\ell_2$ and $\ell_\infty$
norms \cite{HuangKhachSik99,Sik01,ShellSik03}\footnote{Our approach does not cover the $\ell_\infty$
norm, as that would require more work and not give a new result.}

\subsection{Overview: algorithm to find a fixed-point of \LCM}


The algorithm does a nested binary search using Lemmas \ref{lem:cm1} and
\ref{lem:cm2} to find fixpoints of slices with increasing numbers of free
coordinates. We illustrate the algorithm in two dimensions in
Figure~\ref{fig:exact_algo}. The algorithm is recursive. To find the eventual
fixpoint in $d$ dimensions we fix a single coordinate $s_1$, find the
unique $(d-1)$-dimensional fixpoint of $\restr{f}{\ss}$, the
$(d-1)$-dimensional contraction map obtained by fixing the first coordinate of
the input to $f$ to be $s_1$. Let $x$ the unique fixpoint of $\restr{f}{\ss}$
where $x_1 = s_1$. If $f(x_1) > s_1$, then the $d$-dimensional fixpoint $x^*$ of
$f$ has $x^*_1 > s_1$, and if $f(x_1) < s_1$, then $x^*_1 < s_1$ (Lemma \ref{lem:cm2}). We can thus do
a binary search for the value of $x^*_1$. Once we've found $x^*_1$, we can
recursively find the $(d-1)$-dimensional fixpoint of $\restr{f}{\ss}$ where $s_1
= x_1$. The resulting solution will be the $d$-dimensional fixpoint. At each
step in the recursive procedure, we do a binary search for the value of one
coordinate of the fixpoint at the slice determined by all the coordinates
already fixed. For piecewise-linear functions, we know that all fixpoints are
rational with bounded bit-length (as discussed in Section \ref{app:points2}), so we can find each coordinate exactly.

If at any step in recursion our binary search finds two $k-1$ dimensional fixpoints on slices that are adjacent, differing only in the $k$th coordinate and by a small enough amount , we can return these points, which witness the failure of $f$ to be a contraction map. These points correspond to a solution of type \solnref{CMV3} to the \LCM problem. The proof that $f$ is not a contraction is indirect, and uses the fact that the discretized grid implicitly searched by the algorithm will contain every fixpoint of $f$. Since we maintain the invariant that our two pivots bound the coordinate we're searching over from above and below when $f$ is a contraction map, such a pair of points gives proof that $f$ is not contracting.

\begin{figure}[h!]
  \centering
  \begin{center}
    \begin{tikzpicture}
  \pgfmathsetmacro{\size}{3}
  \pgfmathsetmacro{\rad}{0.6}
  \coordinate (bl) at (-\size, -\size);
  \coordinate (tl) at (-\size, \size);
  \coordinate (tr) at (\size, \size);
  \coordinate (br) at (\size, -\size);
  \coordinate (first) at (0, 1.4);
  \coordinate (second) at ($(\size/2, -0.6)$);
  \coordinate (third) at ($(\size/4, 0.8)$); 
  \draw[thick] (tl) -- (tr) -- (br) -- (bl) -- (tl);
  \draw[thick, dashed, darkgray] (0, \size) -- (0, -\size);
  \draw[thick, dashed, darkgray] ($(\size/2, \size)$) -- ($(\size/2, -\size)$);
  \draw[thick, dashed, darkgray] ($(\size/4, \size)$) -- ($(\size/4, -\size)$);
    
  \fill (first) circle [radius=\rad mm];
  \fill (second) circle [radius=\rad mm];
  \fill (third) circle [radius=\rad mm];

  \coordinate (fourth) at ($(0, 0)$);
  \coordinate (fifth) at ($(0, \size/2)$);
  \coordinate (sixth) at ($(0, 3*\size/4)$);
  \coordinate (actual) at ($(1.2, 0.3)$);

  \coordinate (b1) at (0,-\size);
  \coordinate (b2) at (\size/2,-\size);
  \coordinate (b3) at (\size/4,-\size);


  \fill (actual) circle [radius=\rad mm];
  \node [right=1mm of actual] {$x^* = f(x^*)$};

  \node [below=0.5mm of b1] {$\frac{1}{2}$};
  \node [below=0.5mm of b2] {$\frac{3}{4}$};
  \node [below=0.5mm of b3] {$\frac{5}{8}$};

  \draw[->, thick] (first.center) -- ++(1.7, 0.0);
  \draw[->, thick] (second.center) -- ++(-0.4, 0.0);
  \draw[->, thick] (third.center) -- ++(0.5, 0.0);

  \begin{pgfonlayer}{background}
    \fill[pastelred, opacity=0.2] (0, -\size) rectangle (\size, \size);
    \fill[pastelred, opacity=0.4] ($(0, -\size)$) rectangle ($(\size/2, \size)$);
    \fill[pastelred, opacity=0.5] ($(\size/4, -\size)$) rectangle ($(\size/2, \size)$);
  \end{pgfonlayer}{background}
\end{tikzpicture}
  \end{center}
  \caption{An illustration of the algorithm to find a fixpoint of a piecewise-linear contraction map in two dimensions. The algorithm begins by finding a fixpoint along the slice with $x_1 = 1/2$. The fixpoint along that slice points to the right, so we next find a fixpoint along the slice with $x_1 = 3/4$. The fixpoint along that slice points to the left, so we find the fixpoint along $x_1 = 5/8$. We successively find fixpoints of one-dimensional slices, and then use those to do a binary search for the two-dimensional fixpoint. The red regions are the successive regions considered by the binary search, where each successive step in the binary search results in a darker region.} 
  \label{fig:exact_algo}
\end{figure}
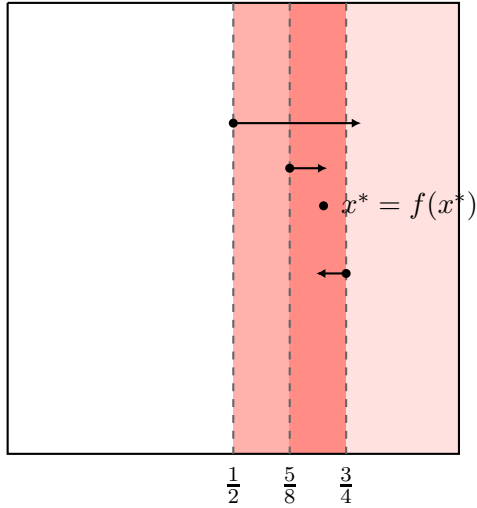

Using this algorithm we obtain the following theorem.

\begin{theorem}
Given a $\LinearFIXP$ circuit $C$ purporting to encode a contraction map $f : [0,1]^d\to
[0,1]^d$ with respect to any $\ell_p$ norm, there is an algorithm to find a
fixpoint of $f$ or return a pair of points witnessing that $f$ is not a contraction map in time that is polynomial in $\size(C)$ and exponential in $d$.
\end{theorem}

The full details of the algorithm can be found in Appendix~\ref{subsec:exact_algo_details}.

\subsection{Overview: algorithm to find an approximate fixed-point of \CM}

Here we generalize our algorithm to find an approximate fixpoint of an arbitrary function
given by an arithmetic circuit, i.e., our algorithm solves \CM, which is 
specified by a circuit $f$ that represents the contraction map,\footnote{The algorithm works even if $f$ is given as an arbitrary black-box, as long as it is guaranteed to be a contraction map.} a
$p$-norm, and $\eps$. Again, let $d$ denote the dimension of the problem, i.e. the number
of inputs (and outputs) of $f$.
Let $x^*$ denote the unique exact fixpoint for the contraction map $f$.
We seek an approximate fixpoint, i.e., a point for which $\Norm{f(x)-x}_p \leq \eps$. 

We do the same recursive binary search as in the algorithm above, but
at each step of the algorithm instead of finding an
exact fixpoint, we will only find an approximate fixpoint of $\restr{f}{\ss}$. The difficulty in this case will come from the fact that Lemma~\ref{lem:cm2} does not apply to approximate fixpoints. Consider the example illustrated in Figure~\ref{fig:approx_algo}. In this example, $y$ is the unique fixpoint of the slice restriction along the gray dashed line. By Lemma~\ref{lem:cm2}, $(f(y)_1 - y_1)(x^*_1 - y_1) \geq 0$ so if we find $y$, we can observe that $f(y)_1 > y_1$ and recurse on the right side of the figure, in the region labeled $\mathcal{R}$. If we try to use the same algorithm but where we only find approximate fixopints at each step, we'll run into trouble. In this case, if we found $z$ instead of $y$, we would observe that $f(z)_1 < z_1$ and conclude that $x^*_1 < z_1$, which is incorrect. As a result, we would limit our search to the region labeled $\mathcal{L}$, and wouldn't be able to find $x^*$. 
\begin{figure}[h!]
  \centering
  \begin{center}
    \begin{tikzpicture}
  \pgfmathsetmacro{\size}{2.5}
  \pgfmathsetmacro{\rad}{0.6}
  \coordinate (apxc) at (-0.6, -0.5);
  \coordinate (actualc) at (-0.6, 0.5);
  \coordinate (globalc) at (0.5, 1.6);
  \coordinate (apxc end) at ($(apxc) +(-0.2, 0.6)$);
  \coordinate (actualc end) at ($(actualc) +(0.6, 0)$);
  \coordinate (bl) at (-\size, -\size);
  \coordinate (tl) at (-\size, \size);
  \coordinate (tr) at (\size, \size);
  \coordinate (br) at (\size, -\size);
  \coordinate (slice top) at (-0.6, \size);
  \coordinate (slice bottom) at (-0.6, -\size);
    
  \draw[thick] (tl) -- (tr) -- (br) -- (bl) -- (tl);
  \draw[thick, dashed, darkgray] (slice bottom) -- (slice top);
  
  \fill (apxc) circle [radius=\rad mm] node (apx) {}; 
  \fill (actualc) circle [radius=\rad mm] node (actual) {}; 
  \fill[gray] (apxc end) circle [radius=\rad mm, fill=darkgray] node (apx end) {}; 
  \fill[gray] (actualc end) circle [radius=\rad mm, fill=darkgray] node (actual end) {}; 
  \fill (globalc) circle [radius=\rad mm] node (global) {}; 
    
  \draw[->, thick] (apx.center) -- (apxc end) node (apx end) [left] {\footnotesize $f(z)$};
  \draw[->, thick] (actual.center) -- (actualc end) node (actual end) [above] {\footnotesize $f(y)$}; 
  \node [above right] at (global) {\footnotesize $x^* = f(x^*)$};
  \node [above left] at (actual) {\footnotesize $y$};
  \node [below right] at (apx) {\footnotesize $z$};
  \node [below] at (slice bottom) {\footnotesize $\ss$};
  \begin{pgfonlayer}{background}
    \node[above right=0.5cm of bl] (L) {\large $\mathcal{L}$}; 
    \node[above left=0.5cm of br] (R) {\large $\mathcal{R}$}; 
    \draw[->, thick, darkpastelblue] (apx.center) -- ++(1.1, 0.0);
    \draw[->, thick, darkpastelblue] (actual.center) -- ++(1.1, 0);
    \draw[->, thick, darkpastelred] (apx.center) -- ++(-0.2, 0.0);
  \end{pgfonlayer}
\end{tikzpicture}
  \end{center}
  \caption{A step in the recursive binary search. Here, $x^*$ is the fixpoint for the original function, $y$ is the fixpoint for the slice restriction $\restr{f}{\ss}$ along the dashed gray line, and $z$ is an approximate fixpoint to the slice restriction.}
  \label{fig:approx_algo}
\end{figure}
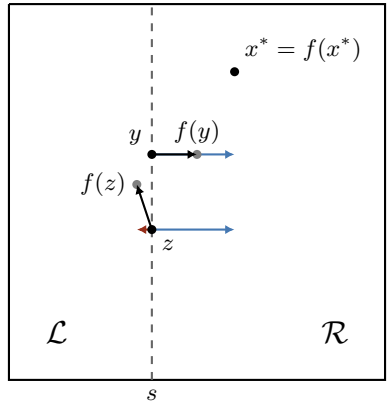

When looking for an approximate fixpoint, we'll have to choose a different precision $\eps_i$ for each level of the recursion so that either the point $x$ returned by the $i$th recursive call to our algorithm satisfies $\Abs{f(x)_i - x_i} > \eps_i$ and we can rely on it for pivoting in the binary search, or $\Abs{f(x)_i - x_i} \leq \eps_i$ and we can return $x$ as an approximate fixpoint to the recursive call one level up. Each different $\ell_p$ norm will require a different choice of $\Paren{\eps_i}_{i=1}^d$.

Using this idea we are able to obtain the following results:

\begin{theorem}
For a contraction map $f:[0,1]^d \to [0,1]^d$ with respect to the $\ell_1$ norm, there is an algorithm to compute a point $v\in [0,1]^d$ such that $\Norm{f(v) - v}_1 < \eps$ or report a violation of contraction in time $O(d^d\log(1/\eps))$. 
\end{theorem}

\begin{theorem}
  For a contraction map $f:[0,1]^d\to [0,1]^d$ under $\Norm{\cdot}_p$ for $2 \leq p < \infty$, there is an algorithm to compute a point $v\in [0,1]^d$ such that $\Norm{f(v) - v}_p < \eps$ or report a violation of contraction in time $O({p}^{d^2}\log^d(1/\eps)\log^d(pd))$.
\end{theorem}

The full details of the algorithms can be found in Appendix~\ref{subsec:approx_algo_details}.

\subsection{Details: finding a fixed-point of \LCM}
\label{subsec:exact_algo_details}
Suppose $f$ is piecewise-linear; it follows that the coordinates of the unique fixpoints of $\restr{f}{\ss}$ will be rational numbers with bounded denominators. Consider the values of $\kappa_i$'s computed in Section \ref{app:points2}. 
The analysis in the section tells us that if we consider an $i$-slice $\ss$ where $s_j$ is a rational with bit-length at most $\kappa_j$ for $j\in \Set{i+1,\dotsc,d}$, the unique fixed-point of $\restr{f}{\ss}$ will have coordinates with bit-lengths of at most $\kappa_1,\dotsc, \kappa_{d}$ Furthermore, $\kappa_1$ is the largest among all $\kappa_i$'s and is bounded by polynomial in the size of the circuit $C$ representing the given \LCM instance.

Let $\Slice_d(2^{\kappa})$ be the set of slices with fixed coordinates having denominators at most $2^{\kappa_i}$ in the $i$th coordinate:
\[ \Slice_d(2^{\kappa}) \iff \ss \in \Slice_d\ \text{and}\ s_i \in \Set{0/2^{\kappa_i},1/k_i,\dotsc,2^{\kappa_i}/2^{\kappa_i}},\ \forall i \in \fixed(\ss)\text{.} \] 
We will design an algorithm assuming an upper bound of $2^{\kappa_i}$ on the $i$th coordinate denominators of fixpoints for any slice restriction $\restr{f}{\ss}$ where $\ss \in \Slice_d(2^{\kappa})$. 

\begin{algorithm}[t]
	\caption{Algorithm for \LCM}\label{alg}
	
	\begin{algorithmic}[1]
          \State Input: A $k$-slice $\ss \in \Slice_d(2^{\kappa})$ for some $k \leq d$.
          \State Output: The unique fixpoint of $\ss$, i.e., a point $y$ 
		such that $\restr{f}{\ss}(y) = \restr{y}{\ss}$ and $y = \restr{y}{\ss}$.
          \Function{FindFP}{$\ss$}
              \State Let $k = \Card{\free(\ss)}$.
              \If{$k= 0$} \Return $\ss$ \EndIf
              \State Set $k \gets k - 1$. Set $\tt^{(\ell)} \gets \ss$, $\tt^{(h)} \gets \ss$.
              \State Set $t^{(\ell)}_k \gets 0$, $t^{(h)}_k \gets 1$. 
              \State Set $v^{(\ell)} \leftarrow \FindFP(\tt^{(\ell)})$, and  $v^{(h)} \leftarrow \FindFP(\tt^{(h)})$.
              \If{$f(v^{(\ell)})_k = t^{(\ell)}_k$} \Return $v^{(\ell)}$. \EndIf
              \If{$f(v^{(h)})_k = t^{(h)}_k$} \Return $v^{(h)}$. \EndIf
              \State Set $\tt \gets \ss$
                  \While{$t^{(h)}_k - t^{(\ell)}_k > \frac{1}{2^{(\kappa_k-1)}}$} 
                      \State Set $t_k \leftarrow \frac{(t^{(h)}_k + t^{(\ell)}_k)}{2}$.
                      \State Set $v \leftarrow\FindFP(\tt)$. \label{alg:vv}
                      \If{$f(v)_k = t_k$}  \Return $v$. \EndIf
                      \If{$f(v)_k > t_k$} Set $t^{(\ell)}_k \leftarrow t_k$
                      \Else\ Set $t^{(h)}_k \leftarrow t_k$. \EndIf
                  \EndWhile
              \State Set $t_k \leftarrow$ the unique number in $(t^{(\ell)}_k, t^{(h)}_k)$ with denominator at most $2^{\kappa_k}$. \label{alg1:candidate_fp}
              \State Set $v^* \leftarrow \FindFP(\tt)$.    
              \If{$f(v^*)_k - v^*_k = 0$} \Return $\FindFP(\tt)$.
              \Else\ \textbf{throw error}: ``The pair $\Paren{v^{(\ell)}, v^{(h)}}$ is a solution of type \solnref{CMV3}.'' \EndIf
              \EndFunction 
\end{algorithmic}
\end{algorithm}

\smallskip

\noindent \textbf{Analysis.}
Given a $k < d$ and a $k$-slice $\ss \in \Slice_d(2^{\kappa})$, we know from Lemma~\ref{lem:cm1} that the restricted function $\restr{f}{\ss}$ is also contracting. We will show that $\FindFP$, the function given in Algorithm~\ref{alg}, computes a fixpoint of $\restr{f}{\ss}$. 
Since the algorithm is recursive, we will prove its correctness by induction. The next lemma establishes the base case of induction and follows by design of the algorithm. It is equivalent to finding a fixpoint of a one dimensional function.

\begin{lemma}\label{lem:cm-ind1}
For any $1$-slice $\ss$, $\FindFP(\ss)$ returns the unique fixpoint of $\restr{f}{\ss}$ if it doesn't throw an error.
\end{lemma}

Now for the inductive step, assuming $\FindFP$ can compute a fixpoint of any $k$-slice restriction of a contraction map, we will show that it can compute one for any $(k+1)$-slice of the contraction map.

\begin{lemma}\label{lem:cm-ind2}
Fix any $k$-slice $\ss \in \Slice_d(2^{\kappa})$ and let $\tt = (\blank,\dotsc,\blank, t_k, s_{k+1}, \dotsc, s_d)$ for some $t_k \in \Set{0/2^{\kappa_k},1/2^{\kappa_k},\dotsc,2^{\kappa_k}/2^{\kappa_k}}$. If $\FindFP(\tt)$ returns the unique fixpoint of the restricted function $\restr{f}{\tt}$, then $\FindFP(\ss)$ returns the unique fixpoint of the function $\restr{f}{\ss}$ if it doesn't throw an error.
\end{lemma}
\begin{proof}
Since $f$ is a contraction map, so is $\restr{f}{\ss}$ due to Lemma~\ref{lem:cm1}. Let $v=\FindFP(\tt)$ be the value returned by the algorithm when input the slice $\tt$. Now, if $t_k=0$ then $f(v)_k \geq 0 = t_k = v_k$ and if $t_k=1$ then $f(v)_k \leq 1 = t_k = v_k$. If either is an equality then $v$ is the unique fixpoint of $\restr{f}{\ss}$ as well and the lemma follows. 

Otherwise, we know that the $k$th coordinate of the unique fixpoint of $\restr{f}{\ss}$, which we'll call $t^*_k$, is between $t^{(\ell)}_k=0$ and $t^{(h)}_k=1$.
Note that when $t_k$ is set to $t^*_k$, the vector $v = \FindFP(\tt)$ which is the unique fixpoint of $\restr{f}{\tt}$ will also be the unique fixpoint of $\restr{f}{\ss}$. Thus, it suffices to show that $t_k$ will eventually be set to $t^*_k$ during the execution of the algorithm. 

The while loop of Algorithm~\ref{alg} does a binary search between $t^{(h)}_k$ and $t^{(\ell)}_k$ to find $t^*_k$, while keeping track of the fixpoints of $\restr{f}{\tt}$. We first observe that after line~\ref{alg:vv} of each execution of the loop in Algorithm~\ref{alg}, $v$ is the unique fixpoint of $\restr{f}{\tt}$. Therefore, whenever $f(v)_k=t_k$ is satisfied, we will return the unique fixpoint of $\restr{f}{\ss}$ and the lemma follows. 

The binary search maintains the invariant that if we let $v^{(\ell)} = \FindFP(t^{(\ell)})$ and $v^{(h)} = \FindFP(t^{(h)})$ we have $f(v^{(\ell)})_k > v^{(\ell)}_k$, and $f(v^{(h)})_k < v^{(h)}_k$. By Lemma~\ref{lem:cm2}, this invariant ensures that $t^*_k$ satisfies
\[ t^{(\ell)}_k < t^*_k < t^{(h)}_k \] at all times. Therefore, at some point in the binary search either one of the endpoints is $t^*_k$ and we return the desired fixpoint or we end the binary search with $t^{(\ell)}_k$ and $t^{(h)}_k$ such that $(t^{(h)}_k-t^{(\ell)}_k)\leq 1/2^{\kappa_k-1}$ and $t^*_k \in [t^{(\ell)}_k, t^{(h)}_k]$. By the assumption we know that $t^*_k$ is a rational number with denominator at most $2^{\kappa_k}$. Since there can be at most one such number in $(t^{(\ell)}_k, t^{(h)}_k)$, $t^*_k$ can be uniquely identified. Let $\tt^*$ denote the slice $\tt$ after setting $t_k \leftarrow t^*_k$. The second to last line of Algorithm~\ref{alg} will then return the unique fixpoint of $\restr{f}{\tt^*}$, which will be the unique fixpoint of $\restr{f}{\ss}$ is a contraction map. The only way this can fail to happen is if the point returned by $\FindFP(\tt^*)$ is not actually a $k$-dimensional fixpoint, in which case the algorithm will throw an error.
\end{proof}

We now address the case where the algorithm returns an error:

\begin{lemma} \label{lem:cm_err}
If $\FindFP(\tt)$ returns an error for $k$-slice $\ss \in \Slice_d(2^{\kappa_1},2^{\kappa_2},\dotsc, 2^{\kappa_d})$, the pair of points $(v^{(\ell)},v^{(h)})$ indicated by the error witness that $f$ is not a contraction map. 
\end{lemma}

\begin{proof}
The algorithm only returns an error when $t^{h(\ell)}_k - t^{(\ell)}_k \leq \frac{1}{2^{\kappa_k-1}}$ and the point $v^*$ returned by $\FindFP(\tt^*)$ is not a $k$-dimensional fixpoint, where $\tt^*$ is the slice obtained by setting $t_k \leftarrow t^*$ in line~\ref{alg1:candidate_fp}. By induction we know that $v^*$ must be a $(k-1)$-fixpoint, since the recursive call to the algorithm didn't throw an error.  If $f$ is a contraction map, then the fact that $v^{(\ell)}$ and $v^{(h)}$ are $(k-1)$-dimensional fixpoints of $\restr{f}{\ss}$ with $\restr{f}{\ss}(v^{(\ell)})_k - v^{(\ell)}_k > 0$ and $\restr{f}{\ss}(v^{(h)})_k - v^{(h)}_k < 0$ together imply that the $k$th coordinate of the true fixpoint $\restr{f}{\ss}$ will lie in $(v^{(\ell)}_k, v^{(h)}_k)$ by Lemma~\ref{lem:cm2}. By Lemma~\ref{lem:points2}, we know that any $k$-dimensional fixpoint of $f$ has $k$th coordinate with bit-length at most $\kappa_k$. Thus, there is a unique value which the $k$th coordinate of the unique fixpoint of $\restr{f}{\ss}$ can have, namely $t^*_k$. But $v^*$ is not a fixpoint of $\restr{f}{\ss}$ and so we must conclude that $\restr{f}{\ss}$ is not a contraction map, which implies that $f$ is not a contraction map. Thus, the pair $(v^{(\ell)}, v^{(h)})$ together witness that $f$ is not a contraction map.
\end{proof}

Using Lemma~\ref{lem:cm_err} and applying induction using Lemma~\ref{lem:cm-ind1} as a base-case and Lemma~\ref{lem:cm-ind2} as an inductive step, the next theorem follows.

\begin{theorem}
$\FindFP(\blank,\blank,\dotsc,\blank)$ returns the unique fixpoint of $f$, or a pair of points proving that $f$ isn't a contraction map in time $O(L^d)$ where $L=\max_k L_k$ is polynomial in the size of the input instance. 
\end{theorem}

\subsection{Details: finding an approximate fixed-point of \CM}
\label{subsec:approx_algo_details}

We now proceed to prove the correctness of our algorithm. 

\begin{algorithm}[t]
	\caption{Algorithm for \CM, for a given $f$, $\eps$, $\ell_p$}\label{alg2}
	
	\begin{algorithmic}[1]
          \State Input: A $k$-slice $\ss \in \Slice_d$ for some $k\leq d$.
          \State Output: An $(\ss, \ell_p, k)$-approximate fixpoint of $\restr{f}{\ss}$.
          \Function{ApproxFindFP}{$\ss$}
              \State Let $k = \Card{\free(\ss)}$.
              \If{$k= d$} \Return $\ss$ \EndIf
              \State Set $k \gets k - 1$. Set $\tt^{(\ell)} \gets \ss$, $\tt^{(h)} \gets \ss$.
              \State Set $t^{(\ell)}_k \gets 0$, $t^{(h)}_k \gets 1$. 
              \State Set $v^{(\ell)} \leftarrow \ApproxFindFP(\tt^{(\ell)})$, and  $v^{(h)} \leftarrow \ApproxFindFP(\tt^{(h)})$.
              \If{$\Abs{\Delta_k\Paren{v^{(\ell)}}} \leq \eps_k(p, d)$} \Return $v^{(\ell)}$. \EndIf
              \If{$\Abs{\Delta_k\Paren{v^{(h)}}} \leq \eps_k(p, d)$} \Return $v^{(h)}$. \EndIf
              \State Set $\tt \gets \ss$
              \While{$t^{(h)}_k - t^{(\ell)}_k > \eps_k(p, d)$}
                  \State Set $t_k \leftarrow \frac{(t^{(h)}_k + t^{(\ell)}_k)}{2}$.
                  \State Set $v \leftarrow\ApproxFindFP(\tt)$. 
                  \If{$\Abs{\Delta_k(v)} \leq \eps_k(p, d)$}  \Return $v$. \EndIf
                  \If{$f(v)_k > t_k$} Set $t^{(\ell)}_k \leftarrow t_k$
                  \Else\ Set $t^{(h)}_k \leftarrow t_k$. \EndIf
              \EndWhile
              \State Set $t_k \leftarrow \frac{(t^{(h)}_k + t^{(\ell)}_k)}{2}$.
              \State Set $v^* \leftarrow \ApproxFindFP(\tt)$
              \If{$\Abs{\Delta_k(v^*)} > \eps_k(p, d)$}
                  \If{$\Delta_k(v^*) > \eps_k(p, d)$} \State\textbf{throw error}: ``The pair $\Paren{v^*, v^{(h)}}$ witnesses $f$ not being a contraction map.''
                  \Else \State\textbf{throw error}: ``The pair $\Paren{v^{(\ell)}, v^*}$ witnesses $f$ not being a contraction map.'' \EndIf
              \EndIf
              \State \Return $v^*$. \label{alg2:final_output}
              \EndFunction 
\end{algorithmic}
\end{algorithm}

\smallskip

\noindent \textbf{Analysis.}

We will show that for any contraction map $f$ with respect to an $\ell_p$ norm, and any $\eps > 0$, if Algorithm~\ref{alg2} doesn't throw an error, then it returns an approximate fixpoint of $f$, i.e. a point $v\in [0,1]^d$ such that $\Norm{f(v) - v} \leq \eps$. To do this, we'll show that for any $k < d$ and $k$-slice $\ss \in \Slice_d$ $\ApproxFindFP(\ss)$ will return a $(\ss, \ell_p, k)$-approximate fixpoint (when it doesn't throw an error). Since Algorithm~\ref{alg2} is recursive, our proof will be by induction. The next lemma establishes the base case of the induction and follows by design of the algorithm.

\begin{lemma}\label{lem:approx-cm-ind1}
For any $1$-slice $\ss$, $\ApproxFindFP(\ss)$ returns a $(\ss, \ell_p, 1)$-approximate fixpoint when it doesn't throw an error.
\end{lemma}

For the inductive step, we show that we can go from approximate fixpoints of $(k-1)$-slices to approximate fixpoints of $k$-slices.

\begin{lemma}\label{lem:approx-cm-ind2}
  Fix some $k$-slice $\ss \in \Slice_d$ and let $\tt = (\blank,\dotsc,\blank,t_k,s_{k+1},\dotsc,s_d)$ for some $t_k \in [0,1]$. If $\ApproxFindFP(\tt)$ returns a $(\tt, \ell_p, k-1)$-approximate fixpoint $v$, then $\ApproxFindFP(\ss)$ returns an $(\ss, \ell_p, k)$-approximate fixpoint when it doesn't throw an error. 
\end{lemma}
\begin{proof}
  We observe that $\restr{f}{\ss}$ is a contraction map by Lemma~\ref{lem:cm1}. We assume that $v=\ApproxFindFP(\tt)$ is a $(\tt, \ell_p, k-1)$-approximate fixpoint of $\restr{f}{\tt}$ for any value of $t_k\in [0,1]$. 

  We first observe that after the first recursive invocations of $\ApproxFindFP$, as $\ApproxFindFP(\tt^{(\ell)})$ and $\ApproxFindFP(\tt^{(h)})$, if $\Abs{f(v^{(h)})_k - v^{(h)}_k} \leq \eps_k(p,d)$ or $\Abs{f(v^{(\ell)})_k - v^{(\ell)}_k}\leq \eps_k(p,d)$, we return $v^{(h)}$ or $v^{(\ell)}$, respectively, so the output of $\ApproxFindFP(\ss)$ satisfies the requirements of the lemma.

  Moreover, in every subsequent call to $\ApproxFindFP(\tt)$, if the output $v$ satisfies $\Abs{f(v)_k - v_k} \leq \eps_k$, we return $v$, so we'll assume in what follows that the point $v^*$ returned by the algorithm is from line \ref{alg2:final_output}.

 We observe that $v^*$ was the output of a call to $\ApproxFindFP$ with a $(k-1)$-slice $\tt$, so it is a $(\tt, \ell_p, k-1)$-approximate fixpoint and in order to get to the final line, we must have $\Abs{\Delta_k(v^*)}\leq \eps_k(p, d)$. Thus, we have that $v^*$ is a $(\tt, \ell_p, k)$-approximate fixpoint. 
\end{proof}

Applying induction using Lemma~\ref{lem:approx-cm-ind1} as a base-case and Lemma~\ref{lem:approx-cm-ind2} as an inductive step, we obtain the following lemma:

\begin{lemma} For a contraction map with respect to an $\ell_p$ norm for $p < \infty$, and the trivial slice $\ss = (\blank, \blank, \dotsc, \blank)$, $\ApproxFindFP(\ss)$ returns a $(\ss, \ell_p, d)$-approximate fixpoint when it doesn't throw an error. 
\end{lemma}

\begin{lemma}\label{lem:approx-cm-err}
  If Algorithm~\ref{alg2} throws an error, the pair $(x, y)$ indicated by the error witnesses $f$ not being a contraction map. 
\end{lemma}

\begin{proof}
  Assume that the error was thrown in a call to $\ApproxFindFP(\ss)$ where $\ss$ is a $k$-slice for some $k \leq d$. Then by inspection we can see that $x$ and $y$ satisfy $x_k < y_k$ and $y_k - x_k < \eps_k(p, d) / 2$. Moreover, $x$ and $y$ were returned by calls to $\ApproxFindFP$ with $(k-1)$-slices that we'll denote $\tt^{(x)}$ and $\tt^{(y)}$, respectively. It follows from the inductive argument above that $x$ and $y$ are $(\tt^{(x)}, \ell_p, k-1)$ and $(\tt^{(y)}, \ell_p, k-1)$-approximate fixpoints, respectively. Furthermore, to get to the line in which the error is thrown, we must have $\Delta_k(x) \geq \eps_k(p,d)$ and $\Delta_k(y) \leq -\eps_k(p,d)$. By Lemma~\ref{lem:approx_cm-lp_final_violations}, we immediately obtain the desired conclusion.
\end{proof}

Now applying Lemma~\ref{lem:slice_approx_to_approx} and analyzing the runtime of our algorithm, we obtain our final results.

\begin{theorem}
For a contraction map with respect to the $\ell_1$ norm, $\ApproxFindFP(\blank,\blank,\dotsc,\blank)$ returns a point $v\in [0,1]^d$ such that $\Norm{f(v) - v}_1 < \eps$ or reports a violation of contraction in time $O(d^d\log(1/\eps))$.
\end{theorem}

\begin{proof}
  In the worst case, the $i$th recursive call to $\ApproxFindFP$ will terminate with $t^{(h)}_i - t^{(\ell)}_i < \eps_i = \eps/4^i$ so will require at most $O(\log(1/\eps)i)$ iterations. The total runtime will thus be bounded by $O(d^d\log^d(1/\eps))$. 
\end{proof}

\begin{theorem}
  For a contraction map under $\Norm{\cdot}_p$ for $2 \leq p < \infty$, $\ApproxFindFP(\blank,\blank,\dotsc,\blank)$ returns a point $v\in [0,1]^d$ such that $\Norm{f(v) - v}_p < \eps$ or reports a violation of contraction in time $O({p}^{d^2}\log^d(1/\eps)\log^d(dp))$.
\end{theorem}

\begin{proof}
  In the worst case, the $i$th recursive call to $\ApproxFindFP$ will terminate with $t^{(h)}_i - t^{(\ell)}_i < \eps_i = \eps^{p^i}\Paren{dp}^{-2\sum_{j=0}^i p^j}$ so will require $O(p^i\log(1/\eps)\log(dp))$ iterations. The total runtime will thus be bounded by $O(p^{d^2}\log^d(1/\eps)\log^d(dp))$. 
\end{proof}

\end{document}